%% file: main.tex
\documentclass[11pt]{article}

\usepackage{amsthm}
\usepackage{graphicx} 
\usepackage{array} 

\usepackage{amsmath, amssymb, amsfonts, verbatim}
\usepackage{hyphenat,epsfig,subcaption,multirow}
\usepackage{nicefrac}
\usepackage{paralist}

\usepackage[font=footnotesize,labelfont=bf]{caption}

\usepackage[usenames,dvipsnames]{xcolor}
\usepackage[ruled]{algorithm2e}

\DeclareFontFamily{U}{mathx}{\hyphenchar\font45}
\DeclareFontShape{U}{mathx}{m}{n}{
      <5> <6> <7> <8> <9> <10>
      <10.95> <12> <14.4> <17.28> <20.74> <24.88>
      mathx10
      }{}
\DeclareSymbolFont{mathx}{U}{mathx}{m}{n}
\DeclareMathSymbol{\bigtimes}{1}{mathx}{"91}

\usepackage{tcolorbox}
\tcbuselibrary{skins,breakable}
\tcbset{enhanced jigsaw}

\usepackage[normalem]{ulem}
\usepackage[compact]{titlesec}

\definecolor{DarkRed}{rgb}{0.5,0.1,0.1}
\definecolor{DarkBlue}{rgb}{0.1,0.1,0.5}

\usepackage{nameref}
\definecolor{ForestGreen}{rgb}{0.1333,0.5451,0.1333}
\definecolor{Red}{rgb}{0.9,0,0}
\usepackage[linktocpage=true,
	pagebackref=true,colorlinks,
	linkcolor=DarkRed,citecolor=ForestGreen,
	bookmarks,bookmarksopen,bookmarksnumbered]
	{hyperref}
\usepackage[noabbrev,nameinlink]{cleveref}
\crefname{property}{property}{Property}
\creflabelformat{property}{(#1)#2#3}
\crefname{equation}{eq}{Eq}
\creflabelformat{equation}{(#1)#2#3}

\usepackage{bm}
\usepackage{url}
\usepackage{xspace}
\usepackage[mathscr]{euscript}

\usepackage{tikz}
\usetikzlibrary{arrows}
\usetikzlibrary{arrows.meta}
\usetikzlibrary{shapes}
\usetikzlibrary{backgrounds}
\usetikzlibrary{positioning}
\usetikzlibrary{decorations.markings}
\usetikzlibrary{patterns}
\usetikzlibrary{calc}
\usetikzlibrary{fit}
\usetikzlibrary{snakes}
\tikzset{vertex/.style={circle, black, fill=Yellow, line width=1pt, draw, minimum width=8pt, minimum height=8pt, inner sep=0pt}}
	
\usepackage{mdframed}

\usepackage[noend]{algpseudocode}
\makeatletter
\def\BState{\State\hskip-\ALG@thistlm}
\makeatother

\usepackage{cite}
\usepackage{enumitem}

\usepackage[margin=1in]{geometry}

\newtheorem{theorem}{Theorem}
\newtheorem{lemma}{Lemma}[section]
\newtheorem{proposition}[lemma]{Proposition}
\newtheorem{corollary}[lemma]{Corollary}
\newtheorem{claim}[lemma]{Claim}
\newtheorem{fact}[lemma]{Fact}

\newtheorem{definition}[lemma]{Definition}

\newtheorem*{claim*}{Claim}
\newtheorem*{theorem*}{Theorem}
\newtheorem*{proposition*}{Proposition}
\newtheorem*{lemma*}{Lemma}
\newtheorem*{problem*}{Problem}

\crefname{lemma}{Lemma}{Lemmas}
\crefname{claim}{Claim}{Claims}

\newtheorem*{mdresult}{Main Result}
\newenvironment{result}{\begin{mdframed}[backgroundcolor=lightgray!40,topline=false,rightline=false,leftline=false,bottomline=false,innertopmargin=2pt, innerleftmargin=10pt]\begin{mdresult}}{\end{mdresult}\end{mdframed}}

\newtheorem*{remark*}{Remark}
\newtheorem{observation}[lemma]{Observation}

\newtheoremstyle{restate}{}{}{\itshape}{}{\bfseries}{~(restated).}{.5em}{\thmnote{#3}}
\theoremstyle{restate}

\theoremstyle{definition}
\newtheorem{mdalg}{algorithm}

\newtheorem{mddist}{Distribution}

\allowdisplaybreaks

\renewcommand{\qed}{\nobreak \ifvmode \relax \else
      \ifdim\lastskip<1.5em \hskip-\lastskip
      \hskip1.5em plus0em minus0.5em \fi \nobreak
      \vrule height0.75em width0.5em depth0.25em\fi}

\newcommand{\Qed}[1]{\rlap{\qed$_{\textnormal{~~~\Cref{#1}}}$}}

\newcommand{\logstar}[1]{\ensuremath{\log^{*}\!{#1}}}

\input{macros}

\title{A Two-Pass Lower Bound for Semi-Streaming  Maximum Matching}
\author{Sepehr Assadi\footnote{(\texttt{sepehr.assadi@rutgers.edu}) Department of Computer Science, Rutgers University.  Research supported in part by a NSF CAREER Grant CCF-2047061, and a gift from Google Research.}  
}

\date{}

\begin{document}
\maketitle

\pagenumbering{roman}


\input{abstract}

\clearpage

\setcounter{tocdepth}{3}
\tableofcontents

\clearpage

\pagenumbering{arabic}
\setcounter{page}{1}

\input{intro}

\input{prelim}

\input{game}

\input{lower}

\input{xor-lemma}

\subsection*{Acknowledgement} 

The author is grateful to Soheil Behnezhad, Michael Kapralov, Raghuvansh Saxena, and Huacheng Yu for illuminating conversations, and to Christian Konrad for helpful discussions on his recent work in~\cite{KonradN21}. 
The author is also indebted to his collaborators Ran Raz in~\cite{AssadiR20}, Gillat Kol, Raghuvansh Saxena, and Huacheng Yu in~\cite{AssadiKSY20}, Vishvajeet N. in~\cite{AssadiN21}, and Soheil Behnezhad in~\cite{AssadiB21} for their previous collaborations that formed various building blocks and inspirations for this work. 

\bibliographystyle{alpha}
\bibliography{new,new2,general}

\clearpage
\appendix
\part*{Appendix}
\input{appendix-info}

\input{appendix-fourier}

\end{document}

%% file: macros.tex
\newcommand{\Leq}[1]{\ensuremath{\underset{\textnormal{#1}}\leq}}

\newcommand{\tvd}[2]{\ensuremath{\norm{#1 - #2}_{\mathrm{tvd}}}}

\newcommand{\eps}{\ensuremath{\varepsilon}}

\newcommand{\Bracket}[1]{\Big[#1\Big]}
\newcommand{\bracket}[1]{\left[#1\right]}
\newcommand{\paren}[1]{\ensuremath{\left(#1\right)}\xspace}
\newcommand{\card}[1]{\left\vert{#1}\right\vert}

\newcommand{\IR}{\ensuremath{\mathbb{R}}}

\newcommand{\norm}[1]{\ensuremath{\|#1\|}}

\newcommand{\set}[1]{\ensuremath{\left\{ #1 \right\}}}
\newcommand{\poly}{\mbox{\rm poly}}

\newcommand{\ALG}{\ensuremath{\mbox{\sc alg}}\xspace}

\DeclareMathOperator*{\Exp}{\ensuremath{{\mathbb{E}}}}
\DeclareMathOperator*{\Prob}{\ensuremath{\textnormal{Pr}}}
\renewcommand{\Pr}{\Prob}

\newcommand{\Ex}{\Exp}

\newenvironment{tbox}{\begin{tcolorbox}[
		enlarge top by=5pt,
		enlarge bottom by=5pt,
		 breakable,
		 boxsep=2pt,
                  left=5pt,
                  right=7pt,
                  top=10pt,
                  arc=0pt,
                  boxrule=1pt,toprule=1pt,
                  colback=white
                  ]
	}
{\end{tcolorbox}}

\newcommand{\event}{\ensuremath{\mathcal{E}}}

\newcommand{\rv}[1]{\ensuremath{{\mathsf{#1}}}\xspace}
\newcommand{\rA}{\rv{A}}
\newcommand{\rB}{\rv{B}}
\newcommand{\rC}{\rv{C}}
\newcommand{\rD}{\rv{D}}

\newcommand{\rY}{\rv{Y}}
\newcommand{\supp}[1]{\ensuremath{\textnormal{\text{supp}}(#1)}}
\newcommand{\distribution}[1]{\ensuremath{\textnormal{dist}(#1)}\xspace}

\newcommand{\kl}[2]{\ensuremath{\mathbb{D}(#1~||~#2)}}
\newcommand{\II}{\ensuremath{\mathbb{I}}}
\newcommand{\HH}{\ensuremath{\mathbb{H}}}
\newcommand{\mi}[2]{\ensuremath{\def\mione{#1}\def\mitwo{#2}\mireal}}
\newcommand{\mireal}[1][]{
  \ifx\relax#1\relax%
    \II(\mione \,; \mitwo)%
  \else%
    \II(\mione \,; \mitwo\mid #1)%
  \fi
}
\newcommand{\en}[1]{\ensuremath{\HH(#1)}}

\newcommand{\itfacts}[1]{\Cref{fact:it-facts}-(\ref{part:#1})\xspace}

\newcommand{\dist}{\ensuremath{\mathcal{D}}}

\newcommand{\prot}{\pi}
\newcommand{\Prot}{\Pi}

\newcommand{\cost}{\ensuremath{\mathsf{C}}}

\newcommand{\bias}{\ensuremath{\mathrm{bias}}}

\newcommand{\rX}{\rv{X}}
\newcommand{\rZ}{\rv{Z}}
\newcommand{\rM}{\rv{M}}

\newcommand{\rs}{Ruzsa-Szemer\'edi\xspace} 
\newcommand{\RS}{\textnormal{\textsf{RS}}}

\newcommand{\GRS}{\ensuremath{G^{\mathsf{RS}}}\xspace}
\newcommand{\ERS}{\ensuremath{E^{\mathsf{RS}}}\xspace}

\newcommand{\rH}{\rv{H}}
\newcommand{\rJ}{\rv{J}}

\newcommand{\uniform}{{U}}

\newcommand{\rPhi}{\rv{\Phi}}

\newcommand{\jstar}{j^{\star}}

\newcommand{\ru}{\rv{u}}

\newcommand{\bH}{\bar{H}}

\newcommand{\Xor}{\mathcal{X}}

\newcommand{\hf}{\widehat{f}}

\newcommand{\encodedrs}{\textnormal{\texttt{Encoded-RS}}\xspace}
\newcommand{\augmentedge}{\textnormal{\texttt{Aug-Edges}}\xspace}
\newcommand{\augmentpath}{\textnormal{\texttt{Aug-Paths}}\xspace}
\newcommand{\augmentgraph}{\textnormal{\texttt{Aug-Graph}}\xspace}

\newcommand{\UU}{\mathcal{U}}

\newcommand{\start}[1]{\ensuremath{\textnormal{start}(#1)}\xspace}
\newcommand{\eend}[1]{\ensuremath{\textnormal{end}(#1)}\xspace}

\newcommand{\aug}[1]{\ensuremath{\textnormal{aug}(#1)}\xspace}
\newcommand{\baug}[1]{\ensuremath{\overline{\textnormal{aug}}(#1)}\xspace}
\newcommand{\rep}[1]{\ensuremath{\textnormal{rep}(#1)}\xspace}

\newcommand{\game}{\ensuremath{\textnormal{\textsf{HiddenMatching}}}\xspace}

\newcommand{\distaug}{\ensuremath{\dist_{\textnormal{aug-graph}}}\xspace}

\renewcommand{\cost}[1]{\ensuremath{\textnormal{cost}(#1)}\xspace}
\newcommand{\val}[1]{\ensuremath{\textnormal{value}(#1)}\xspace}
\newcommand{\out}[1]{\ensuremath{\textnormal{out}(#1)}\xspace}

\renewcommand{\ALG}{\ensuremath{\mathcal{A}}\xspace}

\newcommand{\tG}{\ensuremath{\widehat{G}}}

\newcommand{\rProt}{\rv{\Prot}}

\newcommand{\rG}{\rv{G}}
\newcommand{\MRS}{\ensuremath{M^{\mathsf{RS}}}\xspace}

%% file: abstract.tex

\begin{abstract}

\smallskip

	We prove a  lower bound on the space complexity of \textbf{two-pass semi-streaming} algorithms that approximate the maximum matching problem. The lower bound is parameterized by the density of \emph{\rs graphs}: 
	
	\smallskip
	\begin{itemize}
	\item Any two-pass
	semi-streaming algorithm for maximum matching has approximation ratio at least $\paren{1- \Omega(\frac{\log{\RS(n)}}{\log{n}})}$, where $\RS(n)$ denotes the maximum number of induced matchings of size $\Theta(n)$ in any $n$-vertex graph, i.e., 
	the largest density of a \rs graph. 
	\end{itemize}
	
	\smallskip
	\noindent
	Currently, it is known that $n^{\Omega(1/\!\log\log{n})} \leq \RS(n) \leq \frac{n}{2^{O(\logstar{\!(n)})}}$ and closing this (large) gap between upper and lower bounds has remained a notoriously difficult problem in combinatorics. 
	
	\medskip
	Under the plausible hypothesis that 
	$\RS(n) = n^{\Omega(1)}$, our lower bound is the first to rule out  \textbf{small-constant approximation} two-pass semi-streaming algorithms for the maximum matching problem, making progress on a longstanding open question in the graph streaming literature. 
	
\end{abstract}

%% file: intro.tex

\section{Introduction}\label{sec:intro}

The semi-streaming model of computation, introduced in~\cite{FeigenbaumKMSZ05}, has been at the forefront of research on processing massive graphs. In this model, the edges of an $n$-vertex graph $G=(V,E)$ 
are arriving one by one in a stream; the algorithm can only make one or a small number of passes over the stream and use a limited space of $O(n \cdot \poly\!\log{(n)})$ to solve a given problem on the input graph, say find 
a spanning tree of $G$. In this paper, we focus on the \emph{maximum matching} problem in the semi-streaming model. 

The maximum matching problem has been a cornerstone of research on semi-streaming algorithms and 
 been studied from numerous angles: single-pass algorithms~\cite{FeigenbaumKMSZ05,GoelKK12,Kapralov13,Kapralov21}, two-pass algorithms~\cite{KonradMM12,EsfandiariHM16,KaleT17,Konrad18}, $(1-\eps)$-approximation algorithms~\cite{McGregor05,AhnG11,EggertKMS12,AhnG18,Tirodkar18,GamlathKMS19,AssadiLT21,FischerMU21,AssadiJJST21}, random-order streams~\cite{KonradMM12,Konrad18,AssadiBBMS19,GamlathKMS19,FarhadiHMRR20,Bernstein20,AssadiB21}, 
dynamic streams~\cite{Konrad15,ChitnisCHM15,AssadiKLY16,ChitnisCEHMMV16, AssadiKL17,DarkK20}, weighted matchings~\cite{FeigenbaumKMSZ05,CrouchS14,PazS17,BernsteinDL21}, submodular matchings~\cite{ChakrabartiK14,ChekuriGQ15,LevinW21}, estimating  
size~\cite{KapralovKS14,EsfandiariHLMO15,BuryS15,McGregorV16,CormodeJMM17,McGregorV18,AssadiKL17,KapralovMNT20,AssadiKSY20,AssadiN21}, and exact algorithms and lower bounds~\cite{FeigenbaumKMSZ05,GuruswamiO13,AssadiR20,LiuSZ20,ChenKPSSY21,AssadiJJST21}, among others (this  is by no means a comprehensive summary of prior results). 

In this paper, we focus on \textbf{proving  lower bounds for constant-factor approximation} of the maximum matching problem via semi-streaming algorithms. 
A brief note on the history of this problem is in order. Alongside the introduction of semi-streaming model in~\cite{FeigenbaumKMSZ05}, the authors posed the problem of understanding approximation ratio of multi-pass algorithms for matchings. 
On the upper bound front, numerous results have since been shown for this problem, see, e.g.~\cite{McGregor05,AhnG11,KonradMM12,EggertKMS12,KaleT17,AhnG18,Konrad18,AssadiLT21,FischerMU21,Kapralov21,AssadiJJST21} and references therein. 
On the lower bound front however, the first result appeared almost a decade later in~\cite{GoelKK12} who showed that single-pass semi-streaming algorithms cannot achieve a better than
$(\nicefrac23)$-approximation; this ratio was soon improved to $(1-\nicefrac{1}{e})$-approximation by~\cite{Kapralov13} and very recently to $(\frac{1}{1+\ln{2}})$ in~\cite{Kapralov21}. 
Yet, almost another decade since~\cite{GoelKK12}, we still lack \emph{any} lower bounds for (constant-factor) approximation of the matching problem even in two passes!\footnote{We note that lower bounds for computing \emph{exact} matching up to (almost) $\log{n}$ passes 
are proven in~\cite{GuruswamiO13}; see also~\cite{AssadiR20,ChenKPSSY21}. These lower bounds however at best can only rule out $(1-\frac{1}{n^{o(1)}})$-approximation algorithms even in a 
single pass.} 

\subsection{Our Contribution}\label{sec:results}

We prove the first lower bound on the space complexity of \textbf{two-pass} semi-streaming algorithms that \textbf{approximate} the maximum matching problem. Our lower bound is parameterized by the density of \emph{\rs (RS) graphs} -- these are graphs 
whose edges can be partitioned into \emph{induced} matchings of size $\Theta(n)$ (see~\Cref{sec:rs}). We prove the following result: 
	
\begin{result}[\Cref{cor:stream-lower-RS}]\label{res:main}
	 Any two-pass semi-streaming algorithm for maximum matching (even on bipartite graphs) has approximation ratio at least $1 - \Omega(\frac{\log{ \RS(n)}}{\log{ n}})$, where $\RS(n)$ denotes the maximum number of disjoint induced matchings of size $\Theta(n)$ in any $n$-vertex 
	 graph. 
\end{result}
\noindent
Let us  put this result in some context. 
	\smallskip
	\noindent
	Currently, it is known that 
	\[
	n^{\Omega(1/\!\log\log{n})} \Leq{\cite{FischerLNRRS02}} \RS(n) \Leq{\cite{FoxHS15}} \frac{n}{2^{O(\logstar{\!(n)})}},
	\]
	 and closing this (large) gap between upper and lower bounds has remained a notoriously difficult problem in combinatorics~\cite{FoxHS15} (see also~\cite{Gowers01,ConlonF13}). 
	 With this in mind, we can think of our main result in one of the following two ways: 
	 \begin{itemize}
	 	\item \emph{Conditional lower bound}: Under the plausible hypothesis that $\RS(n)$ can be $n^{\beta}$ for some \emph{constant} $\beta \in (0,1)$, our result would rule out certain {small-constant factor} approximation of 
		maximum matching in two passes of the semi-streaming model; for instance, {assuming} $\RS(n) = n^{1-o(1)}$ (close to the current upper bounds), our lower bound states that 
		no two-pass semi-streaming algorithm can achieve an approximation ratio of $0.98$ for the maximum matching problem (see~\Cref{cor:stream-lower-RS} for the details). 
		
		For comparison, 
		the best known approximation ratio of two-pass semi-streaming algorithms is the $(2-\sqrt{2}) \approx 0.585$-approximation of~\cite{Konrad18} (see also~\cite{EsfandiariHM16}) for bipartite graphs and~$0.53$-approximation of~\cite{KaleT17} for non-bipartite ones. 
		
		\item \emph{Barrier result}: alternatively, our result can be interpreted that any sufficiently small constant factor approximation to matching in two-passes of semi-streaming model -- in particular, a $(1-\eps)$-approximation algorithm -- needs to (at the 
		very least) improve the upper bound on $\RS{(n)}$ from current bounds all the way to $n^{o(1)}$; this puts such a semi-streaming algorithm (seemingly) beyond the reach of current techniques. 
		
		For comparison, 
		current $(1-\eps)$-approximation semi-streaming algorithms require $O(\eps^{-1} \cdot \log{n})$ passes for bipartite graphs~\cite{AssadiJJST21} and non-bipartite ones~\cite{AhnG18},
		or $O(\eps^{-2})$ passes~\cite{AhnG11,AssadiLT21} for bipartite and $O(\poly(\eps^{-1}))$ for non-bipartite ones~\cite{FischerMU21}. 
	 \end{itemize}
	 
Finally, we shall note that, starting from~\cite{GoelKK12}, all previous single- and multi-pass lower bounds for semi-streaming matching problem in~\cite{GoelKK12,Kapralov13,AssadiKL17,AssadiR20,Kapralov21,ChenKPSSY21} were based on RS graphs -- the only exception is  the lower bound result of~\cite{GuruswamiO13} that only holds for exact algorithms (and is improved upon by~\cite{AssadiR20,ChenKPSSY21}). 
Nevertheless, for previous multi-pass lower bounds, even if one assumes $\RS{(n)}$ to be as large as the current best upper bounds, the best approximation ratio ruled out is still $(1-\frac{1}{\text{polylog}{(n)}})$ proven by~\cite{ChenKPSSY21} (which would hold for  (almost) $\log{n}$ passes under such an assumption on $\RS{(n)}$). 

\paragraph{Why \underline{two-pass} algorithms?}Traditionally, two-pass semi-streaming algorithms have been studied extensively 
as a way of breaking the lower bounds or  barriers for single-pass algorithms. For instance,~\cite{KonradMM12,EsfandiariHM16,KaleT17,Konrad18} developed two-pass algorithms for matching with approximation ratio
that breaks the notorious ``$(\nicefrac12)$-approximation barrier'' for current single-pass algorithms. Going beyond the matching problem, it is now a 
established phenomenon that two-pass algorithms can be surprisingly more powerful than single-pass ones; for instance,~\cite{AssadiD21} (building on~\cite{RubinsteinSW18,GhaffariNT20}), gave an $O(n)$ space algorithm for finding 
an exact minimum cut in two passes, while it is known that single-pass algorithms require $\Omega(n^2)$ space for this problem~\cite{Zelke11}; similar separations are also known for the triangle counting problem~\cite{BulteauFKP16,CormodeJ17}, 
among others.  

More recently, there has been a growing interest in proving lower bounds tailored specifically to two-pass streaming algorithms~\cite{AssadiR20,ChenKPS0Y21} (see also~\cite{GargRT19} for an example beyond graph streams). This line of work is 
motivated by both further understanding of two-pass algorithms as the ``second best option'' after single-pass algorithms, as well as a stepping stone for proving stronger multi-pass lower bounds; for instance, many of the ideas  
developed in the two-pass lower bound of~\cite{AssadiR20} (for reachability and exact matching) were used subsequently in the work of~\cite{ChenKPSSY21} that improved the lower bound to $\Omega(\sqrt{\log{n}})$-pass algorithms. Indeed, 
there are several technical difficulties in proving multi-pass lower bounds compared to single-pass ones which are already manifested  when allowing two passes over the input; we elaborate on these challenges when going over our techniques in the subsequent section.

\subsection{Our Techniques and Comparison with Prior Work}\label{sec:techniques} 

Our paper builds on and extend several lines of work on proving streaming lower bounds for single- and multi-pass algorithms: $(i)$ the single-pass RS-graph based lower bound approaches of~\cite{GoelKK12,Kapralov13,Kapralov21}, 
$(ii)$ the two-pass lower bound framework of~\cite{AssadiR20}, $(iii)$ the ``XOR gadgets'' approaches of~\cite{AssadiB21,ChenKPS0Y21}, and finally $(iv)$ streaming 
``XOR Lemmas'' for proving lower bounds for XOR gadgets~\cite{AssadiN21,ChenKPS0Y21}. 
We now elaborate on each of these. 

\paragraph{$(i)$ Single-pass RS-graph based lower bound approaches of~\cite{GoelKK12,Kapralov13,Kapralov21}.} The idea behind the single-pass lower bound of~\cite{GoelKK12} is as follows. The first part of the stream consists of an RS graph with induced matchings of size $\Theta(n)$ (known a-priori) whose $o(1)$-fraction of its edges have been dropped randomly to increase its entropy to almost 
$\Omega(n \cdot \RS{(n)})$. The second part of the stream is created by sampling a random induced matching in the RS graph, and presenting a perfect matching from a new set of vertices to vertices of the RS graph \emph{not} participating in this induced matching. See~\Cref{fig:1-way} below.

\begin{figure}[h!]
\centering
\input{figs/fig-1-way}
\caption{An illustration of approach of~\cite{GoelKK12}. The RS graph, whose edges are dropped w.p. $o(1)$, appears first. The figure only draws edges of second part (solid blue) and the special induced matching of first part (dashed red).}
\label{fig:1-way}
\end{figure}
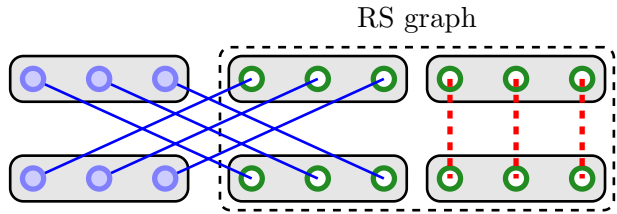

This graph has a near perfect matching but all of its ``large enough'' matchings require using the edges of the special induced matching of the RS graph. At the same time, since the algorithm was oblivious to the identity of this special induced matching in the first part, assuming its memory was $o(n \cdot \RS{(n)})$, it would reduce the entropy of edges of this induced matching by $o(n)$. This only allows the algorithm to output $o(1)$ fraction of edges of the special induced matching at the end without the risk 
of outputting an edge which is dropped from the graph. The follow up work in~\cite{Kapralov13,Kapralov21} then extend this idea by considering multiple parts in the stream and keeping the identity of a large induced matching hidden until the end. 


Nevertheless, it can be seen that this lower bound is inherently tailored to single-pass algorithms: a two-pass algorithm would \emph{reveal} the identity of the special induced matching in the first pass and so in the next pass, the algorithm can simply 
store only these edges of the RS graph in $O(n)$ space.  This is the first challenge we need to overcome in our work. 

\paragraph{$(ii)$ Two-pass lower bound framework of~\cite{AssadiR20}.} The work of~\cite{AssadiR20} developed a framework for proving two-pass lower bounds for several problems including \emph{exact} maximum matching. This framework also used RS graphs but \emph{in an entirely different way}, in particular, for ``hiding'' the information revealed to the second pass of the algorithm. In~\cite{AssadiR20}, the input graph consists of a random bipartite graph and two ``gadget RS graphs'' that each \emph{choose a single vertex} from this random graph, with the following property: the input has a perfect matching iff there is an edge between the chosen vertices. See~\Cref{fig:2-pass} for an illustration.

\begin{figure}[h!]
\centering
\input{figs/fig-2-pass}
\caption{An illustration of the framework of~\cite{AssadiR20}. The first part of the stream consists of a random graph, while the second part identify two vertices of this random graph as special using RS gadgets.}
\label{fig:2-pass}
\end{figure}
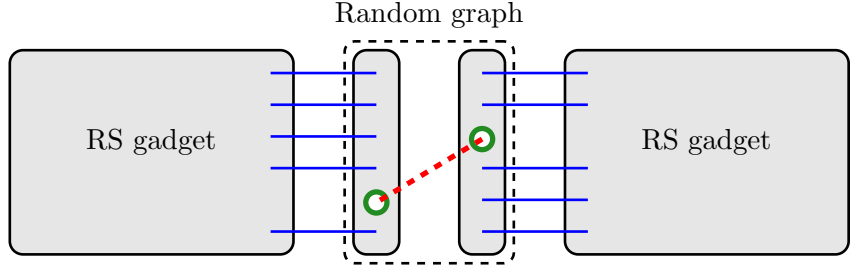

Beyond going into exact details, we mention that the RS gadget has the following property that a single-pass streaming algorithm is not able to identify the special vertex chosen for the gadget. As such, even after the first pass of the stream, the identity of which edge of the random graph is important to ``remember'' is not known to the algorithm, thus the algorithm needs to remember essentially all edges of the random graph in the second pass in order to solve the problem. 

This lower bound is specifically tailored to the perfect matching problem: the RS gadget can only ``hide'' a single vertex, which is not suitable for lower bounds for approximation algorithms. Thus, the second challenge is to work with RS gadgets 
that allow for hiding more than vertices. 

\paragraph{$(iii)$ the ``XOR gadgets'' approaches of~\cite{AssadiB21,ChenKPS0Y21}.} The work of~\cite{AssadiB21} gave a lower bound of $(1-\Theta(\frac{1}{\log{n}}))$-approximation for semi-streaming algorithms of
the matching problem in \emph{random-order} streams. The idea of the lower bound is to follow the approach of~\cite{GoelKK12} described in part $(i)$ above, but \emph{hide} the identity of the induced matching (in a random-order stream, revealing the 
first $O(n\log{n})$ edges of lower bound of~\cite{GoelKK12} reveals which of the induced matchings in the RS graph is special). This is done by replacing each of the edges of the perfect matching to vertices not in the special induced matching, by a path of 
length $\Theta(\log{n})$ that has an ``ON-OFF switch'': ``ON'' means we should leave the last vertex unmatched, and ``OFF'' means we should match it \emph{inside} the gadget. These paths are then put together in a way that 
only vertices of a random induced matching of the RS graph are ON and other vertices are OFF. See~\Cref{fig:part-3} below. 

\begin{figure}[h!]
\centering
\input{figs/fig-part-3}
\caption{An illustration of the lower bound~\cite{AssadiB21}. The paths with their switch ON choose a special matching of the middle RS graph, whose edges should be used in every large matching of the input graph.}
\label{fig:part-3}
\end{figure}
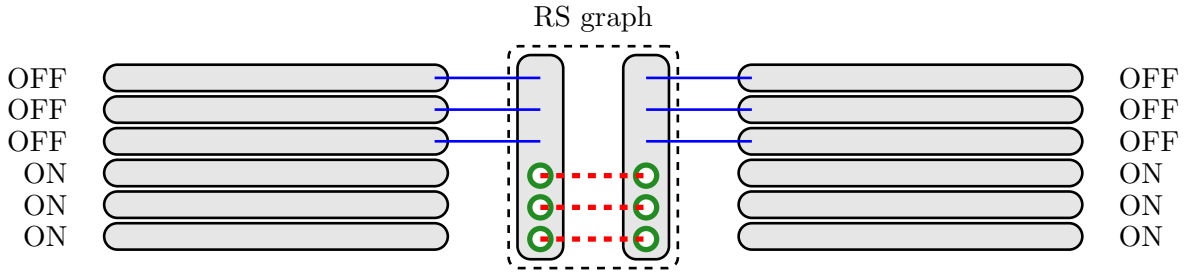

These paths, constructed using properties of the XOR function, have the property that even if one knows all but one edge of the path, it is still not clear whether the path is ON or OFF. Thus, in a random-order stream, with high probability, 
one edge from each of these paths (of length $\Theta(\log{n})$) are missing by the mid-point of the stream, forcing the algorithm to have to remember almost all of the edges of the RS graph visited so far. 


Concurrently to~\cite{AssadiB21},~\cite{ChenKPS0Y21} used a similar approach of using some type of XOR gadgets, combined with the framework of~\cite{AssadiR20} to hide more than one vertices of the graph from a single-pass algorithm. By an intricate 
combination of creating XOR of RS graphs, the authors of~\cite{ChenKPS0Y21} create a graph that have the following property: there is a set of $S$ of size $n^{1-o(1)}$ vertices with switches as described earlier, such that any $o(\log{n})$-pass semi-streaming algorithm cannot determine whether or not any of these vertices is ON or OFF. This allows the author to extend the lower bound of~\cite{AssadiR20} to 
$o(\log{n})$-pass algorithms that can approximate the matching to within a $(1-\frac{1}{n^{o(1)}})$-approximation. 

We note that while~\cite{ChenKPS0Y21} stops at getting a $(1-\frac{1}{n^{o(1)}})$-approximation, there is a natural way of combining their work and the approach of~\cite{GoelKK12} in part $(i)$ as in~\cite{AssadiB21}, so that 
the hidden set $S$ can determine which induced matching of the middle RS graph is special. See~\Cref{fig:part-4}. 

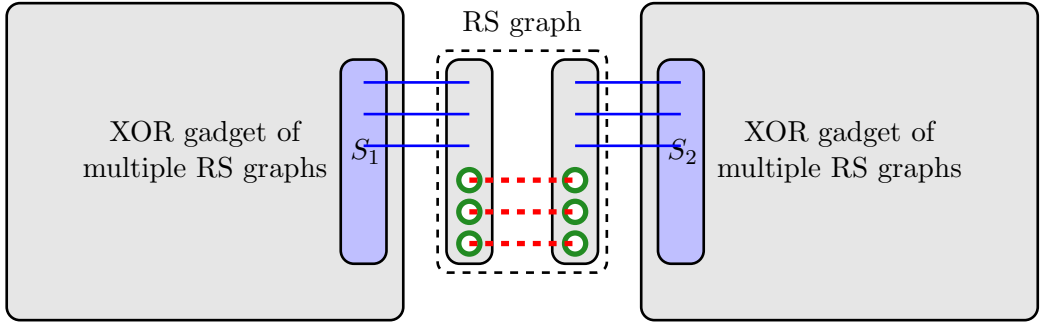
\begin{figure}[h!]
\centering
\input{figs/fig-part-4}
\caption{An illustration of the lower bound~\cite{ChenKPS0Y21} if combined with the RS graph lower bound of~\cite{GoelKK12}. }
\label{fig:part-4}
\end{figure}

Nevertheless, the approach of~\cite{ChenKPS0Y21} requires plugging in $\Omega(\log{n})$ ``XOR of RS graphs'' gadgets together to get their lower bound, thus even with the above approach, the best lower bound would be $(1-\frac{1}{\Theta(\log{n})})$-approximation (even if we assume $\RS{(n)} = n^{1-o(1)}$). 

\paragraph{$(iv)$ ``XOR Lemmas'' for proving lower bounds for XOR gadgets~\cite{AssadiN21,ChenKPS0Y21}.} Finally, let us also mention how previous work proved lower bounds for XOR gadgets. Suppose we have a streaming problem $P$ from $\set{0,1}^n \rightarrow \set{0,1}$ such that solving $P(x)$ for $x$ sampled from some distribution $\mu$ with probability of success, say, $2/3$, requires $p$-passes and $s$-space. Then, how well can we solve $P(x_1) \oplus P(x_2) \oplus \ldots \oplus P(x_k)$ 
for $k$ \emph{independent} choices of $x_1,\ldots,x_k \sim \mu^k$ via streaming algorithms? Such questions are generally referred to as XOR lemmas (in-spirit-of Yao's celebrated XOR Lemma~\cite{Yao82a}), and are the key in proving streaming lower bounds for different ``XOR gadgets'' in prior work such as~\cite{AssadiN21,ChenKPS0Y21} which require ``low-probability'' lower bounds, i.e., lower bounds that rule out even $1/\poly{(n)}$ advantage over random guessing (crucial for ``hiding'' $\Omega(n)$
vertices by union bound/hybrid arguments). In particular, 
\begin{itemize}
	\item~\cite{AssadiN21} proves that in the streaming setting, any $p$-pass $s$-space algorithm for $\oplus_{i=1}^{k} P(x_i)$ on the stream $x_1 \circ x_2 \ldots \circ x_k$ can only gain an advantage of $1/2^{\Omega(k)}$ over random guessing (as shown 
	in~\cite{AssadiN21}, this is the strongest form of XOR lemma possible in that setting). 
	\item~\cite{ChenKPS0Y21} proves that in the special case of $P$ being the \emph{Indexing function} from communication complexity, any single-pass $o(k \cdot n^{1-\Omega(1)})$-space algorithm for $\oplus_{i=1}^{k} P(x_i)$ on a certain \emph{interleaved} stream of $(x_1,\ldots,x_k)$ can only gain an advantage of $1/n^{\Omega(k)}$ over random guessing. 
\end{itemize}

The challenge in using either of these approaches for our purpose (say in a framework like~\Cref{fig:part-4}) is that they (naturally) require \emph{independent} input distributions for $x_1,\ldots,x_k$. In the context of the XOR gadget of RS graphs, this would force one to use \emph{multiple} RS graphs in the construction of the gadget. This in turn reduces the ratio of the number of edges in the hidden induced matching, to the total matching size of the graph, thus significantly reducing the bounds we 
can prove on the approximation ratio in the lower bounds.  

\subsubsection*{Our Approach} 
\vspace{-0.2cm}
In brief, we combine the framework of part $(ii)$ with the approach of part $(i)$ to ``hide'' the special induced matching from the first pass of the streaming algorithm. To do the hiding, we use a \emph{new graph product} by plugging in the XOR gadget of part $(iii)$ into a \emph{single} induced matching of an RS graph (instead of using multiple RS graphs as in part $(iii)$). Finally, we prove a XOR lemma for the case that the inputs of XOR gadgets are not  independent (coming from different RS graphs) but rather all are imposed on edges of a \emph{single} induced matching 
in an RS graph (as we will point out below, this requires an inherently different approach than part $(iv)$). See~\Cref{fig:tech} below. 
\vspace{-0.15cm}
\begin{figure}[h!]
\centering
\input{figs/fig-tech}
\caption{An illustration of our lower bound approach. The ``side'' graphs switching $S_1$ and $S_2$ are each a \emph{single} RS graph, where XOR gadgets are imposed as part of a \emph{single} one of their induced matchings.}
\label{fig:tech}
\end{figure}
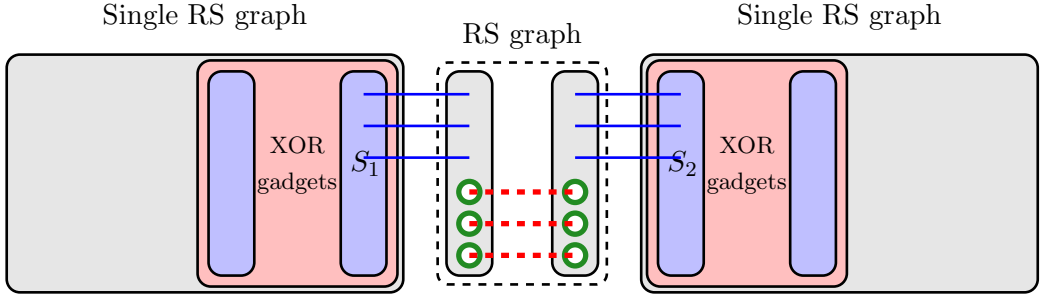

This way, we can prove that the first pass of the semi-streaming algorithm is not able to identify the switches of sets $S_1$ and $S_2$ in the construction above, thus leaving the identity of the special induced matching hidden. The crucial aspect 
of this part is that we can achieve such a gadget with size proportional to that of the hidden induced matching. In the second pass, we show that a semi-streaming algorithm that is unaware of the identity of the special induced 
matching in its second pass will not be able to obtain a sufficiently large approximation to the maximum matching. 

Several technical challenges need to be addressed in implementing this strategy. Beside the exact details of how to modify an RS graph to encode these gadgets without blowing up their size, and how to adapt the framework of~\cite{AssadiR20} to 
handle hiding subsets of vertices as opposed to single ones, the main challenge is in proving the required XOR Lemma. On a high level, the situation is as follows: each of the side RS graphs contain (roughly) $\RS{(n)}$ 
induced matchings of size $\Theta(n)$, \emph{one} of which is imposed by the XOR gadgets. The identity of these  induced matchings is originally unknown to the algorithm. So we would like to say that a semi-streaming algorithm is not able to ``solve'' these
XOR gadgets in its first pass. The problem with applying the approaches of either of~\cite{AssadiN21,ChenKPS0Y21} (part $(iv)$) is that our underlying XOR problems are \emph{correlated} by the choice of the  induced matching they are imposed on. 
Concretely, while in the work of~\cite{AssadiN21,ChenKPS0Y21} one can prove lower bounds on the advantage of algorithms for any $1/\poly{(n)}$ (by modifying the constants), such a bound is simply not true in this setting; consider the algorithm that stores all edges of a random induced matching in the side RS graphs. With probability $\nicefrac1{\RS{(n)}} \geq \nicefrac1n$, such an algorithm has all the information to the underlying XOR gadgets and can solve them exactly!  

Consequently, there is no hope of following approaches of~\cite{AssadiN21,ChenKPS0Y21} that are oblivious to this challenge. Instead, we combine a simple direct-sum style argument using information theory 
with a Fourier analysis approach motivated by the classical work of~\cite{GavinskyKKRW07} in communication complexity (which has since been used extensively to prove streaming lower bounds following~\cite{VerbinY11}). In particular, we first show that with constant probability, the entropy of edges of the XOR gadgets is ``high'', and, conditioned on this event, we prove that XOR gadgets can hide their switches 
using a simple Fourier analysis approach, similar to that of~\cite{GavinskyKKRW07}. 

\subsection{Recent Related Work}\label{sec:rec-related}
Independently and concurrently to our work, Konrad and Naidu~\cite{KonradN21} also studied two-pass semi-streaming algorithms for bipartite matching. 
They observed that currently all known two-pass streaming algorithms for maximum matching only run the greedy algorithm for maximal matching in their first pass. 
 The goal of~\cite{KonradN21} was then to understand limitation of this particular family of algorithms. They proved that any two-pass semi-streaming algorithm 
that solely runs the greedy algorithm for maximal matching in its first pass, and then run an arbitrary semi-streaming algorithm in its second pass cannot achieve a better than $(2/3)$-approximation (the paper also presents another way of 
obtaining a $(2-\sqrt{2})$-approximation two-pass semi-streaming algorithm in addition to the work of~\cite{Konrad18}). 

The idea behind the proof of~\cite{KonradN21} is as follows. The authors work with the same hard instances of \cite{GoelKK12} that proved a $(2/3)$-approximation lower bound for (general) single-pass semi-streaming algorithms (discussed in part $(i)$ of~\Cref{sec:techniques}). They then ``feed'' a fixed perfect matching of the RS graph in this construction at the beginning of the stream to the greedy algorithm so that it does not pick any edges of the second part
of the stream. As a result, the identity of the special induced matching of the RS graph remains hidden even after the first pass of this particular algorithm; thus, at the beginning of the second pass, the algorithm still needs to solve the hard problem of~\cite{GoelKK12} which is not possible in a single pass. 

We shall emphasize that the instances generated by~\cite{KonradN21} are not hard for arbitrary two-pass semi-streaming algorithms: one can simply ignore the edges of the extra perfect matching at the beginning of the stream and instead pick the 
edges that identify the special induced matching of the RS graph at the end of the stream; the second pass of the algorithm can then be spent to pick the edges of the special induced matching. This results in an $O(n)$ space algorithm that finds a perfect matching of the input graph. As such, the approach of~\cite{KonradN21} is tailored to the special family of algorithms introduced earlier. In contrast, our lower bound in this paper works for \emph{all} two-pass semi-streaming algorithms. Moreover, in terms of techniques, while our work also builds on the single-pass lower bound of~\cite{GoelKK12}, the bulk of technical work in our paper involves ``hiding'' the special induced matching of these hard instances from the first pass of arbitrary semi-streaming algorithms (which is entirely bypassed in~\cite{KonradN21} when one focuses on greedy algorithm in the first pass). As such, technique-wise, our work and~\cite{KonradN21} are almost entirely disjoint.

%% file: figs/fig-1-way.tex

\begin{tikzpicture}

\tikzset{layer/.style={rectangle, rounded corners=5pt, draw, black, line width=1pt,  fill=black!10, inner sep=4pt}}
\tikzset{vertex/.style={circle, ForestGreen, fill=white, line width=2pt, draw, minimum width=8pt, minimum height=8pt, inner sep=0pt}}
\tikzset{choose/.style={rectangle, line width=1pt, rounded corners = 5pt, draw, minimum height=5pt, fill=black!10}}
		
	\node(V1) {};
	\node[vertex] (a1) [left=-15pt of V1]{};
	\node[vertex] (a2) [right=15pt of a1]{};
	\node[vertex] (a3) [right=15pt of a2]{};
	\node[vertex] (a4) [right=15pt of a3]{};
	\node[vertex] (a5) [right=15pt of a4]{};
	\node[vertex] (a6) [right=15pt of a5]{};
		
	\node (V2) [below = 30pt of V1]{};
	\node[vertex] (b1) [left=-15pt of V2]{};
	\node[vertex] (b2) [right=15pt of b1]{};
	\node[vertex] (b3) [right=15pt of b2]{};
	\node[vertex] (b4) [right=15pt of b3]{};
	\node[vertex] (b5) [right=15pt of b4]{};
	\node[vertex] (b6) [right=15pt of b5]{};
	
	\node[choose,  fit=(a1) (a3)]{};
	\node[choose, fit=(a4) (a6)]{};
	\node[choose, fit=(b1) (b3)]{};
	\node[choose, fit=(b4) (b6)]{};
	

	\node[dashed, draw, rounded corners, line width=1pt, inner sep=7pt,  fit=(a1) (b6)] (RS) {};	
	\node[above =1pt of RS] {RS graph};
	
	\node[vertex] (a1) [left=-15pt of V1]{};
	\node[vertex] (a2) [right=15pt of a1]{};
	\node[vertex] (a3) [right=15pt of a2]{};
	\node[vertex] (a4) [right=15pt of a3]{};
	\node[vertex] (a5) [right=15pt of a4]{};
	\node[vertex] (a6) [right=15pt of a5]{};
	
	\node[vertex] (b1) [left=-15pt of V2]{};
	\node[vertex] (b2) [right=15pt of b1]{};
	\node[vertex] (b3) [right=15pt of b2]{};
	\node[vertex] (b4) [right=15pt of b3]{};
	\node[vertex] (b5) [right=15pt of b4]{};
	\node[vertex] (b6) [right=15pt of b5]{};
	
	\node(V3) [left=75 pt of V1]{};
	\node[vertex, blue!50, fill=blue!20] (c1) [left=-15pt of V3]{};
	\node[vertex, blue!50, fill=blue!20]  (c2) [right=15pt of c1]{};
	\node[vertex, blue!50, fill=blue!20]  (c3) [right=15pt of c2]{};
	
	\node(V4) [below=30 pt of V3]{};
	\node[vertex, blue!50, fill=blue!20] (d1) [left=-15pt of V4]{};
	\node[vertex, blue!50, fill=blue!20]  (d2) [right=15pt of d1]{};
	\node[vertex, blue!50, fill=blue!20]  (d3) [right=15pt of d2]{};

	\node[choose,  fit=(c1) (c3)]{};
	\node[choose,  fit=(d1) (d3)]{};
	
	\draw[line width=1pt, blue]	
	(c1.center) -- (b1.center)
	(c2.center) -- (b2.center)
	(c3.center) -- (b3.center)
	
	(d1.center) -- (a1.center)
	(d2.center) -- (a2.center)
	(d3.center) -- (a3.center);
	
	\draw[line width=2pt, red, dashed]	
	(a4.center) -- (b4.center)
	(a5.center) -- (b5.center)
	(a6.center) -- (b6.center);	
	\node[vertex, blue!50, fill=blue!20] (c1) [left=-15pt of V3]{};
	\node[vertex, blue!50, fill=blue!20]  (c2) [right=15pt of c1]{};
	\node[vertex, blue!50, fill=blue!20]  (c3) [right=15pt of c2]{};
	
	\node[vertex, blue!50, fill=blue!20] (d1) [left=-15pt of V4]{};
	\node[vertex, blue!50, fill=blue!20]  (d2) [right=15pt of d1]{};
	\node[vertex, blue!50, fill=blue!20]  (d3) [right=15pt of d2]{};
	
\end{tikzpicture}

%% file: figs/fig-2-pass.tex

\begin{tikzpicture}

\tikzset{layer/.style={rectangle, rounded corners=5pt, draw, black, line width=1pt,  fill=black!10, inner sep=4pt}}
\tikzset{vertex/.style={circle, ForestGreen, fill=white, line width=2pt, draw, minimum width=8pt, minimum height=8pt, inner sep=0pt}}
\tikzset{choose/.style={rectangle, line width=1pt, rounded corners = 5pt, draw, minimum height=5pt, fill=black!10}}

	\node[vertex] (a1) {}; 
	\node[vertex] (a2) [below=2pt of a1]{};
	\node[vertex] (a3) [below=2pt of a2]{};
	\node[vertex] (a4) [below=2pt of a3]{};
	\node[vertex] (a5) [below=2pt of a4]{};
	\node[vertex] (a6) [below=2pt of a5]{};
	\node[choose,  fit=(a1) (a6)]{};
	
	\node[vertex] (b1) [right=30pt of a1]{}; 
	\node[vertex] (b2) [below=2pt of b1]{};
	\node[vertex] (b3) [below=2pt of b2]{};
	\node[vertex] (b4) [below=2pt of b3]{};
	\node[vertex] (b5) [below=2pt of b4]{};
	\node[vertex] (b6) [below=2pt of b5]{};
	
	\node[choose,  fit=(b1) (b6)]{};
	
	\node[dashed, draw, rounded corners, line width=1pt, inner sep=7pt,  fit=(a1) (b6)] (RS) {};	
	\node[above =1pt of RS] {Random graph};
	
	\node[vertex] (x1) [left =120 pt of a1]{};
	
	\node[vertex] (x21) [left=30 pt of a1]{};
	\node[vertex] (x22) [left=30 pt of a2]{};
	\node[vertex] (x23) [left=30 pt of a3]{};
	\node[vertex] (x24) [left=30 pt of a4]{};
	\node[vertex] (x25) [left=30 pt of a5]{};
	\node[vertex] (x26) [left=30 pt of a6]{};
	
	\node[choose,  fit=(x1) (x26)]{RS gadget};
	
	\node[vertex] (y1) [right = 120 pt of b1]{};
	\node[vertex] (y21) [right=30 pt of b1]{};
	\node[vertex] (y22) [right=30 pt of b2]{};
	\node[vertex] (y23) [right=30 pt of b3]{};
	\node[vertex] (y24) [right=30 pt of b4]{};
	\node[vertex] (y25) [right=30 pt of b5]{};
	\node[vertex] (y26) [right=30 pt of b6]{};
	
	\node[choose,  fit=(y1) (y26)]{RS gadget};
	
	\draw[line width=1pt, blue]	
	(a1.center) -- (x21.center)
	(a2.center) -- (x22.center)	
	(a3.center) -- (x23.center)
	(a4.center) -- (x24.center)	
	(a6.center) -- (x26.center)	
	
	
	(b1.center) -- (y21.center)
	(b2.center) -- (y22.center)	
	(b4.center) -- (y24.center)	
	(b5.center) -- (y25.center)
	(b6.center) -- (y26.center)	;

	\node[vertex] (a4) [below=15pt of a3]{};
	\node[vertex] (b2) [below=15pt of b1]{};
	
		\draw[line width=2pt, dashed, red]	
	(b2.center) -- (a4.center);

\end{tikzpicture}

%% file: figs/fig-part-3.tex

\begin{tikzpicture}

\tikzset{layer/.style={rectangle, rounded corners=5pt, draw, black, line width=1pt,  fill=black!10, inner sep=4pt}}
\tikzset{vertex/.style={circle, ForestGreen, fill=white, line width=2pt, draw, minimum width=8pt, minimum height=8pt, inner sep=0pt}}
\tikzset{choose/.style={rectangle, line width=1pt, rounded corners = 5pt, draw, minimum height=5pt, fill=black!10}}

	\node[vertex] (a1) {}; 
	\node[vertex] (a2) [below=2pt of a1]{};
	\node[vertex] (a3) [below=2pt of a2]{};
	\node[vertex] (a4) [below=2pt of a3]{};
	\node[vertex] (a5) [below=2pt of a4]{};
	\node[vertex] (a6) [below=2pt of a5]{};
	\node[choose,  fit=(a1) (a6)]{};
	
	\node[vertex] (b1) [right=30pt of a1]{}; 
	\node[vertex] (b2) [below=2pt of b1]{};
	\node[vertex] (b3) [below=2pt of b2]{};
	\node[vertex] (b4) [below=2pt of b3]{};
	\node[vertex] (b5) [below=2pt of b4]{};
	\node[vertex] (b6) [below=2pt of b5]{};
	
	\node[choose,  fit=(b1) (b6)]{};
	
	\node[dashed, draw, rounded corners, line width=1pt, inner sep=7pt,  fit=(a1) (b6)] (RS) {};	
	\node[above =1pt of RS] {RS graph};
	
	\node[vertex] (x1) [left =150 pt of a1]{};
	\node[vertex] (x2) [left =150 pt of a2]{};
	\node[vertex] (x3) [left =150 pt of a3]{};
	\node[vertex] (x4) [left =150 pt of a4]{};
	\node[vertex] (x5) [left =150 pt of a5]{};
	\node[vertex] (x6) [left =150 pt of a6]{};
	
	\node [left=10 pt of x1] {OFF};
	\node [left=10 pt of x2] {OFF};
	\node [left=10 pt of x3] {OFF};
	\node [left=10 pt of x4] {ON};
	\node [left=10 pt of x5] {ON};
	\node [left=10 pt of x6] {ON};

	\node[vertex] (x21) [left=30 pt of a1]{};
	\node[vertex] (x22) [left=30 pt of a2]{};
	\node[vertex] (x23) [left=30 pt of a3]{};
	\node[vertex] (x24) [left=30 pt of a4]{};
	\node[vertex] (x25) [left=30 pt of a5]{};
	\node[vertex] (x26) [left=30 pt of a6]{};
	
	\node[choose, inner sep = 0pt,  fit=(x1) (x21)]{};
	\node[choose, inner sep = 0pt,  fit=(x2) (x22)]{};
	\node[choose, inner sep = 0pt,  fit=(x3) (x23)]{};
	\node[choose, inner sep = 0pt,  fit=(x4) (x24)]{};
	\node[choose, inner sep = 0pt,  fit=(x5) (x25)]{};
	\node[choose, inner sep = 0pt,  fit=(x6) (x26)]{};
						
	\node[vertex] (y1) [right = 150 pt of b1]{};
	\node[vertex] (y2) [right = 150 pt of b2]{};
	\node[vertex] (y3) [right = 150 pt of b3]{};
	\node[vertex] (y4) [right = 150 pt of b4]{};
	\node[vertex] (y5) [right = 150 pt of b5]{};
	\node[vertex] (y6) [right = 150 pt of b6]{};
	
		\node [right=10 pt of y1] {OFF};
	\node [right=10 pt of y2] {OFF};
	\node [right=10 pt of y3] {OFF};
	\node [right=10 pt of y4] {ON};
	\node [right=10 pt of y5] {ON};
	\node [right=10 pt of y6] {ON};
	
	\node[vertex] (y21) [right=30 pt of b1]{};
	\node[vertex] (y22) [right=30 pt of b2]{};
	\node[vertex] (y23) [right=30 pt of b3]{};
	\node[vertex] (y24) [right=30 pt of b4]{};
	\node[vertex] (y25) [right=30 pt of b5]{};
	\node[vertex] (y26) [right=30 pt of b6]{};
	
	\node[choose, inner sep = 0pt,  fit=(y1) (y21)]{};
	\node[choose, inner sep = 0pt,  fit=(y2) (y22)]{};
	\node[choose, inner sep = 0pt,  fit=(y3) (y23)]{};
	\node[choose, inner sep = 0pt,  fit=(y4) (y24)]{};
	\node[choose, inner sep = 0pt,  fit=(y5) (y25)]{};
	\node[choose, inner sep = 0pt,  fit=(y6) (y26)]{};

	
	\draw[line width=1pt, blue]	
	(a1.center) -- (x21.center)
	(a2.center) -- (x22.center)	
	(a3.center) -- (x23.center)
	
	
	(b1.center) -- (y21.center)
	(b2.center) -- (y22.center)	
	(b3.center) -- (y23.center);

	\node[vertex] (w1) [below=15pt of a2]{};
	\node[vertex] (w2) [below=15pt of a3]{};
	\node[vertex] (w3) [below=15pt of a4]{};
	\node[vertex] (z1) [below=15pt of b2]{};
	\node[vertex] (z2) [below=15pt of b3]{};
	\node[vertex] (z3) [below=15pt of b4]{};
	
	\draw[line width=2pt, dashed, red]	
	(w1.center) -- (z1.center)
	(w2.center) -- (z2.center)
	(w3.center) -- (z3.center);

\end{tikzpicture}

%% file: figs/fig-part-4.tex

\begin{tikzpicture}

\tikzset{layer/.style={rectangle, rounded corners=5pt, draw, black, line width=1pt,  fill=black!10, inner sep=4pt}}
\tikzset{vertex/.style={circle, ForestGreen, fill=white, line width=2pt, draw, minimum width=8pt, minimum height=8pt, inner sep=0pt}}
\tikzset{choose/.style={rectangle, line width=1pt, rounded corners = 5pt, draw, minimum height=5pt, fill=black!10}}

	\node[vertex] (a1) {}; 
	\node[vertex] (a2) [below=2pt of a1]{};
	\node[vertex] (a3) [below=2pt of a2]{};
	\node[vertex] (a4) [below=2pt of a3]{};
	\node[vertex] (a5) [below=2pt of a4]{};
	\node[vertex] (a6) [below=2pt of a5]{};

	\node[choose,  fit=(a1) (a6)]{};
	
	\node[vertex] (b1) [right=30pt of a1]{}; 
	\node[vertex] (b2) [below=2pt of b1]{};
	\node[vertex] (b3) [below=2pt of b2]{};
	\node[vertex] (b4) [below=2pt of b3]{};
	\node[vertex] (b5) [below=2pt of b4]{};
	\node[vertex] (b6) [below=2pt of b5]{};
	
	\node[choose,  fit=(b1) (b6)]{};
	
	\node[dashed, draw, rounded corners, line width=1pt, inner sep=7pt,  fit=(a1) (b6)] (RS) {};	
	\node[above =1pt of RS] {RS graph};
	
	\node[vertex] (x1) [left =120 pt of a1]{};	
	\node[vertex] (x21) [left=30 pt of a1]{};
	\node[vertex] (x22) [left=30 pt of a2]{};
	\node[vertex] (x23) [left=30 pt of a3]{};
	\node[vertex] (x24) [left=30 pt of a4]{};
	\node[vertex] (x25) [left=30 pt of a5]{};
	\node[vertex] (x26) [left=30 pt of a6]{};

	\node[vertex] (xx0) [above =5pt of x1]{};	
	\node[vertex] (xx1) [below =5pt of x6]{};	
	
	\node[choose,  fit=(xx0) (xx1) (x26), inner sep=10pt ]{XOR gadget of \\  multiple RS graphs};
	
	\node[choose,  fit=(x21) (x26), fill=blue!25]{$S_1$};

	\node[vertex] (y1) [right = 120 pt of b1]{};
	\node[vertex] (y21) [right=30 pt of b1]{};
	\node[vertex] (y22) [right=30 pt of b2]{};
	\node[vertex] (y23) [right=30 pt of b3]{};
	\node[vertex] (y24) [right=30 pt of b4]{};
	\node[vertex] (y25) [right=30 pt of b5]{};
	\node[vertex] (y26) [right=30 pt of b6]{};
	
	\node[vertex] (yy0) [above =5pt of y1]{};	
	\node[vertex] (yy1) [below =5pt of y6]{};	
	
	\node[choose,  fit=(yy0) (yy1) (y26),  inner sep=10pt ]{XOR gadget of \\  multiple RS graphs};
	
	\node[choose,  fit=(y21) (y26), fill=blue!25, ]{$S_2$};
	
	\draw[line width=1pt, blue]	
	(a1.center) -- (x21.center)
	(a2.center) -- (x22.center)	
	(a3.center) -- (x23.center)
	(b1.center) -- (y21.center)
	(b2.center) -- (y22.center)	
	(b3.center) -- (y23.center);
	
	\node[vertex] (w1) [below=15pt of a2]{};
	\node[vertex] (w2) [below=15pt of a3]{};
	\node[vertex] (w3) [below=15pt of a4]{};
	\node[vertex] (z1) [below=15pt of b2]{};
	\node[vertex] (z2) [below=15pt of b3]{};
	\node[vertex] (z3) [below=15pt of b4]{};
	
	\draw[line width=2pt, dashed, red]	
	(w1.center) -- (z1.center)
	(w2.center) -- (z2.center)
	(w3.center) -- (z3.center);

\end{tikzpicture}

%% file: figs/fig-tech.tex

\begin{tikzpicture}

\tikzset{layer/.style={rectangle, rounded corners=5pt, draw, black, line width=1pt,  fill=black!10, inner sep=4pt}}
\tikzset{vertex/.style={circle, ForestGreen, fill=white, line width=2pt, draw, minimum width=8pt, minimum height=8pt, inner sep=0pt}}
\tikzset{choose/.style={rectangle, line width=1pt, rounded corners = 5pt, draw, minimum height=5pt, fill=black!10}}

	\node[vertex] (a1) {}; 
	\node[vertex] (a2) [below=2pt of a1]{};
	\node[vertex] (a3) [below=2pt of a2]{};
	\node[vertex] (a4) [below=2pt of a3]{};
	\node[vertex] (a5) [below=2pt of a4]{};
	\node[vertex] (a6) [below=2pt of a5]{};

	\node[choose,  fit=(a1) (a6)]{};
	
	\node[vertex] (b1) [right=30pt of a1]{}; 
	\node[vertex] (b2) [below=2pt of b1]{};
	\node[vertex] (b3) [below=2pt of b2]{};
	\node[vertex] (b4) [below=2pt of b3]{};
	\node[vertex] (b5) [below=2pt of b4]{};
	\node[vertex] (b6) [below=2pt of b5]{};
	
	\node[choose,  fit=(b1) (b6)]{};
	
	\node[dashed, draw, rounded corners, line width=1pt, inner sep=7pt,  fit=(a1) (b6)] (RS) {};	
	\node[above =1pt of RS] {RS graph};
	
	\node[vertex] (x1) [left =150 pt of a1]{};	
	\node[vertex] (x21) [left=30 pt of a1]{};
	\node[vertex] (x22) [left=30 pt of a2]{};
	\node[vertex] (x23) [left=30 pt of a3]{};
	\node[vertex] (x24) [left=30 pt of a4]{};
	\node[vertex] (x25) [left=30 pt of a5]{};
	\node[vertex] (x26) [left=30 pt of a6]{};

	\node[choose,  fit=(x1) (x26), inner sep=10pt ](RS1){};
	\node[above=5pt of RS1]{Single RS graph};
	
	\node[choose,  fit=(x21) (x26), fill=blue!25] (S1) {$S_1$};
	
	\node[vertex] (ss21) [left=40 pt of x21]{};
	\node[vertex] (ss26) [left=40 pt of x26]{};

	\node[choose,  fit=(ss21) (ss26), fill=blue!25] (T1) {};

	\node[choose,  fit=(S1) (T1), fill=red!25] {\footnotesize XOR \\ gadgets};
	\node[choose,  fit=(x21) (x26), fill=blue!25] (S1) {$S_1$};
	\node[choose,  fit=(ss21) (ss26), fill=blue!25] (T1) {};

	\node[vertex] (y1) [right = 150 pt of b1]{};
	\node[vertex] (y21) [right=30 pt of b1]{};
	\node[vertex] (y22) [right=30 pt of b2]{};
	\node[vertex] (y23) [right=30 pt of b3]{};
	\node[vertex] (y24) [right=30 pt of b4]{};
	\node[vertex] (y25) [right=30 pt of b5]{};
	\node[vertex] (y26) [right=30 pt of b6]{};
	
	\node[choose,  fit=(y1)  (y26),  inner sep=10pt] (RS2){};
	\node[above=5pt of RS2]{Single RS graph};
	
	\node[choose,  fit=(y21) (y26), fill=blue!25] (S2) {$S_2$};
	
	\node[vertex] (tt21) [right=40 pt of y21]{};
	\node[vertex] (tt26) [right=40 pt of y26]{};

	\node[choose,  fit=(tt21) (tt26), fill=blue!25] (T2) {};

	\node[choose,  fit=(S2) (T2), fill=red!25] {\footnotesize XOR \\ gadgets};
		\node[choose,  fit=(y21) (y26), fill=blue!25] (S2) {$S_2$};
		\node[choose,  fit=(tt21) (tt26), fill=blue!25] (T2) {};

	\draw[line width=1pt, blue]	
	(a1.center) -- (x21.center)
	(a2.center) -- (x22.center)	
	(a3.center) -- (x23.center)
	(b1.center) -- (y21.center)
	(b2.center) -- (y22.center)	
	(b3.center) -- (y23.center);
	
	\node[vertex] (w1) [below=15pt of a2]{};
	\node[vertex] (w2) [below=15pt of a3]{};
	\node[vertex] (w3) [below=15pt of a4]{};
	\node[vertex] (z1) [below=15pt of b2]{};
	\node[vertex] (z2) [below=15pt of b3]{};
	\node[vertex] (z3) [below=15pt of b4]{};
	
	\draw[line width=2pt, dashed, red]	
	(w1.center) -- (z1.center)
	(w2.center) -- (z2.center)
	(w3.center) -- (z3.center);

\end{tikzpicture}

%% file: prelim.tex

\section{Preliminaries}

\paragraph{Notation.} For any integer $t \geq 1$, we use $[t] := \set{1,\ldots,t}$. For any $k$-sequence $X = (X_1,\ldots,X_k)$ and integer $i \in [k]$, we define $X^{<i} := (X_1,\ldots,X_{i-1})$, and $X^{-i} := (X_1,\ldots,X_{i-1},X_{i+1},\ldots,X_k)$. 

For a graph $G=(V,E)$, and vertices $U \subseteq V$, we use $G[U]$ to denote the induced subgraph of $G$ on $U$. For any vertex $v \in V$, and a matching $M$ in $G$, we use $M(v)$ to denote
the matched pair of $v$ ($M(v) = \perp$ if $v$ is unmatched by $M$). We denote bipartite graphs by $G=(L,R,E)$ to specify the bipartition into $L$ and $R$, and for any set $F \subseteq E$, use $L(F)$ and $R(F)$ to 
denote the endpoints of edges in $F$ in $L$ and $R$, respectively. Throughout the paper, by a \emph{$(2n)$-vertex bipartite graph}, we always mean a bipartite graph with $\card{L}=\card{R} = n$. We also use the following fact
about graphs. 

\begin{fact}\label{fact:vc-matching}
	In any graph $G$, size of any vertex cover of $G$ is at least as large as any matching in $G$. 
\end{fact}

We use `sans serif' letters to denote random variables (e.g., $\rA$) , and the corresponding normal letters to denote their values (e.g. $A$). To avoid the clutter in notation, 
in conditioning terms which involve assignments to random variables, we may directly use the value of the random variable (with the same letter), e.g., write $\rB \mid A$ instead of $\rB \mid \rA = A$. 

For random variables $\rA,\rB$, we use $\en{\rA}$ and $\mi{\rA}{\rB}$ to denote the Shannon entropy and mutual information, respectively. Moreover, for two distributions 
$\mu,\nu$, $\tvd{\mu}{\nu}$ denotes the total variation distance, and $\kl{\mu}{\nu}$ is the KL-divergence. A summary of basic information theory definitions and facts
that we use in our proofs appear in~\Cref{sec:info}. 

For a function $f : \set{0,1}^n \rightarrow \IR$, we use $\hf : 2^{[n]} \rightarrow \IR$ to denote the (discrete) Fourier transform of $f$. For any $S \subseteq [n]$, $\Xor_S : \set{0,1}^{n} \rightarrow \set{-1,+1}$ denotes
the character function on $S$. A summary of basic definitions and tools from Fourier analysis on Boolean hypercube that we use in our proofs appear in~\Cref{sec:fourier}. 



\subsection{Communication Complexity}\label{sec:cc}

We work with the two-party communication model of Yao~\cite{Yao79} (with some slightly non-standard aspects mentioned later on). See the excellent textbooks by Kushilevitz and Nisan~\cite{KushilevitzN97} and Rao and Yehudayoff~\cite{RaoY20} for 
an overview of communication complexity. 

Let $P: \mathcal{X} \times \mathcal{Y} \rightarrow \mathcal{Z}$ be a relation.  Alice receives an input $X\in \mathcal{X}$ and Bob receives $Y \in \mathcal{Y}$, where $(X,Y)$ are chosen from a
distribution $\dist$ over $\mathcal{X} \times \mathcal{Y}$. We allow players to have access to both public and private randomness. 
They communicate with each other by exchanging messages according to some \emph{protocol} $\prot$. Each message in $\prot$
depends only on the private input and random bits of the player sending the message, the already communicated messages, and the public randomness. 
At the end, one of the players outputs an answer  $Z$ such that $Z \in P(X,Y)$. For any protocol $\prot$, we use $\Prot := \Prot(X,Y)$ to denote the messages \emph{and} the public randomness used 
by $\prot$ on the input $(X,Y)$.


\subsection{Bipartite \rs Graphs}\label{sec:rs}

Let $G=(V,E)$ be an undirected graph, and $M \subseteq E$ be a matching in $G$. We say that $M$ is an \emph{induced matching} iff the subgraph of $G$ induced on the vertices of $M$ is the matching $M$ itself; in other words, 
there are no other edges between the vertices of this matching. 

\begin{definition}[\textbf{Bipartite \rs Graphs}]\label{def:rs}
For integers $r,t \geq 1$, a bipartite graph $\GRS=(L,R,E)$ is called an \emph{$(r,t)$-\rs graph} (RS graph for short) iff its edge-set $E$ can be partitioned into $t$ \underline{induced} matchings $\MRS_1,\ldots, \MRS_t$, each of size $r$.
\end{definition}

RS graphs have been extensively studied as they arise naturally in property testing, PCP constructions, additive combinatorics, streaming algorithms, graph sparsification, etc. (see, e.g.,~\cite{BirkLM93,HastadW03,FischerLNRRS02,Alon02,TaoV06,AlonS06,AlonMS12,GoelKK12,FoxHS15,AssadiB19,KapralovKTY21}).  In particular, a line of work 
initiated by Goel, Kapralov, and Khanna~\cite{GoelKK12} have used different constructions of these graphs to prove communication complexity lower bounds for graph streaming algorithms~\cite{GoelKK12,Kapralov13,Konrad15,AssadiKLY16,AssadiKL17,CormodeDK19,AssadiR20,Kapralov21,AssadiB21,ChenKPSSY21}. 

\paragraph{A Remark on Bipartite vs Non-Bipartite RS Graphs.} In this work, we focus on bipartite RS graphs, while many constructions in the literature are non-bipartite RS graphs. However, any $(r,t)$-RS non-bipartite graph $G$ on $n$ vertices can be turned into a $(2r,t)$-RS bipartite graph on $2n$ vertices 
by simply taking the bipartite double cover of $G$. As such, throughout this paper, by RS graphs, we always mean \emph{bipartite} RS graphs.

%% file: game.tex

\section{A New Communication Game: Hidden-Matching} \label{sec:communication-game}
We introduce the main communication game we study in this paper in this section. We start by presenting basic constructs we need to setup our communication game, and then present the  game itself together with its underlying hard distribution. 

\subsection{Encoded-RS Graphs}\label{sec:encoded-rs}
 
 We define a simple way of encoding an $(r \times t)$-dimensional matrix inside any arbitrary $(r,t)$-RS graph, to obtain another RS graph with certain properties needed for our proofs. 

\begin{definition}[\textbf{Encoded-RS Graph}]\label{def:encoded-rs}
	Let $\GRS=(L,R,\ERS)$ be an $(r,t)$-RS graph with induced matchings $\MRS_1,\ldots,\MRS_t$ and $X \in \set{0,1}^{r \times t}$ be an  $(r \times t)$-dimensional binary matrix. We define the \textbf{encoded-RS} graph of $G$ and $X$, denoted by $H:= \encodedrs(\GRS,X)$, as the following graph: 
	\begin{itemize}[leftmargin=10pt]
		\item For any vertex $v \in V$, we create two vertices $a_v,b_v$ in $H$. We refer to $a_v,b_v$ as \textbf{representatives} of $v$ and denote them together by $\rep{v} := \set{a_v,b_v}$. 
		\item For any induced matching $\MRS_j$ and any edge $e_i = (u_i,v_i)$ of $\MRS_j$: 
		\begin{enumerate}[label=$(\roman*)$]
		\item if $X_{i,j} = 0$, we add two edges $(a_{u_i},a_{v_i})$ and $(b_{u_i},b_{v_i})$ to $H$; 
		\item if $X_{i,j} = 1$, we add two edges $(a_{u_i},b_{v_i})$ and $(b_{u_i},a_{v_i})$ to $H$. 
		\end{enumerate}
		We refer to the new matching in $H$ obtained from edges $\MRS_j$ as the \textbf{representative} of $\MRS_j$, and denote it by $\rep{\MRS_j}$. 
	\end{itemize}
\end{definition}
\noindent
\Cref{fig:encoded-rs} below gives an illustration. 
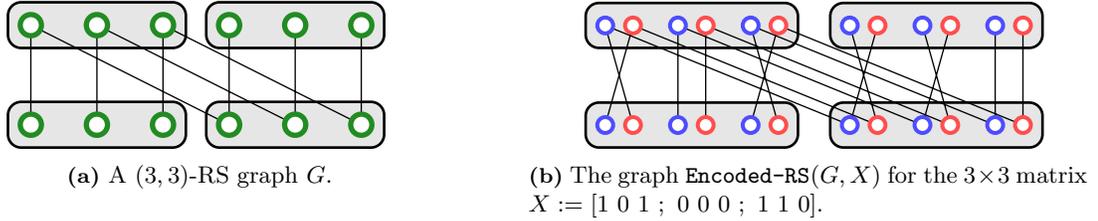
\begin{figure}[h!]
\centering
\subcaptionbox{\footnotesize A $(3,3)$-RS graph $G$.}%
  [.45\linewidth]{
\input{figs/fig-rs}
  } 
  \hspace{0.4cm} 
  \subcaptionbox{\footnotesize The graph $\encodedrs(G,X)$ for the $3 \times 3$ matrix $X:= 
 [
  1 ~ 0 ~ 1 ~ ; ~ 0 ~ 0 ~ 0 ~ ;  ~ 1 ~ 1 ~ 0
  ]$.}%
  [.45\linewidth]{
\input{figs/fig-encoded-rs}
  }
\caption{An illustration of RS graphs and encoded-RS graphs.}
\label{fig:encoded-rs}
\end{figure}


\begin{observation}\label{obs:encoded-rs}
	For any $(2n)$-vertex $(r,t)$-RS graph $\GRS$ and $(r \times t)$-dimensional  matrix $X$, the graph $H:= \encodedrs(\GRS,X)$ is a $(4n)$-vertex $(2r,t)$-RS graph. 
\end{observation}
\begin{proof}
	For any induced matching $\MRS$ in $\GRS$, the matching $M = \rep{\MRS}$ has size $2r$ in $H$. Moreover, $M$ is induced as each edge $(u,v) \in \MRS$ is translated to a perfect matching between $\rep{u}$ and $\rep{v}$ in $H$; thus any edge
	violating the induced property of $M$ in $H$ would correspond to an edge violating induced property of $\MRS$ in $\GRS$ as well, which is not possible. 
\end{proof}

\paragraph{Augmenting edges and paths.} A key definition in encoded-RS graphs is the following. 

\begin{definition}[\textbf{Augmenting Edges/Path}]\label{def:augmenting-edges}
	Consider any $(r,t)$-RS graph $\GRS=(L,R,\ERS)$, $(r \times t)$ binary matrix $X$, index $j \in [t]$, and a sequence $\vec{u} = (u_1,\ldots,u_k)$ of $k \geq 2$ distinct vertices in $L(\MRS_j)$. Let $H=\encodedrs(G,X)$ and $v_i := \MRS_j(u_i) \in R(\MRS_j)$ for all $i \in [k]$. 
	
	\noindent
	We define the \textbf{augmenting edges}, denoted by $AE := \augmentedge(\GRS,X,j,\vec{u})$, as the following edges on vertices of $H$ (note that these edges do \underline{not} belong to $H$):
	\begin{itemize}
		\item For any $i \in [k-1]$, add the edges $(a_{v_i},a_{u_{i+1}})$ and $(b_{v_i},b_{u_{i+1}})$  to $AE$ (by~\Cref{obs:encoded-rs}, these edges do not belong to $H$).
	\end{itemize}
	
	\noindent
	We define an \textbf{augmenting path}, denoted by $AP := \augmentpath(\GRS,X,j,\vec{u})$, as the following path on vertices of $H$ (consisting of edges from $H$ plus augmenting edges): 
	\begin{itemize}
		\item There is a \underline{unique} path from $a_{u_1}$ to either $a_{v_k}$ or $b_{v_k}$ by alternatively following the edges of $R(\MRS)$ in $H$ and augmenting edges in $AE$. 
		We denote this path by $AP(\vec{u})$ and refer to it as an augmenting path. We further use $\start{AP(\vec{u})} = a_{u_1}$ and $\eend{AP(\vec{u})} \in \set{a_{v_k},b_{v_k}}$ to denote the start and end vertex of the path.
	\end{itemize}

\end{definition}
\noindent
\Cref{fig:aug-edge} below gives an illustration. 
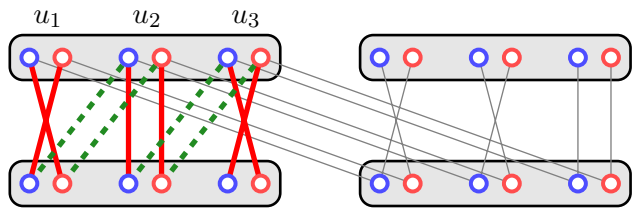
\begin{figure}[h!]
\centering
\input{figs/fig-encoded-rs-2}
\caption{An illustration of augmenting edges and paths. Here, $\MRS$ is the left induced matching, thick (red) edges denote the representative matching of $\MRS$ in $H$, and dashed (green) edges are augmenting edges. An augmenting path
here starts from a top vertex corresponding to $u_1$ and follows thick (red) and dashed (green) edges alternatively to end up at a unique bottom vertex.}
\label{fig:aug-edge}
\end{figure}

The following observation  summarizes the main property of encoded-RS graphs and augmenting paths that we use in our proofs. 

\begin{observation}\label{obs:aug-path}
	Consider augmenting paths $AP := \augmentpath(\GRS,X,j,\vec{u})$. Then, 
	\begin{enumerate}[label=$(\roman*)$]
	\item if $X_{i_1,j} \oplus \cdots \oplus X_{i_k,j} = 0$,  we have $\eend{AP(\vec{u})}={a_{v_{i_k}}}$; 
	\item  if $X_{i_1,j} \oplus \cdots \oplus X_{i_k,j} = 1$, we have $\eend{AP(\vec{u})}={b_{v_{i_k}}}$. 
	\end{enumerate}
\end{observation}
\begin{proof}
	Consider $AP(\vec{u})$ which starts at $a_{u_{i_1}}$. The next vertex on this path is $a_{v_{i_1}}$ if $X_{i_1,j} = 0$ and $b_{v_{i_1}}$ if $X_{i_1,j} = 1$ (this is by construction of encoded-RS graphs). The vertex after that is $a_{u_{i_2}}$ if we were at $a_{v_{i_1}}$, and $b_{u_{i_2}}$ if we were instead at $b_{v_{i_1}}$ (this is 
	by construction of augmenting paths). Continuing this inductively until the last vertex implies the observation. 
\end{proof}

\subsection{Augmentation Graphs}\label{sec:augmentation-graphs} 

We now define a new construction that builds on top of encoded-RS graphs. We first need a quick notation. For any set $W \subseteq V$ of vertices, we say a collection $\UU$ of $k$-sequences on $W$ is \emph{vertex-disjoint} if it consists of $k$-sequences $\vec{u}_i = (u_{i,1},\ldots,u_{i,k})$ such that the vertices used 
across these all sequences are distinct.

\begin{definition}[\textbf{Augmentation Graph/Vertices}]\label{def:aug-graph}
	For any $(r,t)$-RS graph $\GRS=(L,R,\ERS)$, $(r \times t)$ binary matrix $X$, index $j \in [t]$, and a collection $\UU$ of vertex-disjoint $k$-sequences on $L(\MRS_j)$, we define the \textbf{augmentation graph}, denoted by $A := \augmentgraph(\GRS,X,j,\UU)$ as follows: 
	\begin{itemize}
		\item $A$ is a graph on vertices of $H=\encodedrs(\GRS,X)$ plus two new sets of vertices $P,Q$;
		\item $A$ consists of all augmenting edges $AE_i = \augmentedge(G,X,j,\vec{u}_i)$ for $\vec{u}_i \in \UU$, plus a perfect matching between $P$ and vertices of $H$ \emph{not} matched by $\rep{\MRS_j}$, as well as a perfect matching between $Q$ and vertices of $\start{AP(\vec{u}_i)}$ for $\vec{u}_i$ in $\UU$. We use $\bH$ to denote this set of edges. 
	\end{itemize}
	Finally, we define the following vertices in $A$, referred to collectively as \textbf{augmentation vertices}: 
	\begin{align*}
		\aug{A} := \set{a_{v_{i,k}} \mid a_{v_{i,k}} = \eend{AP(\vec{u}_i)}}; \quad \baug{A} := \set{a_{v_{i,k}} \mid a_{v_{i,k}} \neq \eend{AP(\vec{u}_i)} }.
	\end{align*}
	(in words, $\aug{A}$ are end vertices of augmenting paths that are $a$-vertices, and $\baug{A}$ are those $a$-vertices whose corresponding augmenting paths end in a $b$-vertex instead). 
\end{definition}

\noindent
\Cref{fig:aug-graph} below gives an illustration. 
\begin{figure}[h!]
\centering
\input{figs/fig-aug-graph}4
\caption{An illustration of an augmentation graph $A$. All edges drawn belong to $A$ and dashed edges are augmenting edges -- these edges collectively form $\bH$. The middle graph is the encoded-RS graph $H$, whose edges are omitted. 
}
\label{fig:aug-graph}
\end{figure}
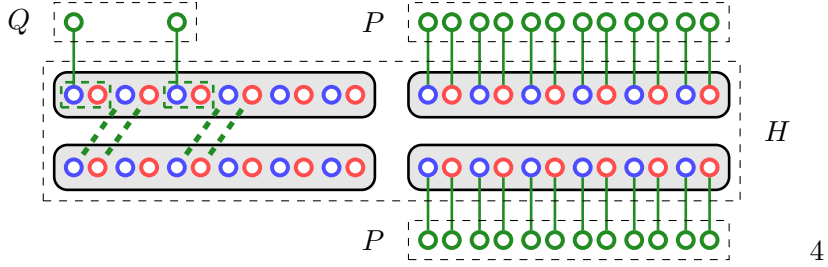

\begin{observation}\label{obs:augmentation-graph}
	For any $(2n)$-vertex $(r,t)$-RS graph $\GRS$ and collection $\UU$ of $\ell$ vertex-disjoint $k$-sequences, the augmentation graph $A = \augmentgraph(\GRS,*,*,\UU)$ has $8n-4r+\ell$ vertices.
\end{observation}
\begin{proof}
	$H = \encodedrs(\GRS,*)$ has $4n$ vertices by~\Cref{obs:encoded-rs}; the set $P$ has $4 \cdot (n-r)$ vertices, and $Q$ has $\ell$ vertices. 
\end{proof}
In the following lemmas, we establish the key properties of augmentation graphs that we need. 

\begin{lemma}\label{lem:augmentation-matching}
	For any $(2n)$-vertex $(r,t)$-RS graph $\GRS$ and  $A = \augmentgraph(\GRS,X,j,\UU)$, 
	there is a matching $M^*$ of size $4n-2r$ in $A$ that does \underline{not} match any of the augmentation vertices in $\aug{A}$. 
\end{lemma}
\begin{proof}
	We construct the matching $M^*$ as follows:
	\begin{itemize}
		\item Add all edges in the perfect matching between $P$ and vertices of $H$ \emph{not} matched by $\rep{\MRS_j}$ to $M^*$; these amount to $4(n-r)$ edges in total. 
		\item Start with the induced matching $\rep{\MRS_j}$ in $H$; for every $\vec{u}_i$ in $\UU$, the edge of the perfect matching between $Q$ and $\start{AP(\vec{u}_i)}$, as well as the remainder of the path $AP(\vec{u}_i)$ form an \emph{alternating} path for $\rep{\MRS_j}$ from $Q$ to $\eend{AP(\vec{u}_i)}$ (because length of $AP(\vec{u}_i)$ is always odd and we added one more edge to it). Add the edges obtained \emph{after} applying these alternating paths\footnote{Given that these 
		paths are ``alternating'' and not ``augmenting'' at this point, our choice of the word ``augmenting paths'' in their definition may sound unnatural; however, in the final construction, which includes further addition to the graph, these paths indeed will 
		become augmenting paths, hence the term (one can think of vertices in $\aug{A}$ as ready to be matched ``outside'').} on $\rep{\MRS_j}$ to $M^*$; these amount to $2r$ edges in total. 
	\end{itemize} 
	It is straightforward to verify that $M^*$ is indeed a matching with size $4n-2r$ since vertices of $\rep{\MRS_j} \cup Q$ are disjoint from the vertices matched in the first part. Moreover, when we apply each alternating path consisting of the $Q$-edge and $AP(\vec{u}_i)$, 
	the last vertex of the path, namely, $\eend{AP(\vec{u}_i)}$ becomes unmatched in $M^*$ as desired.   
\end{proof}

\begin{lemma}\label{lem:augmentation-vc}
	For any $(2n)$-vertex $(r,t)$-RS graph $\GRS$ and $A = \augmentgraph(\GRS,X,j,\UU)$, 
	there is a vertex cover $V^*$ of size $4n-2r$ in $A$ that includes all vertices in $\baug{A}$ and does \underline{not} include any vertex in $\aug{A}$. 	
\end{lemma}
\begin{proof}
	We construct the vertex cover $V^*$ as follows:
	\begin{itemize}
		\item Add all vertices of $H$ \emph{not} matched by $\rep{\MRS_j}$ to $V^*$; these amount to $4(n-r)$ vertices.  
		\item For any augmenting path $AP(\vec{u}_i)$, starting from $\start{AP(\vec{u}_i)}$, add every other alternating vertex on the path to $V^*$. Add the remaining vertices in $R(\rep{\MRS_j})$ that were not part of  augmenting paths to $V^*$. 
		These vertices amount to $2r$ in total. 
	\end{itemize} 
	
	We first argue that $V^*$ is  a vertex cover. The vertices added in the first part cover all edges except the ones with both endpoints in $\rep{\MRS_j} \cup Q$. Adding  $\start{AP(\vec{u}_i)}$ for $\vec{u}_i \in \UU$ also takes care of all edges incident on $Q$. 
	Picking alternating vertices on the paths cover the edges of the paths. The only remaining edges are the ones in $\rep{\MRS_j}$ and augmenting edges $AE$ that are not part of augmenting paths. They will all be covered by the inclusion of the very last set of vertices in $R(\rep{\MRS_j})$ that are not in augmenting paths. 
	Thus, $V^*$ is  a vertex cover with size $4n-2r$. 
	
	Furthermore, the alternating way of picking vertices in $AP(\vec{u}_i)$ plus the fact that length of these paths are odd, means that $\eend{AP(\vec{u}_i)}$ would not be part of $V^*$. This ensures that $V^*$ does not include any vertex from $\aug{A}$. 
	Finally, since  in the last step, we are picking vertices of $R(\rep{\MRS_j})$ that are not in augmenting paths, we will be picking vertices in $\baug{A}$ in $V^*$. This concludes the proof. 
\end{proof}

By duality of matching and vertex cover (\Cref{fact:vc-matching}),~\Cref{lem:augmentation-matching,lem:augmentation-vc} in particular imply that $M^*$ and $V^*$ are maximum matching and minimum vertex cover in $A$ (although we will not use this observation directly and work with the stronger statements in the lemmas). 

\subsubsection*{A Distribution over Augmentation Graphs} We define the following distribution over augmentation graphs. 

\begin{definition}[\textbf{Distribution $\distaug$}]\label{def:dist-aug}
Fix an $(r,t)$-RS graph $\GRS$, integer $k \geq 1$, and vector $Y \in \set{0,1}^{\ell}$ for some $\ell$ such that $k \cdot \ell < r$. We define $\distaug = \distaug(\GRS,Y,k)$ as the following distribution on augmentation graphs $A = \augmentgraph(\GRS,X,j,\UU)$ where $\UU$ consists of $\ell$ vertex-disjoint $k$-sequences:
\begin{enumerate}
	\item Sample index $j \in [t]$ uniformly at random;
	\item Sample matrix $X$ and collection $\UU$ uniformly at random from all pairs such that: 
	\begin{enumerate}
		\item if $Y_i = 0$, then the vertex $\eend{AP(\vec{u}_i)} \in \aug{A}$; 
		\item otherwise, if $Y_i = 1$, then the vertex $\eend{AP(\vec{u}_i)} \in \baug{A}$. 
	\end{enumerate}
	\item[] (recall that by~\Cref{obs:encoded-rs}, the choice of $\eend{AP(\vec{u}_i)}$ is only a function of $X$ and $\UU$ after we conditioned on the choice of $j \in [t]$). 
\end{enumerate}
\end{definition}

We  list some simple observations about this distribution. 

\begin{observation}\label{obs:dist-aug}
	In $\distaug = \distaug(\GRS,Y,k)$ for graphs $A = \augmentgraph(\GRS,X,j,\UU)$: 
	\begin{enumerate}[label=$(\roman*)$]
		\item The choice of $j$ and $X$ are independent (consequently, $j$ and $H=\encodedrs(\GRS,X)$ are also independent); 
		\item Conditioned on the choice of $j$ and $\UU$, the set $\aug{A} \sqcup \baug{A}$ is already fixed -- the partition between the two sets is then solely determined by $X$.
	\end{enumerate}
\end{observation}
The proofs are immediate and we omit them here. 

\subsection{The \textbf{Hidden-Matching} Game}\label{sec:game}

We are finally ready to present our communication game. This is a two player communication game between Alice and Bob, called \game (and follows the same rules described in~\Cref{sec:cc} unless specified otherwise).  
\game goes in two \emph{phases} that loosely correspond to the two passes of streaming algorithms. We start with the parameters and input-independent parts. 

\paragraph{Parameters.} Let $\delta \in (0,1)$ be a fixed constant and $k,n_1 \geq 1$ be integers (think of $n_1$ as \emph{governing} the size of the final graph, and $k$ as a  constant or a slow growing function (doubly-logarithmic) in the size of the graph). Consider a fixed $(2n_1)$-vertex $(r_1,t_1)$-RS graph $\GRS_1$ and another $(2n_2)$-vertex $(r_2,t_2)$-RS graph $\GRS_2$ where $r_2 = (k+\delta) \cdot n_1$. 
We shall emphasize that the parameters $r_1,t_1$ and their connection to $n_1$, as well as $t_2,n_2$ and their connection to $r_2$ are governed by the maximum density of RS graphs we would be able to use in this construction (see~\Cref{cor:stream-lower-RS}). 

These parameters and graphs are known to both players. 

\paragraph{Phase I.} The first phase goes as follows (see~\Cref{fig:phase1} for an illustration of this phase): 

\smallskip

\begin{figure}[h!]
\centering
\input{figs/fig-phase1}
\caption{An illustration of a the \textbf{first phase} of the \game game. The vertices with zero edges given to Alice and Bob are omitted from this figure. The middle graph (red edges) is given to Alice and the outer graphs (blue edges) are given to Bob as input in the first phase.}
\label{fig:phase1}
\end{figure}
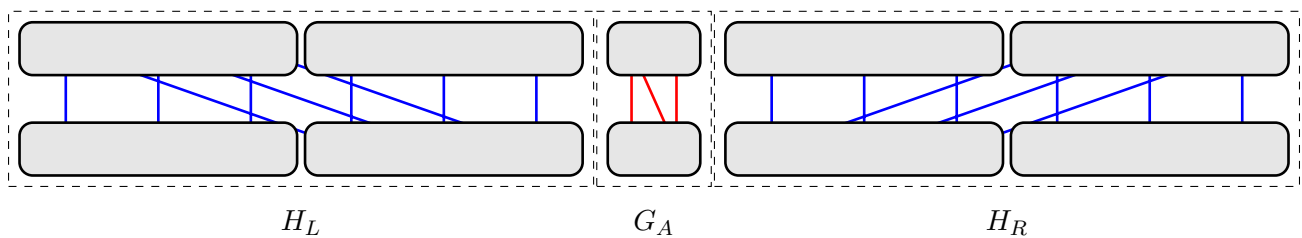

	\begin{itemize}
		\item Initially, Alice receives a copy of $\GRS_1$ such that each edge is removed independently with probability $\delta$. We refer to this graph as $G_A$. 
		\item We sample $j_1 \in [t_1]$ uniformly at random and let:
		\begin{itemize}
			\item $Y_L$ to be the characteristic vector of $L(\MRS_{j_1})$ in $\GRS_1$;
			\item  $Y_R$ to be the characteristic vector of $R(\MRS_{j_1})$ in $\GRS_1$. 
		\end{itemize}
		\item We sample two independent augmentation graphs: 
		\begin{itemize}
		\item $A_L \sim \distaug(\GRS_2,Y_L,k)$ such that $A_L = \augmentgraph(\GRS_2,X_L,j_{L},\UU_L)$, where $A_L$ includes an encoded-RS graph $H_L = \encodedrs(\GRS_2,X_L)$ and remaining edges $\bH_L$; 
		\item $A_R \sim \distaug(\GRS_2,Y_R,k)$ such that $A_R = \augmentgraph(\GRS_2,X_R,j_{R},\UU_R)$, where $A_R$ includes an encoded-RS graph $H_R = \encodedrs(\GRS_2,X_R)$ and remaining edges $\bH_R$. 
		\end{itemize}
		\item Bob receives the edges of $H_L$ and $H_R$ in this phase, referred to as the graph $G_B$. 
		\item At this point, the players run the first round of communication by Alice sending a single message to Bob and Bob responding back with his message. 
	\end{itemize}
This concludes the first phase of the game. Note that at this point, some edges of $A_L$ and $A_R$ have not been given to either player.

\paragraph{Phase II.} We now present the second phase (see~\Cref{fig:phase2} for an illustration of this phase):

\begin{figure}[h!]
\centering
\input{figs/fig-phase2}
\caption{An illustration of a the \textbf{second phase} of the \game game. The solid (green) edges are the ones presented to both players in the second phase. The dashed (red or blue) edges correspond to induced matchings in first phase that play a critical role in the second phase (corresponding to indices $j_1,j_L,j_R$). The remaining 
edges from the first phase are omitted in this figure. The hatched part in $A_L$ and $A_R$ correspond to $\aug{A_L}$ and $\aug{A_R}$, respectively, and are connected to the hidden matching with dashed (red) edges in $G_A$. The double-hatched part in $A_L$
and $A_R$ are vertices of $\MRS_{j_L}$ and $\MRS_{j_R}$ that are \underline{not} incident on augmenting edges. }
\label{fig:phase2}
\end{figure}
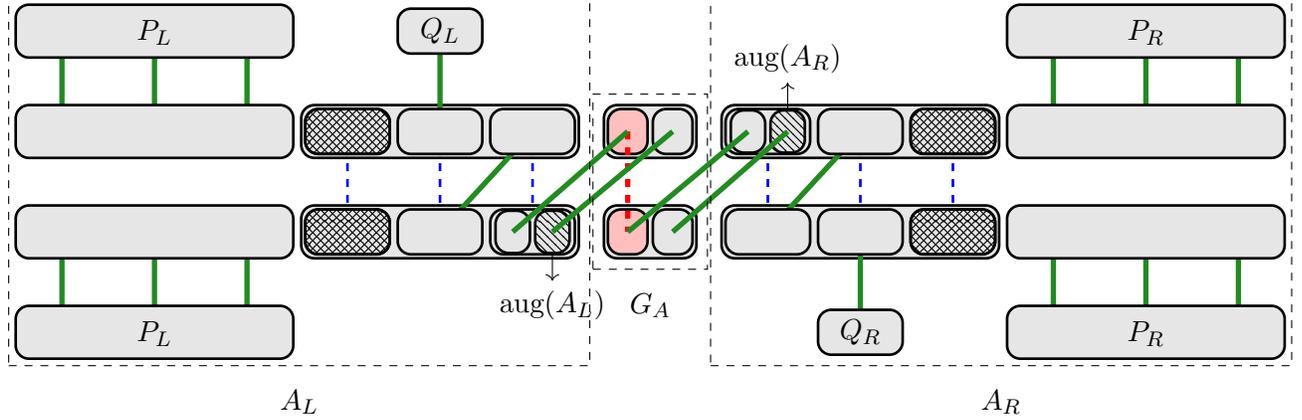

\begin{itemize}
	\item Define the following two matchings between vertices of $G_A$ and $H_L$ and $H_R$, respectively: 
	\begin{itemize}
		\item $M_{L}$: a matching between every $a_i \in \aug{A_L} \cup \baug{A_L}$ and  $v_i \in L(G_A)$;
		\item $M_{R}$: a matching between every  $a_i \in \aug{A_R} \cup \baug{A_R}$ and $v_i \in R(G_A)$.
	\end{itemize}
	\item We give the matchings $M_L$ and $M_R$, as well as edges $\bH_L$ and $\bH_R$ as input to \emph{both} players, denoted by the graph $G_2$. We also reveal the index $j_1$ \emph{but only} to Bob. 
	\item The players run the second round of the protocol by Alice sending a message to Bob, and Bob outputting the following answer defined below. 
	\item  The goal is for Bob to output as many edges as possible from the \textbf{hidden matching} $\MRS_{j_1}$ that appear in the graph $G_A$ of Alice, while outputting no edge that does not belong to $G_A$. 
\end{itemize}

This finalizes the second phase and the overall description of the game.

\begin{observation}\label{obs:game-size}
	For any parameters $(n_1,n_2,r_1,r_2)$ of $\game$, a graph $G$ sampled from $\game$ is a $(2n)$-vertex bipartite graph for $n= 8n_2-4r_2+2n_1$. 
\end{observation}
\begin{proof}
	$A_L$ and $A_R$ each has $8n_2-4r_2+n_1$ vertices by~\Cref{obs:augmentation-graph}, and $G_A$ has $2n_1$ vertices. The bipartition of $G$ into $L$ and $R$
	has equal size, thus the bound follows. 
\end{proof}

Finally, we need the following independence property. 

\begin{observation}\label{obs:game-ind}
	In $\game$, the graphs $G_A$ and $(G_B,G_2)$ are chosen {independently}.  
\end{observation}
\begin{proof}
	The choice of $G_A$ from $\GRS_1$ is independent of all other variables in the game. 
\end{proof}

\subsubsection*{Cost and Value of Protocols for $\game$} 
We conclude with the following  definitions on the performance of protocols for \game. 

\begin{definition}
	Let $\prot$ be a protocol for the $\game$ game. We define: 
	\begin{itemize}
		\item $\cost{\prot}$: the \textbf{communication cost} of $\prot$, which is the \underline{worst-case} number of bits communicated by Alice and Bob in $\prot$ on any input to $\game$.  
		\item $\out{\prot,G}$: the \textbf{output} of $\prot$ on input $G$ sampled from $\game$, which is the set of edges output by Bob that belong to the hidden matching. We will denote $\out{\prot,G} = \emptyset$ if Bob outputs an
		edge that does \underline{not} belong to the input $G_A$ of Alice. 
		\item $\val{\prot}$: The \textbf{value} of $\prot$ is the \underline{expected value} of size of outputs of $\prot$ on inputs sampled from $\game$, i.e., 
		\[
			\val{\prot} := \Exp_{G} \card{\out{\prot,G}}.
		\]
	\end{itemize}
\end{definition}

Our goal in analyzing $\game$ is then to understand the tradeoff between the communication cost and the value obtained by protocols for this game. 

\subsection{Hidden-Matching Game and Streaming Maximum Matching}\label{sec:game-stream} 

We conclude this section by establishing a connection between best possible performance of protocols for $\game$ and the streaming complexity of maximum matching. This will in turn allows us to prove lower bounds for 
streaming matching via lower bounding communication cost of protocols for $\game$. Formally, 

\begin{lemma}\label{lem:game-stream}
	Consider the parameters $(n_1,n_2,r_1,r_2,\delta)$ of $\game$. 
	Suppose there exists a two-pass streaming algorithm with space $s(n)$ on $(2n)$-vertex bipartite graphs for $n= 8n_2-4r_2+2n_1$ that with probability at least $2/3$ achieves a $\paren{1-\beta}$-approximation
	to maximum matching for
	\[
		\beta = \beta(n_1,n_2,r_1,r_2,\delta) = \frac{(1-4\delta) \cdot r_1}{n - (1+2\delta) \cdot r_1}.
	\]
	Then, there is a protocol $\prot$ for $\game$ with:
	\[ 
	\cost{\prot} = O(s(n)) \qquad \text{and} \qquad  \val{\prot} \geq \delta \cdot r_1.
	\] 
\end{lemma}
\begin{proof}
	The proof follows the standard simulation of streaming algorithms via communication protocols. 
	
	Let $\ALG$ be the given streaming algorithm. Define the stream $\sigma = G_A \circ G_B \circ G_2$. We create the protocol $\prot$ as follows: 
\begin{tbox}	
	\underline{The protocol $\prot$ for reduction between $\game$ and streaming maximum matching:} 
	
	\begin{itemize}
	\item \textbf{Phase I:} Alice runs $\ALG$ on $G_A$ and sends  the memory content to Bob; Bob then runs $\ALG$ on $G_B$ and sends back the memory content to Alice. 
	\item \textbf{Phase II:} Alice runs $\ALG$ on $G_2$; at this point, $\ALG$ has made one pass over $\sigma$. Then, Alice runs $\ALG$ again on $G_A$ and sends the memory content to Bob. Bob continues running $\ALG$ on $G_B \circ G_2$, to finish two passes of $\ALG$ on $\sigma$. 
	\item \textbf{Answer:} Let $M$ be the matching returned by $\ALG$ on $\sigma$. Bob outputs all edges of $M$ which belong to $\MRS_{j_1}$ in $\GRS_1$; recall that the index $j_1$ is revealed to Bob in the second phase and $\GRS_1$ is known a-priori.  
	\end{itemize}
\end{tbox}

	It is straightforward to verify that $\prot$ is a valid protocol for $\game$. Moreover, as messages of players in $\prot$ corresponds to the memory content of $\ALG$, we have that $\cost{\prot} = O(s(n))$. We now analyze the value of this protocol in the following two claims. 
	
	\begin{claim}\label{clm:game-perfect}
		W.p. $1-o(1)$, a graph $G \sim \game$ has a matching $M^*$ of size $n-(1+2\delta) \cdot r_1$. 
	\end{claim}
	\begin{proof}
		We create the matching $M^*$ in the following steps: 
		\begin{enumerate}[label=$(\roman*)$]
			\item By~\Cref{lem:augmentation-matching}, there are matchings $M^*_L$ and $M^*_R$ in subgraphs $A_L$ and $A_R$ of $G$, respectively, with size $4n_2 - 2r_2$ each. We add these edges to $M^*$.  
			\item By~\Cref{lem:augmentation-matching}, the matching $M^*$ so far leaves augmentation vertices in $\aug{A_L}$ and $\aug{A_R}$ unmatched. We additionally match these augmentation vertices to $n_1-r_1$ vertices of $L(\GRS_1) \setminus L(\MRS_{j_1})$ and $R(\GRS_1) \setminus R(\MRS_{j_1})$, respectively, using the edges in $M_L$ and $M_R$. 
			\item At this point, the only unmatched vertices in $M^*$ are vertices of $\MRS_{j_1}$;  we can match these to each other using the edges of $\MRS_{j_1}$ that appear in $G_A$. 
			Since each edge is deleted independently with probability $\delta$, and by Chernoff bound ($r_1 \gg \delta^{-1}$ as the latter is constant), w.p. $1-o(1)$, we can match at least $(1-2\delta) \cdot r_1$ edges here as well. This concludes the construction of $M^*$. 
		\end{enumerate}
		The size of $M^*$ is now: 
		\[
			\card{M^*} = 2 \cdot (4n_2 - 2r_2) + 2 \cdot (n_1 - r_1) + (1-2\delta) \cdot r_1 = n - (1+2\delta) \cdot r_1,  
		\]
		w.p. $1-o(1)$, as desired. \Qed{clm:game-perfect} 
		
	\end{proof}
	\begin{claim}\label{clm:game-alg}
		For any graph $G \sim \game$, size of any matching in $G$ that does not use edges of the hidden matching, i.e., size of maximum matching in $G \setminus \MRS_{j_1}$, is at most $n-2r_1$. 
	\end{claim}
	\begin{proof}
		Define $\tG := G \setminus \MRS_{j_1}$. We prove that the minimum vertex cover size in $\tG$ is of size $n-2r_1$, which immediately proves the lemma by the duality of maximum matching and minimum vertex cover (\Cref{fact:vc-matching}). 
		We create this vertex cover $V^*$ (of $\tG$) as follows: 
		\begin{enumerate}[label=$(\roman*)$]
			\item By~\Cref{lem:augmentation-vc}, there are vertex covers $V^*_L$ and $V^*_R$ for subgraphs $A_L$ and $A_R$ of $\tG$, respectively, with size $4n_2 - 2r_2$ each. We add these to $V^*$
			\item Again by~\Cref{lem:augmentation-vc}, the vertex cover $V^*$ currently includes all vertices in $\baug{A_L}$ and $\baug{A_R}$. Recall that these vertices are connected by matchings $M_L$ and $M_R$ to 
			vertices of $L(\MRS_{j_1})$ and $R(\MRS_{j_1})$. Thus, these subset of edges of $M_L$ and $M_R$ are also already covered. 
	
			\item We further add vertices in $L(\GRS_1) \setminus L(\MRS_{j_1})$ and $R(\GRS_1) \setminus R(\MRS_{j_1})$ to $V^*$. This will cover all remaining edges of $M_L$ and $M_R$, as well as any edge in $G_A$ 
			which does not belong to $\MRS_{j_1}$. As such $V^*$ at this point is a vertex cover of $\tG$. 
		\end{enumerate}
		The size of $V^*$ is now: 
		\[
			\card{V^*} = 2 \cdot (4n_2 - 2r_2) + 2 \cdot (n_1 - r_1) = n-2r_1,
		\]
		as desired. \Qed{clm:game-alg}
		
	\end{proof}
	We continue with the proof of~\Cref{lem:game-stream}. Conditioned on the event of~\Cref{clm:game-perfect} and that $\ALG$ outputs a $(1-\beta)$-approximation (which happens w.p. $2/3$), we have that with probability $2/3 - o(1) > 1/2$, 
	$\ALG$ outputs a matching of size at least 
	\[
		(1-\beta) \cdot \paren{n-(1+2\delta) \cdot r_1} = (1-\frac{(1-4\delta) \cdot r_1}{n - (1+2\delta) \cdot r_1}) \cdot \paren{n-(1+2\delta) \cdot r_1} = n-2r_1 + 2\delta \cdot r_1. 
	\]
	Combining this with~\Cref{clm:game-alg}, we have that, with probability at least $1/2$, the matching output by $\ALG$ contains $2\delta \cdot r_1$ edges from $\MRS_{j_1}$ that belong to $G_A$ (and no edge that does not belong to $G$ as $\ALG$ does not err in this case). Given that $\Prot$ will output 
	all these edges in this case, we have, 
	\[
		\val{\Prot} = \Exp_{G} \card{\out{\Prot,G}} \geq \frac{1}{2} \cdot 2\delta \cdot r_1 = \delta \cdot r_1. 
	\]
	This concludes the proof of~\Cref{lem:game-stream}. \Qed{lem:game-stream}
	
\end{proof}

%% file: figs/fig-rs.tex

\begin{tikzpicture}

\tikzset{layer/.style={rectangle, rounded corners=5pt, draw, black, line width=1pt,  fill=black!10, inner sep=4pt}}
\tikzset{vertex/.style={circle, ForestGreen, fill=white, line width=2pt, draw, minimum width=8pt, minimum height=8pt, inner sep=0pt}}
\tikzset{choose/.style={rectangle, line width=1pt, rounded corners = 5pt, draw, minimum height=5pt, fill=black!10}}
		
	\node(V1) {};
	\node[vertex] (a1) [left=-15pt of V1]{};
	\node[vertex] (a2) [right=15pt of a1]{};
	\node[vertex] (a3) [right=15pt of a2]{};
	\node[vertex] (a4) [right=15pt of a3]{};
	\node[vertex] (a5) [right=15pt of a4]{};
	\node[vertex] (a6) [right=15pt of a5]{};
		
	\node (V2) [below = 30pt of V1]{};
	\node[vertex] (b1) [left=-15pt of V2]{};
	\node[vertex] (b2) [right=15pt of b1]{};
	\node[vertex] (b3) [right=15pt of b2]{};
	\node[vertex] (b4) [right=15pt of b3]{};
	\node[vertex] (b5) [right=15pt of b4]{};
	\node[vertex] (b6) [right=15pt of b5]{};
	
	\node[choose,  fit=(a1) (a3)]{};
	\node[choose, fit=(a4) (a6)]{};
	\node[choose, fit=(b1) (b3)]{};
	\node[choose, fit=(b4) (b6)]{};
	
	\draw[line width=0.5pt, black]	
	(a1.center) -- (b1.center)
	(a2.center) -- (b2.center)
	(a3.center) -- (b3.center)
	(a4.center) -- (b4.center)
	(a5.center) -- (b5.center)
	(a6.center) -- (b6.center)
	(a1.center) -- (b4.center)
	(a2.center) -- (b5.center)
	(a3.center) -- (b6.center);	
	
	\node[vertex] (a1) [left=-15pt of V1]{};
	\node[vertex] (a2) [right=15pt of a1]{};
	\node[vertex] (a3) [right=15pt of a2]{};
	\node[vertex] (a4) [right=15pt of a3]{};
	\node[vertex] (a5) [right=15pt of a4]{};
	\node[vertex] (a6) [right=15pt of a5]{};
	
	\node[vertex] (b1) [left=-15pt of V2]{};
	\node[vertex] (b2) [right=15pt of b1]{};
	\node[vertex] (b3) [right=15pt of b2]{};
	\node[vertex] (b4) [right=15pt of b3]{};
	\node[vertex] (b5) [right=15pt of b4]{};
	\node[vertex] (b6) [right=15pt of b5]{};

\end{tikzpicture}

%% file: figs/fig-encoded-rs.tex

\begin{tikzpicture}

\tikzset{layer/.style={rectangle, rounded corners=5pt, draw, black, line width=1pt,  fill=black!10, inner sep=4pt}}
\tikzset{vertex/.style={circle, ForestGreen, fill=white, line width=1.5pt, draw, minimum width=6pt, minimum height=6pt, inner sep=0pt}}
\tikzset{vertex1/.style={circle, blue!70, fill=white, line width=1.5pt, draw, minimum width=6pt, minimum height=6pt, inner sep=0pt}}
\tikzset{vertex2/.style={circle, red!70, fill=white, line width=1.5pt, draw, minimum width=6pt, minimum height=6pt, inner sep=0pt}}

\tikzset{choose/.style={rectangle, line width=1pt, rounded corners = 5pt, draw, fill=black!10, minimum height=17pt}}
		
	\node(V1) {};
	\node[vertex1] (a1) [left=-20pt of V1]{};
	\node[vertex1] (a2) [right=20pt of a1]{};
	\node[vertex1] (a3) [right=20pt of a2]{};
	\node[vertex1] (a4) [right=30pt of a3]{};
	\node[vertex1] (a5) [right=20pt of a4]{};
	\node[vertex1] (a6) [right=20pt of a5]{};
	
	\node[vertex2] (c1) [right=3pt of a1]{};
	\node[vertex2] (c2) [right=20pt of c1]{};
	\node[vertex2] (c3) [right=20pt of c2]{};
	\node[vertex2] (c4) [right=30pt of c3]{};
	\node[vertex2] (c5) [right=20pt of c4]{};
	\node[vertex2] (c6) [right=20pt of c5]{};
		
	\node (V2) [below = 30pt of V1]{};
	\node[vertex1] (b1) [left=-20pt of V2]{};
	\node[vertex1] (b2) [right=20pt of b1]{};
	\node[vertex1] (b3) [right=20pt of b2]{};
	\node[vertex1] (b4) [right=30pt of b3]{};
	\node[vertex1] (b5) [right=20pt of b4]{};
	\node[vertex1] (b6) [right=20pt of b5]{};
	
	\node[vertex2] (d1) [right=3pt of b1]{};
	\node[vertex2] (d2) [right=20pt of d1]{};
	\node[vertex2] (d3) [right=20pt of d2]{};
	\node[vertex2] (d4) [right=30pt of d3]{};
	\node[vertex2] (d5) [right=20pt of d4]{};
	\node[vertex2] (d6) [right=20pt of d5]{};
	
	\node[choose, fit=(a1) (c3)]{};
	\node[choose, fit=(a4) (c6)]{};
	\node[choose, fit=(b1) (d3)]{};
	\node[choose, fit=(b4) (d6)]{};

	\draw[line width=0.5pt, black]	
	(a1.center) -- (d1.center)
	(c1.center) -- (b1.center)
	
	(a2.center) -- (b2.center)
	(c2.center) -- (d2.center)
	
	(a3.center) -- (d3.center)
	(c3.center) -- (b3.center)
	
	(a4.center) -- (d4.center)
	(c4.center) -- (b4.center)
	
	(a5.center) -- (d5.center)
	(c5.center) -- (b5.center)
	
	(a6.center) -- (b6.center)
	(c6.center) -- (d6.center)
	
	(a1.center) -- (b4.center)
	(c1.center) -- (d4.center)
	
	(a2.center) -- (b5.center)
	(c2.center) -- (d5.center)
	
	(a3.center) -- (b6.center)
	(c3.center) -- (d6.center);	
	
	\node[vertex1] (a1) [left=-20pt of V1]{};
	\node[vertex1] (a2) [right=20pt of a1]{};
	\node[vertex1] (a3) [right=20pt of a2]{};
	\node[vertex1] (a4) [right=30pt of a3]{};
	\node[vertex1] (a5) [right=20pt of a4]{};
	\node[vertex1] (a6) [right=20pt of a5]{};
	
	\node[vertex2] (c1) [right=3pt of a1]{};
	\node[vertex2] (c2) [right=20pt of c1]{};
	\node[vertex2] (c3) [right=20pt of c2]{};
	\node[vertex2] (c4) [right=30pt of c3]{};
	\node[vertex2] (c5) [right=20pt of c4]{};
	\node[vertex2] (c6) [right=20pt of c5]{};
	
	\node[vertex1] (b1) [left=-20pt of V2]{};
	\node[vertex1] (b2) [right=20pt of b1]{};
	\node[vertex1] (b3) [right=20pt of b2]{};
	\node[vertex1] (b4) [right=30pt of b3]{};
	\node[vertex1] (b5) [right=20pt of b4]{};
	\node[vertex1] (b6) [right=20pt of b5]{};
	
	\node[vertex2] (d1) [right=3pt of b1]{};
	\node[vertex2] (d2) [right=20pt of d1]{};
	\node[vertex2] (d3) [right=20pt of d2]{};
	\node[vertex2] (d4) [right=30pt of d3]{};
	\node[vertex2] (d5) [right=20pt of d4]{};
	\node[vertex2] (d6) [right=20pt of d5]{};

\end{tikzpicture}

%% file: figs/fig-encoded-rs-2.tex

\begin{tikzpicture}

\tikzset{layer/.style={rectangle, rounded corners=5pt, draw, black, line width=1pt,  fill=black!10, inner sep=4pt}}
\tikzset{vertex/.style={circle, ForestGreen, fill=white, line width=1.5pt, draw, minimum width=6pt, minimum height=6pt, inner sep=0pt}}
\tikzset{vertex1/.style={circle, blue!70, fill=white, line width=1.5pt, draw, minimum width=6pt, minimum height=6pt, inner sep=0pt}}
\tikzset{vertex2/.style={circle, red!70, fill=white, line width=1.5pt, draw, minimum width=6pt, minimum height=6pt, inner sep=0pt}}

\tikzset{choose/.style={rectangle, line width=1pt, rounded corners = 5pt, draw, fill=black!10, minimum height=17pt}}
		
	\node(V1) {};
	\node[vertex1] (a1) [left=-30pt of V1]{};
	\node[vertex1] (a2) [right=30pt of a1]{};
	\node[vertex1] (a3) [right=30pt of a2]{};
	\node[vertex1] (a4) [right=50pt of a3]{};
	\node[vertex1] (a5) [right=30pt of a4]{};
	\node[vertex1] (a6) [right=30pt of a5]{};
	
	\node [above right=5pt and -5pt of a1] {$u_1$};
	\node [above right=5pt and -5pt of a2] {$u_2$};
	\node [above right=5pt and -5pt of a3] {$u_3$};
	
	\node[vertex2] (c1) [right=5pt of a1]{};
	\node[vertex2] (c2) [right=30pt of c1]{};
	\node[vertex2] (c3) [right=30pt of c2]{};
	\node[vertex2] (c4) [right=50pt of c3]{};
	\node[vertex2] (c5) [right=30pt of c4]{};
	\node[vertex2] (c6) [right=30pt of c5]{};
		
	\node (V2) [below = 40pt of V1]{};
	\node[vertex1] (b1) [left=-30pt of V2]{};
	\node[vertex1] (b2) [right=30pt of b1]{};
	\node[vertex1] (b3) [right=30pt of b2]{};
	\node[vertex1] (b4) [right=50pt of b3]{};
	\node[vertex1] (b5) [right=30pt of b4]{};
	\node[vertex1] (b6) [right=30pt of b5]{};
	
	\node[vertex2] (d1) [right=5pt of b1]{};
	\node[vertex2] (d2) [right=30pt of d1]{};
	\node[vertex2] (d3) [right=30pt of d2]{};
	\node[vertex2] (d4) [right=50pt of d3]{};
	\node[vertex2] (d5) [right=30pt of d4]{};
	\node[vertex2] (d6) [right=30pt of d5]{};
	
	\node[choose, fit=(a1) (c3)]{};
	\node[choose, fit=(a4) (c6)]{};
	\node[choose, fit=(b1) (d3)]{};
	\node[choose, fit=(b4) (d6)]{};

	\draw[line width=0.5pt, black!50]	
	
	(a4.center) -- (d4.center)
	(c4.center) -- (b4.center)
	
	(a5.center) -- (d5.center)
	(c5.center) -- (b5.center)
	
	(a6.center) -- (b6.center)
	(c6.center) -- (d6.center)
	
	(a1.center) -- (b4.center)
	(c1.center) -- (d4.center)
	
	(a2.center) -- (b5.center)
	(c2.center) -- (d5.center)
	
	(a3.center) -- (b6.center)
	(c3.center) -- (d6.center);	
	
	\draw[line width=2pt, red]	
	(a1.center) -- (d1.center)
	(c1.center) -- (b1.center)
	
	(a2.center) -- (b2.center)
	(c2.center) -- (d2.center)
	
	(a3.center) -- (d3.center)
	(c3.center) -- (b3.center);

\draw[line width=2pt, ForestGreen, dashed]	
	(b1.center) -- (a2.center)
	(d1.center) -- (c2.center)
	
	(b2.center) -- (a3.center)
	(d2.center) -- (c3.center);
	
	\node(V1) {};
	\node[vertex1] (a1) [left=-30pt of V1]{};
	\node[vertex1] (a2) [right=30pt of a1]{};
	\node[vertex1] (a3) [right=30pt of a2]{};
	\node[vertex1] (a4) [right=50pt of a3]{};
	\node[vertex1] (a5) [right=30pt of a4]{};
	\node[vertex1] (a6) [right=30pt of a5]{};
	
	\node[vertex2] (c1) [right=5pt of a1]{};
	\node[vertex2] (c2) [right=30pt of c1]{};
	\node[vertex2] (c3) [right=30pt of c2]{};
	\node[vertex2] (c4) [right=50pt of c3]{};
	\node[vertex2] (c5) [right=30pt of c4]{};
	\node[vertex2] (c6) [right=30pt of c5]{};
		
	\node (V2) [below = 40pt of V1]{};
	\node[vertex1] (b1) [left=-30pt of V2]{};
	\node[vertex1] (b2) [right=30pt of b1]{};
	\node[vertex1] (b3) [right=30pt of b2]{};
	\node[vertex1] (b4) [right=50pt of b3]{};
	\node[vertex1] (b5) [right=30pt of b4]{};
	\node[vertex1] (b6) [right=30pt of b5]{};
	
	\node[vertex2] (d1) [right=5pt of b1]{};
	\node[vertex2] (d2) [right=30pt of d1]{};
	\node[vertex2] (d3) [right=30pt of d2]{};
	\node[vertex2] (d4) [right=50pt of d3]{};
	\node[vertex2] (d5) [right=30pt of d4]{};
	\node[vertex2] (d6) [right=30pt of d5]{};

\end{tikzpicture}

%% file: figs/fig-aug-graph.tex

\begin{tikzpicture}

\tikzset{layer/.style={rectangle, rounded corners=5pt, draw, black, line width=1pt,  fill=black!10, inner sep=4pt}}
\tikzset{vertex/.style={circle, ForestGreen, fill=white, line width=1.5pt, draw, minimum width=6pt, minimum height=6pt, inner sep=0pt}}
\tikzset{vertex1/.style={circle, blue!70, fill=white, line width=1.5pt, draw, minimum width=6pt, minimum height=6pt, inner sep=0pt}}
\tikzset{vertex2/.style={circle, red!70, fill=white, line width=1.5pt, draw, minimum width=6pt, minimum height=6pt, inner sep=0pt}}
\tikzset{vertex3/.style={circle, ForestGreen, fill=white, line width=1.5pt, draw, minimum width=6pt, minimum height=6pt, inner sep=0pt}}

\tikzset{choose/.style={rectangle, line width=1pt, rounded corners = 5pt, draw, fill=black!10, minimum height=17pt}}
		
\tikzset{group/.style={rectangle, black, fill=white, draw, minimum width=5pt, minimum height=15pt, inner sep=0pt}}
		
\node(V1){};
\node[vertex1](r11) [left=-30pt of V1]{};
\node[vertex2](b11) [right=1.5pt of r11]{};

\foreach \j in {2,...,6}
{
	\pgfmathtruncatemacro{\jp}{\j-1};
	\node[vertex1] (r1\j) [right=3pt of b1\jp]{};
	\node[vertex2] (b1\j) [right=1.5pt of r1\j]{};
}
\begin{scope}[on background layer]
\node[choose, fit=(r11) (b16)] (U1) {};
\end{scope}

\node[vertex1](r21) [right=20pt of b16]{};
\node[vertex2](b21) [right=1.5pt of r21]{};

\foreach \j in {2,...,6}
{
	\pgfmathtruncatemacro{\jp}{\j-1};
	\node[vertex1] (r2\j) [right=3pt of b2\jp]{};
	\node[vertex2] (b2\j) [right=1.5pt of r2\j]{};
}
\begin{scope}[on background layer]
\node[choose, fit=(r21) (b26)] (U2) {};
\end{scope}

\node[vertex1](r31) [below=20pt of r11]{};
\node[vertex2](b31) [right=1.5pt of r31]{};

\foreach \j in {2,...,6}
{
	\pgfmathtruncatemacro{\jp}{\j-1};
	\node[vertex1] (r3\j) [right=3pt of b3\jp]{};
	\node[vertex2] (b3\j) [right=1.5pt of r3\j]{};
}
\begin{scope}[on background layer]
\node[choose, fit=(r31) (b36)] (U3) {};
\end{scope}

\node[vertex1](r41) [right=20pt of b36]{};
\node[vertex2](b41) [right=1.5pt of r41]{};

\foreach \j in {2,...,6}
{
	\pgfmathtruncatemacro{\jp}{\j-1};
	\node[vertex1] (r4\j) [right=3pt of b4\jp]{};
	\node[vertex2] (b4\j) [right=1.5pt of r4\j]{};
}
\begin{scope}[on background layer]
\node[choose, fit=(r41) (b46)](U4) {};
\end{scope}

\node[rectangle, draw, dashed, ForestGreen, line width=1pt, inner sep=1pt, fit=(r11) (b11)] {};
\node[rectangle, draw, dashed, ForestGreen, line width=1pt, inner sep=1pt, fit=(r13) (b13)] {};

\begin{scope}[on background layer]
\draw[line width=2pt, ForestGreen, dashed]	
	(r31.center) -- (r12.center)
	(b31.center) -- (b12.center)
	
	(r33.center) -- (r14.center)
	(b33.center) -- (b14.center);
\end{scope}

\node[vertex3] (q1) [above=20pt of r11]{};
\node[vertex3] (q2) [above=20pt of r13]{};

\node[vertex3](x21) [above=20pt of r21]{};
\node[vertex3](y21) [right=1.5pt of x21]{};

\foreach \j in {2,...,6}
{
	\pgfmathtruncatemacro{\jp}{\j-1};
	\node[vertex3] (x2\j) [right=3pt of y2\jp]{};
	\node[vertex3] (y2\j) [right=1.5pt of x2\j]{};
}

\node[vertex3](x41) [below=20pt of r41]{};
\node[vertex3](y41) [right=1.5pt of x41]{};

\foreach \j in {2,...,6}
{
	\pgfmathtruncatemacro{\jp}{\j-1};
	\node[vertex3] (x4\j) [right=3pt of y4\jp]{};
	\node[vertex3] (y4\j) [right=1.5pt of x4\j]{};
}

\foreach \i in {2,4}
{
\foreach \j in {1,...,6}
{
\begin{scope}[on background layer]
\draw[line width=1pt, ForestGreen]	
	(x\i\j) -- (r\i\j)
	(y\i\j) -- (b\i\j);
\end{scope}
}
}

\begin{scope}[on background layer]
\draw[line width=1pt, ForestGreen]	
	(q1) -- (r11)
	(q2) -- (r13); 
\end{scope}

\node[rectangle, dashed, draw, fit=(U1) (U4)](H){};
\node [right = 5pt of H]{$H$};

\node[rectangle, dashed, draw, fit=(q1) (q2)](Q){};
\node [left = 5pt of Q]{$Q$};

\node[rectangle, dashed, draw, fit=(x21) (y26)](P1){};
\node [left = 5pt of P1]{$P$};

\node[rectangle, dashed, draw, fit=(x41) (y46)](P2){};
\node [left = 5pt of P2]{$P$};

\end{tikzpicture}

%% file: figs/fig-phase1.tex

\begin{tikzpicture}

\tikzset{layer/.style={rectangle, rounded corners=5pt, draw, black, line width=1pt,  fill=black!10, inner sep=4pt}}
\tikzset{vertex/.style={circle, ForestGreen, fill=white, line width=1.5pt, draw, minimum width=6pt, minimum height=6pt, inner sep=0pt}}
\tikzset{vertex1/.style={circle, blue!70, fill=white, line width=1.5pt, draw, minimum width=6pt, minimum height=6pt, inner sep=0pt}}
\tikzset{vertex2/.style={circle, red!70, fill=white, line width=1.5pt, draw, minimum width=6pt, minimum height=6pt, inner sep=0pt}}
\tikzset{vertex3/.style={circle, ForestGreen, fill=white, line width=1.5pt, draw, minimum width=6pt, minimum height=6pt, inner sep=0pt}}

\tikzset{choose/.style={rectangle, line width=1pt, rounded corners = 5pt, draw, fill=black!10, minimum width=32pt, minimum height=17pt}}
		
\tikzset{group/.style={rectangle, black, fill=white, draw, minimum width=5pt, minimum height=15pt, inner sep=0pt}}

\node[choose, minimum width=15pt] (X1){};
\node[choose, minimum width=15pt] (X2) [right=1pt of X1] {};

\node[choose, minimum width=15pt] (Y1) [below=20pt of X1]{};
\node[choose, minimum width=15pt] (Y2) [right=1pt of Y1] {};

\node[choose, inner sep=1pt, fit=(X1) (X2)](L1){};
\node[choose, inner sep=1pt, fit=(Y1) (Y2)](R1){};

\node[choose](A1)[left=10pt of L1]{};
\node[choose](A2)[left=2pt of A1]{};
\node[choose](A3)[left=2pt of A2]{};

\node[choose](A4)[left=5pt of A3]{};
\node[choose](A5)[left=2pt of A4]{};
\node[choose](A6)[left=2pt of A5]{};		
	
\node[choose](B1)[right=10pt of L1]{};
\node[choose](B2)[right=2pt of B1]{};
\node[choose](B3)[right=2pt of B2]{};

\node[choose](B4)[right=5pt of B3]{};
\node[choose](B5)[right=2pt of B4]{};
\node[choose](B6)[right=2pt of B5]{};

\foreach \j in {1,...,6}
{
\node[choose](C\j)[below=20pt of A\j]{};
\node[choose](D\j)[below=20pt of B\j]{};
}

\foreach \i in {A,B,C,D}
{
\node[choose, fit=(\i1) (\i3), inner sep=1pt](\i10){};
\node[choose, fit=(\i4) (\i6), inner sep=1pt](\i20){};
}

\begin{scope}[on background layer]
\draw[red, line width=1pt]
	(X1.center) -- (Y2.center)
	(X1.center) -- (Y1.center)
	(X2.center) -- (Y2.center);
	
\draw[blue, line width=1pt]	
	(A1.center) -- (C1.center)
	(A2.center) -- (C2.center)
	(A3.center) -- (C3.center)
	
	(A4.center) -- (C4.center)
	(A5.center) -- (C5.center)
	(A6.center) -- (C6.center)
	
	(A4.center) -- (C1.center)
	(A5.center) -- (C2.center)
	(A6.center) -- (C3.center)
	
	(B1.center) -- (D1.center)
	(B2.center) -- (D2.center)
	(B3.center) -- (D3.center)
	
	(B4.center) -- (D4.center)
	(B5.center) -- (D5.center)
	(B6.center) -- (D6.center)
	
	(B4.center) -- (D1.center)
	(B5.center) -- (D2.center)
	(B6.center) -- (D3.center);
\end{scope}

\node[rectangle, dashed, draw, fit=(A10) (C20)](HL){};
\node[below=5 pt of HL]{$H_L$};

\node[rectangle, dashed, draw, fit=(B10) (D20)](HR){};
\node[below=5 pt of HR]{$H_R$};

\node[rectangle, dashed, draw, fit=(L1) (R1)](GA){};
\node[below=5 pt of GA]{$G_A$};

\end{tikzpicture}

%% file: figs/fig-phase2.tex

\begin{tikzpicture}

\tikzset{choose/.style={rectangle, line width=1pt, rounded corners = 5pt, draw, fill=black!10, minimum width=32pt, minimum height=17pt}}
		
\tikzset{group/.style={rectangle, black, fill=white, draw, minimum width=5pt, minimum height=15pt, inner sep=0pt}}

\node[choose, minimum width=15pt, fill=red!25] (X1){};
\node[choose, minimum width=15pt] (X2) [right=1pt of X1] {};

\node[choose, minimum width=15pt, fill=red!25] (Y1) [below=20pt of X1]{};
\node[choose, minimum width=15pt] (Y2) [right=1pt of Y1] {};

\begin{scope}[on background layer]
\node[choose, inner sep=1pt, fit=(X1) (X2)](L1){};
\node[choose, inner sep=1pt, fit=(Y1) (Y2)](R1){};
\end{scope}

\node[choose](A1)[left=10pt of L1]{};
\node[choose](A2)[left=2pt of A1]{};
\node[choose](A3)[left=2pt of A2]{};

\node[choose](A4)[left=5pt of A3]{};
\node[choose](A5)[left=2pt of A4]{};
\node[choose](A6)[left=2pt of A5]{};		
	
\node[choose](B1)[right=10pt of L1]{};
\node[choose](B2)[right=2pt of B1]{};
\node[choose](B3)[right=2pt of B2]{};

\node[choose](B4)[right=5pt of B3]{};
\node[choose](B5)[right=2pt of B4]{};
\node[choose](B6)[right=2pt of B5]{};

\foreach \j in {1,...,6}
{
\node[choose](C\j)[below=20pt of A\j]{};
\node[choose](D\j)[below=20pt of B\j]{};
}

\begin{scope}[on background layer]
\foreach \i in {A,B,C,D}
{
\node[choose, fit=(\i1) (\i3), inner sep=1pt](\i10){};
}
\end{scope}

\foreach \i in {A,B,C,D}
{
\node[choose, fit=(\i4) (\i6), inner sep=1pt](\i20){};
}

\draw[red, dashed, line width=2pt]
	(X1.center) -- (Y1.center);
		
\begin{scope}[on background layer]

\draw[blue, dashed, line width=1pt]	
	(A1.center) -- (C1.center)
	(A2.center) -- (C2.center)
	(A3.center) -- (C3.center)
	(B1.center) -- (D1.center)
	(B2.center) -- (D2.center)
	(B3.center) -- (D3.center);
\end{scope}

\node[choose](QL)[above=20pt of A2]{$Q_L$};
\node[choose](QR)[below=20pt of D2]{$Q_R$};

\foreach \j in {4,5,6}
{
	\node[choose](PA\j)[above=20pt of A\j]{};
	\node[choose](PC\j)[below=20pt of C\j]{};
	\node[choose](PB\j)[above=20pt of B\j]{};
	\node[choose](PD\j)[below=20pt of D\j]{};	
}

\foreach \j in {A,C}
{
\node[choose, fit=(P\j4) (P\j6), inner sep=1pt](PX\j){};
\node[below=-10pt of PX\j, anchor= center]{$P_L$};
}

\foreach \j in {B,D}
{
\node[choose, fit=(P\j4) (P\j6), inner sep=1pt](PY\j){};
\node[below=-10pt of PY\j, anchor= center]{$P_R$};
}

\begin{scope}[on background layer]
\foreach \i in {A,B,C,D}
{
	\foreach \j in {4,5,6}
	{
		\draw[ForestGreen, line width=2pt]
			(P\i\j.center) -- (\i\j.center);
	}
}

\draw[ForestGreen, line width=2pt]
	(QL.center) -- (A2.center)
	(QR.center) -- (D2.center)
	(C2.center) -- (A1.center)
	(D1.center) -- (B2.center);
\end{scope}

\node[choose, minimum width=13pt, minimum height=16pt] (ZC1)[left=-16pt of C1]{};
\node[choose, minimum width=13pt, minimum height=16pt, pattern=north west lines] (ZC2)[right=1pt of ZC1]{};

\node[below=10 pt of ZC2](augL){$\aug{A_L}$};
\draw[->] (ZC2) to (augL);

\node[choose, minimum width=13pt, minimum height=16pt] (ZB1)[left=-16pt of B1]{};
\node[choose, minimum width=13pt, minimum height=16pt,  pattern=north west lines] (ZB2)[right=1pt of ZB1]{};

\node[above=10 pt of ZB2](augR){$\aug{A_R}$};
\draw[->] (ZB2) to (augR);

\node[choose, pattern=crosshatch](A3)[left=2pt of A2]{};
\node[choose, pattern=crosshatch](C3)[below=20pt of A3]{};

\node[choose, pattern=crosshatch](B3)[right=2pt of B2]{};
\node[choose, pattern=crosshatch](D3)[below=20pt of B3]{};

\draw[ForestGreen, line width=2pt]
	(ZC1.center) -- (X1.center)
	(ZC2.center) -- (X2.center)
	(ZB1.center) -- (Y1.center)
	(ZB2.center) -- (Y2.center);

\node[rectangle, dashed, draw, fit=(PA4) (PC6) (A10)](HL){};
\node[below=5 pt of HL]{$A_L$};

\node[rectangle, dashed, draw, fit=(PB4) (PD6) (B10)](HR){};
\node[below=5 pt of HR]{$A_R$};

\node[rectangle, dashed, draw, fit=(L1) (R1)](GA){};
\node[below=5 pt of GA]{$G_A$};

\end{tikzpicture}

%% file: lower.tex

\section{The Lower Bound for the \textbf{Hidden-Matching} Game}\label{sec:lower}

We prove our main lower bound for the $\game$ in this section. 

\begin{theorem}\label{thm:game}
	Any protocol $\prot$ (deterministic or randomized) for $\game$ with 
	\[
	\cost{\prot} = \min\set{o(t_2 \cdot r_2^{1-2/k}), o(t_1 \cdot r_1)},  
	\]
	can only have $\val{\prot} = o(r_1)$. 
\end{theorem}

As a direct corollary of this theorem and~\Cref{lem:game-stream}, we obtain the following result for semi-streaming maximum matching problem. 

\begin{corollary}\label{cor:stream-lower-RS}
	Suppose that for infinitely many choices of $N \geq 1$, 
	there exists  $(r,t)$-RS $(2N)$-vertex bipartite graphs such that $r = \alpha \cdot N$ and $t=N^{\beta}$ for some parameters $\alpha$ and $\beta$. The parameters $\alpha$ and $\beta$ can depend on $N$ and we only 
	assume that $\alpha = \Omega(1/\log{N})$ and $\beta = \Omega(1/\log\log{N})$ \footnote{Given that there is already an RS graph construction with $\alpha = 1/2 - o(1)$ and $\beta = \Omega(1/\log\log{N})$ by~\cite{GoelKK12}, this assumption is without loss of generality -- this assumption is only made to simplify the calculations in the proof and in general is not needed.}.
	
	Then, any two-pass semi-streaming algorithm for the maximum matching problem that outputs a correct answer with probability at least $2/3$ cannot achieve 
	an approximation factor better than 
	\[
		1- \frac{\alpha}{\frac{16}{\alpha \cdot \beta} -  \frac{8}{\beta} + 2 - \alpha} \cdot (1-o(1)). 
	\]
\end{corollary}
\begin{proof}
	The proof of this corollary is simply by calculating the values of various parameters in $\game$ and then applying~\Cref{lem:game-stream} to get the semi-streaming lower bound. For our proof, 
	we use $\game$ with the following parameters: 
		\begin{align*}
			&\delta = \text{a vanishingly small {constant}}; \\ 
			&k = \frac{2}{(1-\delta) \cdot \beta}; \\
			&n_1 = N, \quad r_1 = \alpha \cdot N, \quad t_1 = N^{\beta}; \\
			&n_2 = (k+\delta) \cdot \frac{N}{\alpha}, \quad r_2 = (k+\delta) N, \quad t_2 = ((k+\delta) \cdot \frac{N}{\alpha})^{\beta}; \\
			&n = 8n_2-4r_2+2n_1 = \Theta(k \cdot N/\alpha); 
		\end{align*}
	By the promise of corollary statement on the existence of the prescribed RS graph, and using this graph family as $\GRS_1$ and $\GRS_2$, one can verify that the above parameters match those the construction of $\game$. 
	
	By~\Cref{thm:game}, any protocol $\prot$ with 
	\[
		\cost{\prot} = n \cdot \poly\!\log{(n)} \ll n^{1+\Omega(1/\log\log{n})} \ll o(N^{\beta} \cdot (\alpha \cdot N)^{1-2/k}) = o(t_1 \cdot r_1^{1-2/k}) 
	\]
	will have $\val{\prot} = o(r_1)$ (note that  $t_1 \leq t_2$ and $r_1 \leq r_2$ in the parameters above and so $\cost{\prot}$ is even smaller than the min-term in~\Cref{thm:game}). Plugging in this bound in~\Cref{lem:game-stream}, 
	implies that the best approximation ratio achievable by any semi-streaming algorithm will be 
	\[
		1-\frac{(1-4\delta) \cdot r_1}{n - (1+2\delta) \cdot r_1} = 1- \frac{r_1}{n-r_1} \cdot (1-o(1)) = 1-\frac{\alpha}{\frac{16}{\alpha \cdot \beta} - \frac{8}{\beta} + 2 - \alpha} \cdot (1-o(1)).
	\]
	This concludes the proof. 
\end{proof}

\paragraph{Implications of~\Cref{cor:stream-lower-RS}.} Before we move on, let us instantiate this lower bound for different choices of $\alpha$ and $\beta$, given the state-of-the-art on density of RS graphs. 

\begin{itemize}
	\item The \textbf{current best construction} of RS graphs with $\alpha = \Omega(1)$, allows for setting $\beta = \Omega(\frac{1}{\log\log{N}})$. \emph{In case this construction turns out to be the best possible}, 
	then the lower bound in~\Cref{cor:stream-lower-RS} would be $\paren{1-(\frac{\theta}{\log\log{n}})}$-approximation for some absolute constant $\theta \in (0,1)$. 
	\item The \textbf{current best upper bound} on the density of RS graphs for $\alpha = 1/2-o(1)$  forces $\beta$ to be at most $1-\Omega(\frac{\log\log{N}}{\log{N}})$. \emph{In case this upper bound turns out to be  the best possible},
	then the lower bound in~\Cref{cor:stream-lower-RS} would be (at least) $0.98$-approximation. 
	\item In general, \emph{in case there is any RS graph with both $\alpha,\beta$ being a constant}, then the lower bound in~\Cref{cor:stream-lower-RS} would be $(1-\Omega(1))$-approximation. This can be seen as either: 
	\begin{itemize}
		\item a \emph{conditional} lower bound
	that rules out small-constant approximation algorithms for matching under the plausible hypothesis that both $\alpha,\beta$ can be constant; or alternatively, 
	\item a \emph{barrier} result showing that getting (sufficiently) small-constant approximation algorithms to matching, requires (at the very least) improving the current best bounds on density of RS graphs  from $O(N^2/\log{N})$ (for $\alpha=1/2-o(1)$ case) and $O(N^2/2^{O(\log^*{N})})$ (for arbitrary constant $\alpha$)  all the way to $N^{1+o(1)}$ edges.  
	\end{itemize}
\end{itemize}

\input{lower-1}

\input{lower-2}

\input{lower-3}

\input{lower-4}

%% file: lower-1.tex


\subsection{Setup and Notation} \label{sec:lower-1}

In the following, we fix a choice of parameters $(k,n_1,n_2,r_1,r_2,\delta)$ for the $\game$. Let $\prot$ be any protocol for \game with $\cost{\prot}$ as in~\Cref{thm:game}. 
Since $\game$ is a distributional game, we can assume without loss of generality by the easy direction of Yao's minimax principle that $\prot$ is deterministic. 
We shall upper bound $\val{\prot}$ in our proof. 

We will use the following notation: 
\begin{itemize}[leftmargin=10pt]
	\item $\Prot_{A1}$: the message of Alice in phase one -- $\Prot_{A1}$ is a deterministic function of $G_A$; 
	\item $\Prot_{B1}$: the message of Bob in phase one -- $\Prot_{B1}$ is a deterministic function of $\Prot_{A1}$ and $G_B$; 
	\item $Z_A = (\Prot_{A1},\Prot_{B1},G_2)$: the \emph{extra} information known to Alice in the second phase; 
	\item $\Prot_{A2}$: the message of Alice in phase two -- $\Prot_{A2}$ is a deterministic function of $G_A$ and $Z_A$;
	\item $Z_B = (Z_A, \Prot_{A2}, j_1)$: the \emph{extra} information known to Bob in the second phase; 
	\item $\Prot_{B2}$: the output edges of Bob in phase two -- $\Prot_{B2}$ is a deterministic function of $G_B$ and $Z_B$. 
\end{itemize}
(Recall that sans serif fonts will refer to random variables for above parameters, e.g., $\rProt_{A1}$ denote the random variable for $\Prot_{A1}$).

\subsubsection*{Conditional Independence Properties}
 We establish the following conditional independence properties between the random variables above that will be crucial for our proofs. 

\begin{claim}[``Alice's second message does not depend on the hidden matching'']\label{clm:A2-j1-ZA}
	\[\rProt_{A2} \perp \rJ_1 \mid \rZ_A.\] 
\end{claim}
\begin{proof} We have, 
	\begin{align*}
		\mi{\rProt_{A2}}{\rJ_1 \mid \rZ_A} &= \mi{\rProt_{A2}}{\rJ_1 \mid \rProt_{A1},\rProt_{B1}, \rG_2} \tag{by definition of $Z_A$} \\
		&\leq \mi{\rProt_{A2}}{\rG_B \mid \rProt_{A1},\rProt_{B1}, \rG_2} \tag{by data processing inequality~(\itfacts{data-processing}), as augmentation graphs in $\rG_B \cup \rG_2$ determines $\rJ_1$} \\
		&\leq \mi{\rG_A}{\rG_B \mid \rProt_{A1},\rProt_{B1}, \rG_2} \tag{by data processing inequality~(\itfacts{data-processing}), as $\rG_A \cup \rZ_A$ determines $\rProt_{A2}$} \\
		&\leq \mi{\rG_A}{\rG_B \mid \rProt_{A1},\rG_2} \tag{by~\Cref{prop:info-decrease} as $\rG_A \perp \rProt_{B1} \mid \rG_B,\rProt_{A1},\rG_2$ as $\rProt_{B1}$ is fixed by $\rProt_{A1}$ and $\rG_B$} \\
		&\leq \mi{\rG_A}{\rG_B \mid \rG_2} \tag{by~\Cref{prop:info-decrease} as $\rG_B \perp \rProt_{A1} \mid \rG_A,\rG_2$ as $\rProt_{A1}$ is fixed by $\rG_A$} \\
		&= 0 \tag{by~\itfacts{info-zero} as $\rG_A \perp \rG_B \mid \rG_2$ by~\Cref{obs:game-ind}}. 
	\end{align*}
	The claim now follows from~\itfacts{info-zero}. 
\end{proof}

\begin{claim}[``Alice's second message does not correlate her input with the hidden matching'']\label{clm:GA-j1-ZA-ProtA2}
	\[\rG_A \perp \rJ_1 \mid \rZ_A,\rProt_{A2}.\] 
\end{claim}
\begin{proof} We have, 
\begin{align*}
	\mi{\rG_A}{\rJ_1 \mid \rZ_A,\rProt_{A2}} &= \mi{\rG_A}{\rJ_1 \mid \rProt_{A1},\rProt_{B1}, \rG_2, \rProt_{A2}} \tag{by definition of $Z_A$} \\
	&\leq \mi{\rG_A}{\rG_B \mid \rProt_{A1},\rProt_{B1}, \rG_2,\rProt_{A2}} \tag{by data processing inequality~(\itfacts{data-processing}), as augmentation graphs in $\rG_B \cup \rG_2$ determines $\rJ_1$} \\
	&\leq \mi{\rG_A}{\rG_B \mid \rProt_{A1},\rProt_{B1}, \rG_2} \tag{by~\Cref{prop:info-decrease} as $\rG_B \perp \rProt_{A2} \mid \rG_A,\rProt_{A1},\rProt_{B1},\rG_2$ as $\rProt_{A2}$ is fixed by $\rG_A$,$\rProt_{A1},\rProt_{B1}$ and $\rG_2$} \\
	&= 0 \tag{as shown in the proof of~\Cref{clm:A2-j1-ZA}}. 
\end{align*}
	The claim now follows from~\itfacts{info-zero}. 
\end{proof}

\begin{claim}[``Bob output is independent of hidden matching edges conditioned on his knowledge\footnote{We emphasize that Bob's output is certainly \emph{not} independent of hidden matching edges -- after all, Bob is outputting edges of this matching. However, Bob on his own does not know the hidden matching edges and thus is only ``conveying'' his knowledge (from Alice) about these edges to the output; thus, once we condition on Bob's knowledge, we can expect his output to become independent of the hidden matching edges in particular.}'']\label{clm:GB-GAcapMRSj1-ZB}
	\[\rProt_{B2} \perp \rG_A \cap \MRS_{\rJ_1} \mid \rZ_B.\]  
\end{claim}
\begin{proof} We have, 
	\begin{align*}
	\mi{\rProt_{B2}}{\rG_A \cap \MRS_{\rJ_1} \mid \rZ_B} &\leq  \mi{\rG_B}{\rG_A \cap \MRS_{\rJ_1} \mid \rZ_B} \tag{by data processing inequality~(\itfacts{data-processing}) as $\rProt_{B2}$ is a deterministic function of $\rG_B,\rZ_B$} \\
	 &=\mi{\rG_B}{\rG_A \cap \MRS_{\rJ_1} \mid \rProt_{A1},\rProt_{B1}, \rG_2,\rProt_{A2}, \rJ_1} \tag{by definition of $Z_B$} \\
	&\leq \mi{\rG_B}{\rG_A  \mid \rProt_{A1},\rProt_{B1},  \rG_2,\rProt_{A2},\rJ_1}  \tag{by data processing inequality~(\itfacts{data-processing}), as $\rG_A$ determines $\rG_A \cap \MRS_{\rJ_1}$ conditioned on $\rJ_1$} \\
	&\leq \mi{\rG_B}{\rG_A  \mid \rProt_{A1},\rProt_{B1},  \rG_2,\rProt_{A2}} \tag{by~\Cref{prop:info-decrease} as $\rG_A \perp \rJ_1 \mid \rG_B,\rProt_{A1},\rProt_{B1}, \rG_2,\rProt_{A2}$ as $\rJ_1$ is fixed by $\rG_B \cup \rG_2$} \\
	&= 0 \tag{as shown in the proof of~\Cref{clm:GA-j1-ZA-ProtA2}}. 
	\end{align*}
	The claim now follows from~\itfacts{info-zero}. 
\end{proof}

%% file: lower-2.tex

\subsection{Communication in Phase One} \label{sec:lower-2}

In this section, we focus on the communication happening in the first phase and its effect on the ``knowledge'' of players given their second-phase inputs. In particular, the following lemma  establishes that given the communication happening in the first phase
plus the second-phase input \emph{to Alice}, the distribution of the hidden matching does not alter too much for her\footnote{Note that as opposed to Alice, for Bob, the second-phase input directly reveals the identity of the hidden matching.}.

\begin{lemma}[``Alice still does not know hidden matching in the second phase'']\label{lem:phase-one}
	\[
	\Exp_{Z_A \sim \rZ_A}\tvd{\distribution{\rJ_1 \mid Z_A}}{\distribution{\rJ_1}} = o(1). 
	\]
\end{lemma}
In words,~\Cref{lem:phase-one} states that even at the beginning of the second phase of the game, from the perspective of Alice, the distribution of the hidden matching index, namely, $j_1$, has effectively not changed 
from its original (uniform) distribution over $[t_1]$.  

The strategy for the proof of~\Cref{lem:phase-one} is as follows: we show that the messages communicated in the first phase (and in particular Bob's message) are not able 
to change the distribution of $X_L$ or $X_R$ enough for the distribution of $Y_L$ or $Y_R$, and subsequently $j_1$, to change sufficiently either. The heart of this proof is the following lemma that establishes a key property of augmentation graphs in \emph{hiding} the partitioning of augmentation vertices into $\aug{A} \sqcup \baug{A}$ (for $A \in \set{A_L,A_R}$). 

\begin{lemma}[``Augmentation graphs hide the partitioning into $\aug{\cdot}$ and $\baug{\cdot}$'']\label{lem:xor-lemma}
Let $\GRS$ be an $(r,t)$-RS graph and integers $k,\ell$ be such that $k \cdot \ell = (1-\delta) \cdot r$ for some absolute constant $\delta \in (0,1)$. 

Suppose we sample $X$ uniformly  from $\set{0,1}^{r \times t}$ and let $H := \encodedrs(\GRS,X)$. 
Additionally, for any vector $Y \in \set{0,1}^{\ell}$, we let $\bar{H}(Y)$ be a graph obtained by sampling an augmentation graph $A$ from $\distaug(\GRS,Y,k)$
conditioned on encoded-RS graph of $A$ being $H$, and then letting $\bar{H}(Y)$ to be the edges $\bar{H}$ of $A$ outside the encoded-RS graph $H$.  

Consider any  function $\phi :\supp{\rH} \rightarrow \set{0,1}^s$ that encodes each graph $H$ to a message $\phi(H)$ of size $s = o(\delta \cdot t \cdot (\delta r)^{1-2/k})$. 
Then, w.p. $1-o(1)$ over the choice of encoding $\phi \sim \phi(\rH)$ for $\rH$ chosen as above, the following event $\event(\phi)$ happens: 

\begin{itemize}[leftmargin=10pt]
		\item \textbf{Event $\event(\phi)$:} For any pairs of vectors $Y_1 \neq Y_2 \in \set{0,1}^{\ell}$, 
		\[
			\tvd{\distribution{\bar{\rH}(Y_1) \mid \phi(\rH)=\phi}}{\distribution{\bar{\rH}(Y_2) \mid \phi(\rH)=\phi}} = o(1). 
		\] 
\end{itemize}
\end{lemma}

Let us parse~\Cref{lem:xor-lemma}: Suppose we are given an encoded-RS graph $H$ from $\distaug$, and we compress $H$ to a smaller message $\phi(H)$.  Then, ``most'' of the times, namely, when the event $\event(\phi)$ happens, 
these encoded messages have the property that they are ``consistent'' with every possible augmentation graph for any choice of vector $Y \in \set{0,1}^{\ell}$; in other words, the distribution of $\bar{H}$ from  $\distaug$ conditioned on either $\rY = Y_1$ or $ \rY=Y_2$ for any pairs of $Y_1,Y_2 \in \set{0,1}^{\ell}$ is almost identical. Put another way, given $\phi(H)$ and $\bar{H}$ (and conditioned on the event $\event(\phi)$), we will have ``no knowledge'' of partitioning of $\aug{A} \sqcup \baug{A}$ for the underlying
augmentation graph $A$; hence, this partitioning is kept ``hidden'' from us. 

We postpone the proof of~\Cref{lem:xor-lemma} to~\Cref{sec:xor-lemma}. In the following, we show how to use this lemma to conclude the proof of~\Cref{lem:phase-one}. 

\begin{proof}[Proof of~\Cref{lem:phase-one}] 
We  start with expressing the LHS  based on the vectors $Y_L$ and $Y_R$ defined in $\game$. In particular, recall that by construction, the tuple $(\rY_L,\rY_R)$ is a deterministic function of $\rJ_1$ and vice versa. Thus, 
to prove~\Cref{lem:phase-one}, we can alternatively prove the following: 
\begin{align}
\Exp_{Z_A \sim \rZ_A}\tvd{\distribution{\rY_L,\rY_R \mid Z_A}}{\distribution{\rY_L,\rY_R}} = o(1). \label{eq:we-want-to-prove}
\end{align}

 To do this, we need the following claim. 

\begin{claim}\label{clm:Y_L-Y_R}
	With probability $1-o(1)$ over the choice of $(\Prot_{A1},\Prot_{B1}) \sim (\rProt_{A1},\rProt_{B1})$, for any pairs of $(Y_{L1},Y_{R1})$ and $(Y_{L2},Y_{R2})$ in $\supp{(\rY_L,\rY_R)}$, 
	\[
	\tvd{\distribution{\rG_2 \mid \rY_L=Y_{L1},\rY_{R} = Y_{R1}, \Prot_{A1},\Prot_{B1}}}{\distribution{\rG_2 \mid \rY_L=Y_{L2},\rY_{R} = Y_{R2}, \Prot_{A1},\Prot_{B1}}} = o(1). 
	\]
\end{claim}
\begin{proof}
	The proof is a combination of a simple hybrid argument plus a ``reduction'' to the compression setting of~\Cref{lem:xor-lemma}. To simplify the notation, in the following, we use $\Prot_1 = (\Prot_{A1},\Prot_{B1})$ to 
	denote the messages communicated in the first phase, and when clear from the context, we only write $\mid Y_{L*}$ instead of $\mid \rY_L = Y_{L*}$ (similarly for $\rY_R$) to avoid the clutter in the notation. 
	
	Firstly, 
	\begin{align}
		\text{LHS of~\Cref{clm:Y_L-Y_R}} &= \tvd{\distribution{\rG_2 \mid Y_{L1},Y_{R1}, \Prot_{1}}}{\distribution{\rG_2 \mid Y_{L2},Y_{R2}, \Prot_{1}}} \notag \\
		&\leq \tvd{\distribution{\rG_2 \mid Y_{L1},Y_{R1}, \Prot_{1}}}{\distribution{\rG_2 \mid Y_{L2},Y_{R1}, \Prot_{1}}} \notag \\ 
		&\hspace{1cm} + \tvd{\distribution{\rG_2 \mid Y_{L2},Y_{R1}, \Prot_{1}}}{\distribution{\rG_2 \mid Y_{L2},Y_{R2}, \Prot_{1}}}, \label{eq:hybrid-2}
	\end{align}
	by triangle inequality. We note that the two hybrids defined in~\Cref{eq:hybrid-2} are not ``standard'' distributions in the context of $\game$ and will not arise there; we only use them here for the sake of the analysis. In particular, these distributions 
	decorrelate the variables $\rY_L$ and $\rY_R$ from each other (while in $\game$ they are both tied to each other through the choice of $j_1$). 

	We now bound each term in~\Cref{eq:hybrid-2} by $o(1)$. By symmetry, we only focus on the first term; the proof for the second term is identical. 
	Recall that  $G_2 = (\bar{H}_L,\bar{H}_R,M_L,M_R)$. We have, 
\begin{align*}
	&\tvd{\distribution{\rG_2 \mid Y_{L1},Y_{R1}, \Prot_{1}}}{\distribution{\rG_2 \mid Y_{L2},Y_{R1}, \Prot_{1}}} \\
	&\hspace{0.25cm} = 	\tvd{\distribution{\bar{\rH}_L,\bar{\rH}_R,\rM_L,\rM_R \mid Y_{L1},Y_{R1}, \Prot_{1}}}{\distribution{\bar{\rH}_L,\bar{\rH}_R,\rM_L,\rM_R\mid Y_{L2},Y_{R1}, \Prot_{1}}} \\
	&\hspace{.25cm} \leq \tvd{\distribution{\bar{\rH}_L \mid Y_{L1},Y_{R1}, \Prot_{1}}}{\distribution{\bar{\rH}_L \mid Y_{L2},Y_{R1}, \Prot_{1}}} \\
	&\hspace{0.5cm} + \Ex_{\bar{\rH}_L \mid Y_{L1},Y_{R1}, \Prot_{1}} \!\tvd{\distribution{\bar{\rH}_R,\rM_L,\rM_R \mid \bar{H}_L,Y_{L1},Y_{R1}, \Prot_{1}}}{\distribution{\bar{\rH}_R,\rM_L,\rM_R\mid \bar{H}_L,Y_{L2},Y_{R1}, \Prot_{1}}}, \tag{by~\Cref{fact:tvd-chain-rule}} \\
	&\hspace{0.25cm} = \tvd{\distribution{\bar{\rH}_L \mid Y_{L1},\Prot_{1}}}{\distribution{\bar{\rH}_L \mid Y_{L2}, \Prot_{1}}},
\end{align*}
	where the last equality is because of the following: 
	\begin{itemize}
		\item $\bar{\rH}_L \perp Y_{R1} \mid Y_{L*},\Prot_1$: the choice of $\bar{H}_L$ is part of the choice of $A_L$. Conditioned on $Y_{L*}$, the distribution of $A_L$ is $\distaug(\GRS_2,Y_{L*},k)$, independent of all other variables. This continues to be 
		the case even after we condition $\Prot_1$ by the rectangle property of communication protocols as $A_R$ is entirely an input to Bob; this implies that
		\[
			\tvd{\distribution{\bar{\rH}_L \mid Y_{L1},Y_{R1}, \Prot_{1}}}{\distribution{\bar{\rH}_L \mid Y_{L2},Y_{R1}, \Prot_{1}}} = \tvd{\distribution{\bar{\rH}_L \mid Y_{L1},\Prot_{1}}}{\distribution{\bar{\rH}_L \mid Y_{L2}, \Prot_{1}}},
		\]
		for the first term. 
		\item $\bar{\rH}_R \perp Y_{L*} \mid \bar{H}_L, Y_{R1},\Prot_1$: by the same exact argument as above; 
		\item $\bar{\rM}_L,\bar{\rM}_R \perp Y_{L*} \mid \bar{H}_L,\bar{H}_R,\Prot_1$: conditioned on $\bar{\rH}_L$ and $\bar{\rH}_R$, these two matchings are deterministically fixed and thus have the same distribution. This and previous 
		item implies that 
		\[
			 \Ex_{\bar{\rH}_L \mid Y_{L1},Y_{R1}, \Prot_{1}} \!\!\tvd{\distribution{\bar{\rH}_R,\rM_L,\rM_R \mid \bar{H}_L,Y_{L1},Y_{R1}, \Prot_{1}}}{\distribution{\bar{\rH}_R,\rM_L,\rM_R\mid \bar{H}_L,Y_{L2},Y_{R1}, \Prot_{1}}} = 0,
		\]
		for the second term, as the distributions are identical. 
	\end{itemize}
	Consequently, we can conclude that
	\[
		\tvd{\distribution{\rG_2 \mid Y_{L1},Y_{R1}, \Prot_{1}}}{\distribution{\rG_2 \mid Y_{L2},Y_{R1}, \Prot_{1}}} \leq  \tvd{\distribution{\bar{\rH}_L \mid Y_{L1}, \Prot_{1}}}{\distribution{\bar{\rH}_L \mid Y_{L2}, \Prot_{1}}}.
	\]
	Now note that RHS of this equation is exactly what is bounded in~\Cref{lem:xor-lemma}. In particular, we can think of $\Prot_1$ as a compression scheme for $H_L$ as follows: 
	\begin{itemize}
		\item Given an encoded-RS graph $H=\encodedrs(\GRS_2,X)$ in~\Cref{lem:xor-lemma}, we can sample the input of Alice and Bob in the first phase of the game from the distribution of the 
		first hybrid in~\Cref{eq:hybrid-2} conditioned on $\rH = H$ and $\rY_R = Y_{R1}$ always. 
		\item The encoding function $\phi$ then maps $H$ into the message $\Prot_1$ of size $o(\delta \cdot t_2 \cdot (\delta r_2)^{1-2/k})$ by our assumption on $\cost{\prot}$ (recall that $\delta$ is an absolute constant). 
	\end{itemize}
	As such, by~\Cref{lem:xor-lemma}, we have that w.p. $1-o(1)$, for every choice of $Y_{L1},Y_{L2}$, 
	\[
	\tvd{\distribution{\bar{\rH}_L \mid Y_{L1}, \Prot_{1}}}{\distribution{\bar{\rH}_L \mid Y_{L2}, \Prot_{1}}} = o(1). 
	\]
	Repeating the same argument for the second term of~\Cref{eq:hybrid-2} and plugging in the bound in the equation concludes the proof. \Qed{clm:Y_L-Y_R}
	
\end{proof}

We can now conclude the proof of~\Cref{lem:phase-one} as follows. By~\Cref{clm:Y_L-Y_R}, there is an event $\event$ depending of $\Prot_{A1},\Prot_{B1}$ that happens with probability $1-o(1)$, 
and conditioned on $\event$, 
\[
	\tvd{\distribution{\rG_2 \mid \rY_L=Y_{L1},\rY_{R} = Y_{R1}, \Prot_{A1},\Prot_{B1}}}{\distribution{\rG_2 \mid \rY_L=Y_{L2},\rY_{R} = Y_{R2}, \Prot_{A1},\Prot_{B1}}} = o(1),
\]
for every pairs of $(Y_{L1},Y_{R1})$ and $(Y_{L2},Y_{R2})$. We condition on this event in the following. By~\Cref{fact:tvd-sample}, 
this means that given a graph $\rG_2$ from a uniform mixture of $(Y_{L1},Y_{R1})$ and $(Y_{L2},Y_{R2})$ conditioned on $\Prot_{A1},\Prot_{B1}$, the probability that we can detect 
the origin of the sample is at most $\frac{1}{2} + o(1)$. By using this in~\Cref{lem:pair-uniform},  
we obtain that 
\[
	\Exp_{\rG_2 \mid \Prot_{A1},\Prot_{B1}} \Pr\paren{(\rY_L,\rY_R)=(Y_L,Y_R) \mid G_2,\Prot_{A1},\Prot_{B1}} = \frac{1 \pm o(1)}{t_1}. 
\]
Given that the original distribution of $(\rY_L,\rY_R)$ (with no conditioning) is also uniform over its support of size $t_1$, the above implies that conditioned on $\event$, 
\[
	\Exp_{Z_A \sim \rZ_A \mid \event}\tvd{\distribution{\rY_L,\rY_R \mid Z_A}}{\distribution{\rY_L,\rY_R}} = o(1). 
\]
Given that $\event$ itself also happens with probability $1-o(1)$ (and TVD is bounded by $1$), we can conclude the proof of~\Cref{eq:we-want-to-prove}. This finalizes the proof of~\Cref{lem:phase-one}. \Qed{lem:phase-one}

\end{proof}

%% file: lower-3.tex

\subsection{Communication in Phase Two} \label{sec:lower-3}

We now switch to the second phase of the game and show that the message communicated by Alice in the second phase is not that helpful to Bob in identifying the edges of the hidden matching. 
\begin{lemma}[``Bob's does not know the edges of the hidden matching in the second phase'']\label{lem:phase-two}
	\[
	\mi{\rG_A \cap \MRS_{\rJ_1}}{\rZ_B} = o(r_1).
	\] 
\end{lemma}

In words,~\Cref{lem:phase-two} states that the information between second-phase knowledge of Bob, $Z_B$, and the edges of the hidden matching is negligible (compared to the size of the hidden matching). 

\begin{proof}[Proof of~\Cref{lem:phase-two}]
By definition, we can write the mutual information term in the lemma as: 
\begin{align}
	\mi{\rG_A \cap \MRS_{\rJ_1}}{\rZ_B} = \en{\rG_A \cap \MRS_{\rJ_1}} - \en{\rG_A \cap \MRS_{\rJ_1} \mid \rZ_B}. \label{eq:mutual-info-term}
\end{align}
The first term above is simply $H_2(\delta) \cdot r_1$ as each of the $r$ edges of $\MRS_{\rJ_1}$ is dropped independently with probability $\delta$. We can thus focus on bounding the second term. We have, 
\begin{align*}
	\en{\rG_A \cap \MRS_{\rJ_1} \mid \rZ_B} &= \en{\rG_A \cap \MRS_{\rJ_1} \mid \rZ_A, \rProt_{A2},\rJ_1}  \tag{by the definition of $\rZ_B$}\\
	&= \Exp_{\rZ_A, \rProt_{A2},\rJ_1} \bracket{\en{\rG_A \cap \MRS_{j_1} \mid Z_A,\Prot_{A2},j_1}} \tag{by the definition of conditional entropy}\\ 
	&= \Exp_{\rZ_A, \rProt_{A2},\rJ_1} \bracket{\en{\rG_A \cap \MRS_{j_1} \mid Z_A,\Prot_{A2}}} \tag{by~\Cref{clm:GA-j1-ZA-ProtA2}, $\rG_A \cap \MRS_{j_1} \perp \rJ_1=j_1 \mid Z_A,\Prot_{A2}$ and so we can drop the conditioning on $j_1$} \\
	&= \Exp_{\rZ_A,\rProt_{A2}} \Exp_{\rJ_1 \mid Z_A,\Prot_{A2}}  \bracket{\en{\rG_A \cap \MRS_{j_1} \mid Z_A,\Prot_{A2}}} \\
	&= \Exp_{\rZ_A,\rProt_{A2}} \Exp_{\rJ_1 \mid Z_A} \bracket{\en{\rG_A \cap \MRS_{j_1} \mid Z_A,\Prot_{A2}}} \tag{by~\Cref{clm:A2-j1-ZA}, $\rJ_1 \perp \rProt_{A2}=\Prot_{A2} \mid Z_A$ and so we can drop the conditioning on $\Prot_{A2}$} \\
	&\leq \Exp_{\rZ_A,\rProt_{A2}} \Exp_{\rJ_1} \bracket{\en{\rG_A \cap \MRS_{j_1} \mid Z_A,\Prot_{A2}}} + \Exp_{Z_A \sim \rZ_A}\tvd{\distribution{\rJ_1 \mid Z_A}}{\distribution{\rJ_1}} \cdot r_1, 
\end{align*}
where in the last equation, we used~\Cref{fact:tvd-small} to change the distribution of $\rJ$ and ``pay'' the difference in the maximum value of the entropy term. 

By~\Cref{lem:phase-one}, we 
already have that the second term above is $o(r_1)$, so it remains to bound the first term, which is done in the following claim (we emphasize that in the following claim, $\rJ_1$ is chosen independent of $\rZ_A,\rProt_{A2}$ from its original
distribution, which was uniform over $[t_1]$).  

\begin{claim}\label{clm:index-lb}
$\Exp_{\rZ_A,\rProt_{A2}} \Exp_{\rJ_1} \bracket{\en{\rG_A \cap \MRS_{j_1} \mid Z_A,\Prot_{A2}}} = H_2(\delta) \cdot r_1 - o(r_1).$
\end{claim}
\begin{proof}
	Given that the distribution of $\rJ_1$ is uniform over $[t_1]$, and by the definition of $\rZ_A$, we have, 
	\begin{align*}
		\Exp_{\rZ_A,\rProt_{A2}} \Exp_{\rJ_1} \bracket{\en{\rG_A \cap \MRS_{j_1} \mid Z_A,\Prot_{A2}}} &= \frac{1}{t_1} \cdot \sum_{j_1 = 1}^{t_1}  {\en{\rG_A \cap \MRS_{j_1} \mid \rProt_{A1},\rProt_{B1},\rProt_{A2},\rG_2}} \\
		&\geq  \frac{1}{t_1} \cdot \en{\rG_A \mid \rProt_{A1},\rProt_{B1},\rProt_{A2},\rG_2} \tag{by the sub-additivity of entropy (\itfacts{sub-additivity}) as $\rG_A := (\rG_A \cap \MRS_{1},\ldots,\rG_A \cap \MRS_{t_1})$} \\
		&\geq  \frac{1}{t_1} \cdot \paren{\en{\rG_A \mid \rG_2} - \en{\rProt_{A1},\rProt_{B1},\rProt_{A2}}} \tag{by~\itfacts{cond-reduce}} \\
		&=  \frac{1}{t_1} \cdot \paren{\en{\rG_A} - \en{\rProt_{A1},\rProt_{B1},\rProt_{A2}}} \tag{as $\rG_A \perp \rG_2$ by~\Cref{obs:game-ind} and so we can apply~\itfacts{cond-reduce}} \\
		&\geq \frac{1}{t_1} \cdot \paren{t_1 \cdot H_2(\delta) \cdot r_1 - o(t_1 \cdot r_1)} = H_2(\delta) \cdot r_1 - o(r_1),
	\end{align*}
	where in the last inequality, we used the fact that $\rG_A$ consists of $t_1$ induced matchings whose edges are being dropped independently w.p. $\delta$, and that the size of the message communicated by Alice and Bob is at most $o(t_1 \cdot r_1)$. 
	This concludes the proof of the claim. \Qed{clm:index-lb}
	
\end{proof}

We can now complete the proof of~\Cref{lem:phase-two}. By~\Cref{eq:mutual-info-term} and the discussion above it, plus~\Cref{clm:index-lb} and the preceding equations, we have, 
\[
	\mi{\rG_A \cap \MRS_{\rJ_1}}{\rZ_B} = \en{\rG_A \cap \MRS_{\rJ_1}} - \en{\rG_A \cap \MRS_{\rJ_1} \mid \rZ_B} = H_2(\delta) \cdot r_1 - (H_2(\delta) \cdot r_1 - o(r_1) - o(r_1)) = o(r_1),
\]
finalizing the proof.  \Qed{lem:phase-two} 

\end{proof}

%% file: lower-4.tex

\subsection{Concluding the Proof of~\Cref{thm:game}} \label{sec:lower-4}

We are now ready to conclude the proof of~\Cref{thm:game}. For that, we need the following lemma. 

\begin{lemma}[``The protocol's value is small'']\label{lem:lower-4}
	\[
	\val{\prot} = o(r_1).
	\] 
\end{lemma}
\begin{proof}
	Recall that $\Prot_{B2}$ denotes the set of edges from $G_A \cap \MRS_{j_1}$ that is output by Bob at the end of the game; as such, $\val{\prot} = \Exp\card{\rProt_{B2}}$. We focus on upper bounding 
	this expectation term. Firstly, 
	\begin{align}
		\en{\rG_A \cap \MRS_{\rJ_1} \mid \rProt_{B2}} \leq H_2(\delta) \cdot (r_1 - \Exp\card{\rProt_{B2}}), \label{eq:ent-matching}
	\end{align}
	because the edges in $\rProt_{B2}$ can no longer be removed from $\MRS_{j_1}$ when defining the graph $G_A$. 
	
	We will now lower bound the LHS of~\Cref{eq:ent-matching} to finalize the proof. We have, 
	\begin{align*}
		\en{\rG_A \cap \MRS_{\rJ_1} \mid \rProt_{B2}} &\geq \en{\rG_A \cap \MRS_{\rJ_1} \mid \rZ_B,\rProt_{B2}} \tag{as conditioning can only reduce the entropy (\itfacts{cond-reduce})} \\
		&= \en{\rG_A \cap \MRS_{\rJ_1} \mid \rZ_B}, \tag{as $\rG_A \cap \MRS_{\rJ_1} \perp \rProt_{B2} \mid \rZ_B$ by~\Cref{clm:GB-GAcapMRSj1-ZB} so we can apply~\itfacts{cond-reduce}} \\
		&=  \en{\rG_A \cap \MRS_{\rJ_1}} - \mi{\rG_A \cap \MRS_{\rJ_1}}{\rZ_B} \tag{by the definition of mutual information in~\Cref{eq:mi}} \\
		&= H_2(\delta) \cdot r_1 - o(r_1) \tag{by the distribution of $G_A \cap \MRS_{j_1}$ for the first term and~\Cref{lem:phase-two} for the second}.
	\end{align*}
	Plugging in this bound in~\Cref{eq:ent-matching}, we get that, 
	\[
		\val{\prot} = \Exp\card{\rProt_{B2}} = o(r_1),
	\]
	as desired. 
\end{proof}

\noindent
\Cref{thm:game} now follows  from~\Cref{lem:lower-4} and our assumption in~\Cref{sec:lower-1} on $\cost{\prot}$. 

%% file: xor-lemma.tex

\newcommand{\rS}{\rv{S}}
\newcommand{\rz}{\rv{z}}

\section{The Hiding Property of Augmentation Graphs (\Cref{lem:xor-lemma})}\label{sec:xor-lemma}

In this section, we prove~\Cref{lem:xor-lemma} used in~\Cref{sec:lower-1} which was the missing part of the proof of our main lower bound in~\Cref{thm:game}. 


\begin{lemma*}[Re-statement of~\Cref{lem:xor-lemma}]
Let $\GRS$ be an $(r,t)$-RS graph and integers $k,\ell$ be such that $k \cdot \ell = (1-\delta) \cdot r$ for some absolute constant $\delta \in (0,1)$. 

Suppose we sample $X$ uniformly  from $\set{0,1}^{r \times t}$ and let $H := \encodedrs(\GRS,X)$. 
Additionally, for any vector $Y \in \set{0,1}^{\ell}$, we let $\bar{H}(Y)$ be a graph obtained by sampling an augmentation graph $A$ from $\distaug(\GRS,Y,k)$
conditioned on encoded-RS graph of $A$ being $H$, and then letting $\bar{H}(Y)$ to be the edges of $A$ outside the encoded-RS graph $H$.  

Consider any  function $\phi :\supp{\rH} \rightarrow \set{0,1}^s$ that encodes each graph $H$ to a message $\phi(H)$ of size $s = o(\delta \cdot t \cdot (\delta r)^{1-2/k})$. 
Then, w.p. $1-o(1)$ over the choice of encoding $\phi \sim \phi(\rH)$ for $\rH$ chosen as above, the following event $\event(\phi)$ happens: 

\begin{itemize}[leftmargin=10pt]
		\item \textbf{Event $\event(\phi)$:} For any pairs of vectors $Y_1 \neq Y_2 \in \set{0,1}^{\ell}$, 
		\[
			\tvd{\distribution{\bar{\rH}(Y_1) \mid \phi(\rH)=\phi}}{\distribution{\bar{\rH}(Y_2) \mid \phi(\rH)=\phi}} = o(1). 
		\] 
\end{itemize}
\end{lemma*}


We start the proof of~\Cref{lem:xor-lemma} with the following notation. 

\paragraph{Notation.} For any $i \in [r]$ (resp. $j \in [t])$, we use $X_i$ (resp. $X^j$) to denote the $i$-th row (resp. $j$-th column) of $X$; similarly, $X^j_i$ denotes the $(i,j)$-entry of the matrix $X$. To avoid confusion, we use $\jstar \in [t]$ 
to denote the index of the random induced matching $\MRS_{\jstar}$  in $\distaug$. Additionally, for any $\vec{u}_i \in \UU$ (chosen in the augmentation graph), we use $X(\vec{u}_i) := (X^{\jstar}_{i_1},\ldots,X^{\jstar}_{i_k})$
for $(i_1,\ldots,i_k) = \vec{u}_i$ and define $\oplus X(\vec{u}_i) = \oplus_{x \in X(\vec{u}_i)} x$. Similarly, define $\oplus(X(\UU)) = (\oplus X(\vec{u}_1),\ldots, \oplus X(\vec{u}_\ell))$. We have $Y = \oplus X(\UU)$ by~\Cref{obs:aug-path}, thus, our goal is to 
show that $\oplus X(\UU)$ remains hidden. 

Finally, given that there is a one-to-one mapping between $X$ and the encoded-RS graph $H$, to avoid clutter, we slightly abuse the notation
and write $\phi(X)$ instead of $\phi(H)$ and consider $\phi$  as a mapping from $\set{0,1}^{r \times t} \rightarrow \set{0,1}^s$.

We are now ready for the proof. The first step is to show that w.p. $1-o(1)$, the entropy of $X^{\jstar}$ is sufficiently large. Formally, 
\begin{claim}\label{clm:high-entropy-jstar}
	W.p. $1-o(1)$ over the choice of $(\phi,\jstar) \sim (\rPhi,\rJ^{\star})$, we have, 
	\[
		\en{\rX^{\jstar} \mid \rPhi = \phi, \rJ^{\star}=\jstar} =  r-o(\delta \cdot (\delta r)^{1-2/k}). 
	\]
	(We denote this event by $\event(\phi,\jstar)$). 
\end{claim}
\begin{proof}
	The proof is a simple direct-sum style argument  as follows: 
	\begin{align*}
		\mi{\rX^{\rJ^{\star}}}{\rPhi \mid \rJ} &= \frac{1}{t} \cdot \sum_{j=1}^{t} \mi{\rX^{j}}{\rPhi \mid \rJ^{\star}=j} \tag{by the uniform choice of distribution of $\jstar$} \\
		&= \frac{1}{t} \cdot \sum_{j=1}^{t} \mi{\rX^{j}}{\rPhi} \tag{as $(\rX^j,\rPhi) \perp \rJ^{\star}=j$ since $\rX$ is independently uniform and $\rPhi$ is a function of $\rX$} \\
		&\leq  \frac{1}{t} \cdot \sum_{j=1}^{t} \mi{\rX^{j}}{\rPhi \mid \rX^{<j}} \tag{by~\Cref{prop:info-increase} as $\rX^j \perp \rX^{<j}$} \\
		&= \frac{1}{t} \cdot \mi{\rX}{\rPhi} \tag{by chain rule of mutual information (\itfacts{chain-rule})} \\
		&\leq \frac{1}{t} \cdot \en{\rPhi} =  o(\delta \cdot (\delta r)^{1-2/k}). \tag{by~\itfacts{uniform} as the message size is $o(\delta t \cdot (\delta r)^{1-2/k})$ bits} 
	\end{align*}
	By the definition of mutual information, we have, 
	\[
		\en{\rX^{\rJ^{\star}} \mid \rPhi,\rJ} = \en{\rX^{\rJ^{\star}} \mid \rJ^{\star}} - \mi{\rX^{\rJ^{\star}}}{\rPhi \mid \rJ^{\star}} = r - o(\delta \cdot (\delta r)^{1-2/k}), 
	\]
	as $\rX^{\rJ^{\star}} \mid \rJ^{\star}$ is uniformly distributed over $\set{0,1}^{r}$. 
	Given that entropy of $\rX^{\rJ^{\star}}$ can never be more than $r$ (as its support has at most $r$ variables), by Markov bound, 
	we have that w.p. $1-o(1)$ over the choice of $(\phi,\jstar) \sim (\rPhi , \rJ^{\star})$, $\en{\rX^{\jstar} \mid \rPhi = \phi, \rJ^{\star}=\jstar} = r-o(\delta \cdot (\delta r)^{1-2/k})$ as desired. \Qed{clm:high-entropy-jstar}
	
\end{proof}

We now use the bound on the entropy of $X^{\jstar}$ to argue that its distribution is almost a convex combination of a series of near-uniform distributions over ``large'' supports. 

\begin{claim}\label{clm:near-uniform}
	Conditioned on the event $\event(\phi,\jstar)$, we have, 
	\[
		\distribution{\rX^{\jstar} \mid \phi(\rX) = \phi,\rJ^{\star}=\jstar} = \sum_{l=0}^{L} p_l \cdot \mu_l,
	\]
	such that $p_0 = o(1)$ and for every $l \in [L] \setminus \set{0}$: 
	\begin{itemize}
		\item $\log{\card{\supp{\mu_l}}} \geq r - o(\delta \cdot (\delta \cdot r)^{1-2/k})$;
		\item $\tvd{\mu_l}{\uniform_l} = o(1)$, where $\uniform_l$ is the uniform distribution over $\supp{\mu_l}$. 
	\end{itemize}
	(We further use $\event(\phi,\jstar,l)$ for $l \in [L]$ to denote the combined event of $\event(\phi,\jstar)$ and that $\rX$ is chosen from $\mu_l$ in the given convex combination). 
\end{claim}
\begin{proof}
	A direct corollary of~\Cref{lem:aux1} by setting $\gamma = o(\delta \cdot (\delta \cdot r)^{1-2/k})$ and $\eps = o(1)$. 
\end{proof}

\noindent
We can now start bounding the LHS of~\Cref{lem:xor-lemma} whenever $\event(\phi,\jstar,l)$ for some $l \in [L] \setminus \set{0}$ happens. 
Define: 
\begin{itemize}
	\item $\rX^*$ to be a random variable sampled from $\uniform_l$ (so, by~\Cref{clm:near-uniform}, we have that distribution of $\rX^{\jstar}$ is close to that of $\rX^*$ conditioned on $\event(\phi,\jstar,l)$). 
	\item $\rH^*(Y_1)$ and $\rH^*(Y_2)$ to be the same as $\bar{\rH}(Y_1)$ and $\bar{\rH}(Y_2)$, respectively, 
with the difference that to sample them, instead of sampling $\rX^{\jstar}$ from $\mu_l$ and then defining the corresponding augmentation graph, we sample $X^{\jstar}$ from $\rX^*$ and then continue as before. 
\end{itemize}
Given the similarity of distributions $\mu_l$ and $\uniform_l$ by~\Cref{clm:near-uniform}, we can use~\Cref{fact:tvd-small} to get, 
\begin{align}
	&\tvd{\distribution{\bar{\rH}(Y_1) \mid \event(\phi,\jstar,l)}}{\distribution{\bar{\rH}(Y_2) \mid \event(\phi,\jstar,l)}} \notag \\ 
	&\hspace{1cm} \leq \tvd{\distribution{{\rH^*}(Y_1) \mid \event(\phi,\jstar,l)}}{\distribution{{\rH^*}(Y_2) \mid \event(\phi,\jstar,l)}} + o(1) \label{eq:change-distribution}. 
\end{align}
\noindent
Notice that after conditioning on $\jstar$, the only random choice in $\rH^*$ is the choice of $\UU$. 
As such, 
We can now apply a hybrid argument using~\Cref{fact:tvd-chain-rule} as follows: 
\begin{align}
	&\tvd{\distribution{{\rH^*}(Y_1) \mid \event(\phi,\jstar,l)}}{\distribution{{\rH^*}(Y_2) \mid \event(\phi,\jstar,l)}} \notag \\
	&\hspace{0.5cm} = \tvd{\distribution{\UU \mid Y_1,\event(\phi,\jstar,l)}}{\distribution{\UU \mid Y_2, \event(\phi,\jstar,l)}} \notag \\
	&\hspace{0.5cm} \leq \sum_{i=1}^{\ell} \Ex_{\UU^{<i} \mid Y_1,\event(\phi,\jstar,l)} \tvd{\distribution{\vec{\ru}_i \mid \UU^{<i},Y_1,\event(\phi,\jstar,l)}}{\distribution{\vec{\ru}_i \mid \UU^{<i},Y_2,\event(\phi,\jstar,l)}} \tag{by~\Cref{fact:tvd-chain-rule}}\\
	&\hspace{0.5cm} \leq \sum_{i=1}^{\ell} \Ex_{\UU^{<i} \mid Y_1,\event(\phi,\jstar,l)} \Ex_{\vec{\ru}_i \mid \UU^{<i},\event(\phi,\jstar,l)} \notag \\
	&\hspace{0.5cm} \card{\Pr\paren{\oplus\rX^*(\vec{u}_i)=Y_{1,i} \mid \vec{u}_i, \oplus\rX^*(\UU^{<i})=Y_1^{<i}}-\Pr\paren{\oplus\rX^*(\vec{u}_i)=Y_{2,i} \mid \vec{u}_i, \oplus\rX^*(\UU^{<i})=Y_1^{<i}}} \label{eq:hybrid},
\end{align}
where the last inequality is by~\Cref{fact:tvd-sample}. In words, the difference between $\rH^*(Y_1)$ and $\rH^*(Y_2)$ can be bounded by \emph{sum} of the following: 
\begin{itemize}
	\item What is the difference between the probability of $\oplus \rX^*(\vec{u}_i)$ being $0$ or $1$, in expectation over the choice of $\vec{u}_i$, assuming that $\rX^*$ is chosen 
	such that $\oplus \rX^*(\UU^{<i})$ is chosen according to the first $(i-1)$ indices of $Y_1$? 
\end{itemize}
In other words,~\Cref{eq:hybrid} reduces our task of proving the lower bound to bounding the advantage one gets (from the message $\phi$ and other conditioned terms) when focusing only 
on a single term $\vec{u}_i$ even if we condition on the remaining values of $\UU$. This is the content of the following claim. 

\begin{claim}\label{clm:fourier-xor}
	Conditioned on $\event(\phi,\jstar,l)$ for any $l \in [L] \setminus \set{0}$, we have that for every $i \in [\ell]$, 
	\begin{align*}
	&\hspace{5cm} \Ex_{\UU^{<i} \mid Y_1,\event(\phi,\jstar,l)} \quad \Ex_{\vec{\ru}_i \mid \UU^{<i},\event(\phi,\jstar,l)} \\
	&\card{\Pr\paren{\oplus\rX^*(\vec{u}_i)=0 \mid \vec{u}_i, \oplus\rX^*(\UU^{<i})=Y_1^{<i}}-\Pr\paren{\oplus\rX^*(\vec{u}_i)=1 \mid \vec{u}_i, \oplus\rX^*(\UU^{<i})=Y_1^{<i}}} \\
	&\hspace{7cm} = o(1).
	\end{align*}
\end{claim}
\begin{proof}
	We first consider the distribution of $\rX^{*} \mid \oplus\rX^*(\UU^{<i})=Y_1^{<i}$. Each $\vec{u}_j \in \UU^{<i}$ fixes the value of $X^*$ in $k$ coordinates. As such, 
	after this conditioning, and by re-indexing the unfixed coordinates, we can think of $X^*$ as a subset of $\set{0,1}^{r'}$ for $r' = r - (i-1) \cdot k$. Similarly, the choice of $\vec{u}_i$ will also be a uniform $k$-subset of $[r']$ in this case as 
	each $\vec{u}_j \in \UU^{<i}$  ``consumes'' $k$ indices of $[r]$, leaving the indices of $[r']$ untouched. 
	As such, we can use $\rz \in_R \set{0,1}^{r'}$ to denote $X^{*}$ after this transitioning and $\rS=(i_1,\ldots,i_k) \subseteq_R [r']$ to denote $\vec{u}_i$. We now have, 
	\begin{align*}
	&\text{LHS of~\Cref{clm:fourier-xor}} = \Exp_{\rS}\card{\Pr\paren{\oplus_{i \in S} \rz_i=0} - \Pr\paren{\oplus_{i \in S} \rz_i = 1}}.
	\end{align*}
	
	Fourier analysis now gives us a standard tool to bound the RHS of this equation (see~\Cref{sec:fourier} for an overview of basic definitions of Fourier analysis on Boolean hypercube). Define: 
	\begin{itemize}
		\item $\bias(\rz,S) := {\Pr\paren{\oplus_{i \in S} \rz_i = 0} - \Pr\paren{\oplus_{i \in S} \rz_i = 1}}$, as the bias of XOR of $\rz$ on indices of $S$. 
		\item $f :\set{0,1}^{r'} \rightarrow \set{0,1}$ as the characteristic function of $\supp{\rz}$: $f(z) = 1$ iff $z \in \supp{\rz}$;
		\item $\Xor_S: \set{0,1}^{r'} \rightarrow \set{0,1}$ as the {character} function over $[r']$ (see~\Cref{sec:fourier}). 
	\end{itemize}
	
	\noindent
	Firstly, by the definition $\bias(\rz,S)$, we have, 
	\begin{align}
		&\text{LHS of~\Cref{clm:fourier-xor}} = \Exp_\rS {\card{\bias(\rz,S)}}. \label{eq:bias}
	\end{align}
	Secondly, for any $S \subseteq [r']$, we have, 
	\[
		\bias(\rz,S) = \Exp_{z \sim \rz} \bracket{\Xor_S(z)} = \frac{1}{\card{\supp{\rz}}} \cdot \sum_{z \in \set{0,1}^{r'}} f(z) \cdot {\Xor_S(z)} = \frac{2^{r'} \cdot \hf(S)}{\card{\supp{\rz}}},
	\]
	by the definition of Fourier coefficients. We now use the KKL inequality (\Cref{prop:kkl}) to bound the sum of squared Fourier coefficients and then use this in~\Cref{eq:bias}. In particular, for any $\gamma \in (0,1)$, 
	\begin{align*}
		\sum_{S} \bias(\rz,S)^2 &= \frac{2^{2r'}}{\card{\supp{\rz}}^2} \cdot \sum_{S} \hf(S)^2 \tag{by the previous equation} \\
		&\leq \frac{2^{2r'}}{\card{\supp{\rz}}^2} \cdot \gamma^{-k} \cdot \paren{\frac{\supp{\rz}}{2^{r'}}}^{\frac{2}{1+\gamma}} \tag{by KKL inequality in~\Cref{prop:kkl}} \\
		&=  \gamma^{-k} \cdot \paren{\frac{2^{r'}}{\card{\supp{\rz}}}}^{\frac{2\gamma}{1+\gamma}}. 
	\end{align*}
	By picking $\gamma = k \cdot (\log{(2^{r'}/\card{\supp{\rz}})})^{-1}$, we have, 
	\begin{align*}
		\Exp_{\rS} \bracket{\bias(\rz,S)^2} \leq \frac{1}{{r' \choose k} } \cdot \gamma^{-k} \cdot \paren{\frac{2^{r'}}{\card{\supp{\rz}}}}^{{2\gamma}} = \paren{O(\frac{\log{(2^{r'}/\card{\supp{\rz}})}}{r'})}^{k}. 
	\end{align*} 
	By~\Cref{clm:near-uniform}, we have that 
	\[
		\log{(\card{\supp{\rz}})} \geq r' - o(\delta \cdot (\delta r)^{1-2/k}). 
	\]
	Plugging this bound in the previous equation and noting that $r' \geq \delta r$, we get that 
	\[
		\Exp_{S} \bracket{\bias(\rz,S)^2} = o(\frac{1}{r^2}). 
	\]
	We now use this to get that, 
	\[
		\Exp_{\rS}\card{\Pr\paren{\oplus_{i \in S} \rz_i=0} - \Pr\paren{\oplus_{i \in S} \rz_i = 1}} = \Exp_\rS {\card{\bias(\rz,S)}} \leq \sqrt{\Exp_{\rS} \bracket{\bias(\rz,S)^2}} = o(1/r), 
	\]
	where the first inequality is by~\Cref{eq:bias} and the second is Jensen's inequality for $\sqrt{\cdot}$. Given that $r \geq \ell$, this concludes the proof. \Qed{clm:fourier-xor} 
	
\end{proof}

We can now conclude the proof of~\Cref{lem:xor-lemma}. With probability $1-o(1)$, we have the event in~\Cref{clm:high-entropy-jstar}. With another probability $1-o(1)$, by choice of $p_0 = o(1)$ in~\Cref{clm:near-uniform}, 
we can conclude that $\rX$ conditioned on $\phi,\jstar$ and the latter event is chosen from a near-uniform distribution, i.e., from $\mu_l$ for $l \in [L] \setminus \set{0}$. Conditioned on these events, 
by~\Cref{eq:change-distribution} and~\Cref{eq:hybrid} and~\Cref{clm:fourier-xor}, we have that 
\begin{align*}
	&\tvd{\distribution{\bar{\rH}(Y_1) \mid \event(\phi,\jstar,l)}}{\distribution{\bar{\rH}(Y_2) \mid \event(\phi,\jstar,l)}} = o(1). 
\end{align*}
for any choice of $Y_1,Y_2 \in \set{0,1}^{\ell}$. This concludes the proof of~\Cref{lem:xor-lemma}.

%% file: appendix-info.tex

\section{Basic Tools From Information Theory}\label{sec:info}

We now briefly introduce some definitions and facts from information theory that are needed in this paper. We refer the interested reader to the  textbook by Cover and Thomas~\cite{CoverT06} for an excellent introduction to this field. 

For a random variable $\rA$, we use $\supp{\rA}$ to denote the support of $\rA$ and $\distribution{\rA}$ to denote its distribution. 
When it is clear from the context, we may abuse the notation and use $\rA$ directly instead of $\distribution{\rA}$, for example, write 
$A \sim \rA$ to mean $A \sim \distribution{\rA}$, i.e., $A$ is sampled from the distribution of random variable $\rA$.

We denote the \emph{Shannon Entropy} of a random variable $\rA$ by
$\en{\rA}$, which is defined as: 
\begin{align}
	\en{\rA} := \sum_{A \in \supp{\rA}} \Pr\paren{\rA = A} \cdot \log{\paren{1/\Pr\paren{\rA = A}}} \label{eq:entropy}
\end{align} 
\noindent
The \emph{conditional entropy} of $\rA$ conditioned on $\rB$ is denoted by $\en{\rA \mid \rB}$ and defined as:
\begin{align}
\en{\rA \mid \rB} := \Ex_{B \sim \rB} \bracket{\en{\rA \mid \rB = B}}, \label{eq:cond-entropy}
\end{align}
where 
$\en{\rA \mid \rB = B}$ is defined in a standard way by using the distribution of $\rA$ conditioned on the event $\rB = B$ in Eq~(\ref{eq:entropy}).

The \emph{mutual information} of two random variables $\rA$ and $\rB$ is denoted by
$\mi{\rA}{\rB}$ and  defined as:
\begin{align}
\mi{\rA}{\rB} := \en{\rA} - \en{\rA \mid  \rB} = \en{\rB} - \en{\rB \mid  \rA}. \label{eq:mi}
\end{align}
\noindent
The \emph{conditional mutual information} $\mi{\rA}{\rB \mid \rC}$ is $\en{\rA \mid \rC} - \en{\rA \mid \rB,\rC}$ and hence by linearity of expectation:
\begin{align}
	\mi{\rA}{\rB \mid \rC} = \Ex_{C \sim \rC} \bracket{\mi{\rA}{\rB \mid \rC = C}}. \label{eq:cond-mi}
\end{align} 

Finally, we use $H_2$ to denote the \emph{binary Entropy function} where for any real number $ \delta \in (0,1)$, we define: 
\begin{align}
H_2(\delta) = \delta\log{\frac{1}{\delta}} + (1-\delta)\log{\frac{1}{1-\delta}},
\end{align}
i.e., the entropy of a Bernoulli random variable with mean $\delta$.


\subsection{Useful Properties of Entropy and Mutual Information}\label{sec:prop-en-mi}

We use the following basic properties of entropy and mutual information throughout. 

\begin{fact}[cf.~\cite{CoverT06}]\label{fact:it-facts}
  Let $\rA$, $\rB$, $\rC$, and $\rD$ be four (possibly correlated) random variables.
   \begin{enumerate}
  \item \label{part:uniform} $0 \leq \en{\rA} \leq \log{\card{\supp{\rA}}}$. The right equality holds
    iff $\distribution{\rA}$ is uniform.
  \item \label{part:info-zero} $\mi{\rA}{\rB}[\rC] \geq 0$. The equality holds iff $\rA$ and
    $\rB$ are \emph{independent} conditioned on $\rC$.
  \item \label{part:cond-reduce} \emph{Conditioning on a random variable reduces entropy}:
    $\en{\rA \mid \rB,\rC} \leq \en{\rA \mid  \rB}$.  The equality holds iff $\rA \perp \rC \mid \rB$. At the same time, $\en{\rA \mid \rB,\rC} \geq \en{\rA \mid \rB} - \en{\rC}$. 
    \item \label{part:sub-additivity} \emph{Subadditivity of entropy}: $\en{\rA,\rB \mid \rC}
    \leq \en{\rA \mid C} + \en{\rB \mid  \rC}$.
   \item \label{part:ent-chain-rule} \emph{Chain rule for entropy}: $\en{\rA,\rB \mid \rC} = \en{\rA \mid \rC} + \en{\rB \mid \rC,\rA}$.
  \item \label{part:chain-rule} \emph{Chain rule for mutual information}: $\mi{\rA,\rB}{\rC \mid \rD} = \mi{\rA}{\rC \mid \rD} + \mi{\rB}{\rC \mid  \rA,\rD}$.
  \item \label{part:data-processing} \emph{Data processing inequality}: for a deterministic function $f(\rA)$, $\mi{f(\rA)}{\rB \mid \rC} \leq \mi{\rA}{\rB \mid \rC}$. 
   \end{enumerate}
\end{fact}


\noindent

We also use the following two standard propositions, regarding the effect of conditioning on mutual information.

\begin{proposition}\label{prop:info-increase}
  For random variables $\rA, \rB, \rC, \rD$, if $\rA \perp \rD \mid \rC$, then, 
  \[\mi{\rA}{\rB \mid \rC} \leq \mi{\rA}{\rB \mid  \rC,  \rD}.\]
\end{proposition}
 \begin{proof}
  Since $\rA$ and $\rD$ are independent conditioned on $\rC$, by
  \itfacts{cond-reduce}, $\HH(\rA \mid  \rC) = \HH(\rA \mid \rC, \rD)$ and $\HH(\rA \mid  \rC, \rB) \ge \HH(\rA \mid  \rC, \rB, \rD)$.  We have,
	 \begin{align*}
	  \mi{\rA}{\rB \mid  \rC} &= \HH(\rA \mid \rC) - \HH(\rA \mid \rC, \rB) = \HH(\rA \mid  \rC, \rD) - \HH(\rA \mid \rC, \rB) \\
	  &\leq \HH(\rA \mid \rC, \rD) - \HH(\rA \mid \rC, \rB, \rD) = \mi{\rA}{\rB \mid \rC, \rD}. \qed
	\end{align*}
	
\end{proof}

\begin{proposition}\label{prop:info-decrease}
  For random variables $\rA, \rB, \rC,\rD$, if $ \rA \perp \rD \mid \rB,\rC$, then, 
  \[\mi{\rA}{\rB \mid \rC} \geq \mi{\rA}{\rB \mid \rC, \rD}.\]
\end{proposition}
 \begin{proof}
 Since $\rA \perp \rD \mid \rB,\rC$, by \itfacts{cond-reduce}, $\HH(\rA \mid \rB,\rC) = \HH(\rA \mid \rB,\rC,\rD)$. Moreover, since conditioning can only reduce the entropy (again by \itfacts{cond-reduce}), 
  \begin{align*}
 	\mi{\rA}{\rB \mid  \rC} &= \HH(\rA \mid \rC) - \HH(\rA \mid \rB,\rC) \geq \HH(\rA \mid \rD,\rC) - \HH(\rA \mid \rB,\rC) \\
	&= \HH(\rA \mid \rD,\rC) - \HH(\rA \mid \rB,\rC,\rD) = \mi{\rA}{\rB \mid \rC,\rD}.  \qed
 \end{align*}
 
\end{proof}

\subsection{Measures of Distance Between Distributions}\label{sec:prob-distance}

We will use the following two standard measures of distance (or divergence) between distributions. 

\paragraph{KL-divergence.} For two distributions $\mu$ and $\nu$, the \emph{Kullback-Leibler divergence} between $\mu$ and $\nu$ is denoted by $\kl{\mu}{\nu}$ and defined as: 
\begin{align}
\kl{\mu}{\nu}:= \Ex_{a \sim \mu}\Bracket{\log\frac{\Pr_\mu(a)}{\Pr_{\nu}(a)}}. \label{eq:kl}
\end{align}
The following states the relation between mutual information and KL-divergence. 
\begin{fact}\label{fact:kl-info}
	For random variables $\rA,\rB,\rC$, 
	\[\mi{\rA}{\rB \mid \rC} = \Ex_{(b,c) \sim {(\rB,\rC)}}\Bracket{ \kl{\distribution{\rA \mid \rB=b,\rC=c}}{\distribution{\rA \mid \rC=c}}}.\] 
\end{fact}

\paragraph{Total variation distance.} We denote the total variation distance between two distributions $\mu$ and $\nu$ on the same 
support $\Omega$ by $\tvd{\mu}{\nu}$, defined as: 
\begin{align}
\tvd{\mu}{\nu}:= \max_{\Omega' \subseteq \Omega} \paren{\mu(\Omega')-\nu(\Omega')} = \frac{1}{2} \cdot \sum_{x \in \Omega} \card{\mu(x) - \nu(x)}.  \label{eq:tvd}
\end{align}
\noindent
We use the following basic properties of total variation distance. 
\begin{fact}\label{fact:tvd-small}
	Suppose $\mu$ and $\nu$ are two distributions for a random variable $\rX$, then, 
	\[
	\Ex_{\mu}\bracket{\rX} \leq \Ex_{\nu}\bracket{\rX} + \tvd{\mu}{\nu} \cdot \max_{X_0 \in \supp{\rX}} X_0.
	\]
\end{fact}

\begin{fact}\label{fact:tvd-sample}
	Suppose $\mu$ and $\nu$ are two distributions over the same support $\Omega$; then, given one sample $s$ from either $\mu$ or $\nu$, the probability we can decide whether $s$ came from $\mu$ or $\nu$ 
	is $\frac12 + \frac12\cdot\tvd{\mu}{\nu}$; alternatively, 
	\[
		\Exp_{s} \card{\Pr\paren{\mu \mid s} - \Pr\paren{\nu \mid s}} = \tvd{\mu}{\nu}. 
	\]
\end{fact}

\begin{fact}\label{fact:tvd-chain-rule}
	Suppose $\mu$ and $\nu$ are two distributions for the tuple $(\rX_1,\ldots,\rX_t)$; then, 
	\[
		\tvd{\mu(\rX_1,\ldots,\rX_t)}{\nu(\rX_1,\ldots,\rX_t)} \leq \sum_{i=1}^{n} \Exp_{(X_1,\ldots,X_{i-1}) \sim \mu} \tvd{\mu(\rX_i \mid X_1,\ldots,X_{i-1})}{\nu(\rX_i \mid X_1,\ldots,X_{i-1})}.
	\]
\end{fact}

The following Pinsker's inequality bounds the total variation distance between two distributions based on their KL-divergence, 

\begin{fact}[Pinsker's inequality]\label{fact:pinskers}
	For any distributions $\mu$ and $\nu$, 
	$
	\tvd{\mu}{\nu} \leq \sqrt{\frac{1}{2} \cdot \kl{\mu}{\nu}}.
	$ 
\end{fact}

\subsection{Some Auxiliary Lemmas} 

We use the following auxiliary lemmas in our proofs. The first lemma shows one typical way that one can see high entropy random variables as almost-uniform distributions. 
\begin{lemma}[cf.~\cite{AssadiKL17,AAAK17}]\label{lem:aux1}
	Suppose $\rA$ is a random variable with $\en{\rA} \geq \log{\card{\supp{\rA}}} - \gamma$ for some $\gamma \geq 1$.  Then, for any $\eps > \exp(-\gamma)$, we have $\rA = \sum_{l = 0}^{L} p_l \cdot \mu_l$ 
	for $L = O(\gamma/\eps^3)$ such that $p_0 = O(\eps)$ and for every $l \in L \setminus \set{0}$: 
	\begin{itemize}
	\item $\log{\card{\supp{\mu_l}}} \geq \log{\card{\supp{\rA}}} - \gamma/\eps$;
	\item $\tvd{\mu_l}{\uniform_l} = O(\eps)$, where $\uniform_l$ is the uniform distribution on $\supp{\mu_l}$. 
	\end{itemize}
\end{lemma}

\noindent
The next lemma gives a simple of way of showing a distribution is (point wise) close to uniform. 

\begin{lemma}\label{lem:pair-uniform}
	Let $\rX$ be any random variable such that for any pairs $X_1,X_2 \in \supp{\rX}$, 
	\[
		\frac{1-\eps}{2} \leq \Pr\paren{\rX = X_1 \mid \rX \in \set{X_1,X_2}} \leq \frac{1+\eps}{2} 
	\]
	for some $\eps \in (0,1/2)$. Then, for every $X \in \supp{\rX}$, 
	\[
		\frac{1-2\eps}{\card{\supp{\rX}}} \leq \Pr\paren{\rX = X} \leq \frac{1+2\eps}{\card{\supp{\rX}}}
	\]
\end{lemma}
\begin{proof}
	Let $t = \card{\supp{\rX}}$. Suppose there exist some $X \in \supp{\rX}$ such that 
	\[
	p(X):= \Pr\paren{\rX = X} < \frac{1-\eps}{t}.
	\]
	This naturally implies that there exist at least one element $Y \in \supp{\rX}$ such that 
	\[
		p(Y) := \Pr\paren{\rY = Y}  >  \frac{1}{t},
	\]
	as otherwise the total sum of probabilities of atoms in $\supp{\rX}$ will not add up to $1$. We have, 
	\[
		\Pr\paren{\rX = X \mid \rX \in \set{X,Y}} = \frac{p(X)}{p(X)+p(Y)} = \frac{1}{1+\frac{p(Y)}{p(X)}} < \frac{1}{1+\frac{1}{1-2\eps}} = \frac{1-2\eps}{2-2\eps} < \frac{1-\eps}{2},
	\]
	contradicting the assumption in the lemma statement. The other case can be proven symmetrically, finalizing the proof. 
\end{proof}

%% file: appendix-fourier.tex

\newcommand{\inner}[2]{\langle #1 , #2 \rangle} 

\section{Basic Tools from Fourier Analysis on Boolean Hypercube}\label{sec:fourier} 
 
We briefly introduce some definitions and facts from Fourier Analysis on Boolean hypercube that are needed in this paper. We refer the interested reader to the text by de Wolf~\cite{Wolf08} for an excellent introduction to this field. 

For any two functions $f,g: \set{0,1}^{n} \rightarrow \IR$, we define the \emph{inner product} between $f$ and $g$ as: 
\[
	\inner{f}{g} = \Exp_{x \in \set{0,1}^{n}} \bracket{f(x) \cdot g(x)} = \sum_{x \in \set{0,1}^n} \frac{1}{2^n} \cdot f(x) \cdot g(x). 
\]
For a set $S \subseteq \set{0,1}$, we define the \emph{character} function $\Xor_S : \set{0,1}^n \rightarrow \set{-1,+1}$ as: 
\[
	\Xor_S(x) = (-1)^{(\sum_{i \in S} x_i)} = \begin{cases} 1 & \text{if $\oplus_{i \in S}~ x_i = 0$} \\ -1 &  \text{if $\oplus_{i \in S}~ x_i = 1$} \end{cases}. 
\]
The \emph{Fourier transform} of $f: \set{0,1}^{n} \rightarrow \IR$ is a function $\hf : 2^{[n]} \rightarrow \IR$ such that: 
\[
	\hf(S) = \inner{f}{\Xor_S} = \sum_{x \in \set{0,1}^n} \frac{1}{2^n} \cdot f(x) \cdot \Xor_S(x). 
\]
We refer to each $\hf(S)$ as a \emph{Fourier coefficient}. 

\medskip
We use the following \emph{KKL inequality} due to~\cite{KahnKL88}, for bounding sum of \emph{squared} of Fourier coefficients in our proofs. 

\begin{proposition}[\!\cite{KahnKL88}]\label{prop:kkl}
	For every function $f \in \set{0,1}^n \rightarrow \set{-1,0,+1}$ and every $\gamma \in (0,1)$
	\[
		\sum_{S \subseteq [n]} \gamma^{\card{S}} \cdot \hf(S)^2 \leq \paren{\frac{\supp{f}}{2^n}}^{\frac{2}{1+\gamma}}. 
	\]
\end{proposition}